\documentclass[11pt]{article}
\usepackage[margin=0.9in]{geometry}
\usepackage{tikz-cd}
\usepackage[toc,page]{appendix}

\usepackage{complexity}
\usepackage{multirow}
\usepackage{hyperref}
\usepackage{amsthm}
\usepackage{amssymb}
\usepackage{url}
\usepackage{enumerate}
\usepackage[noadjust]{cite}
\usepackage{color}
\usepackage{graphicx}
\usepackage{mathtools}
\usepackage{amsmath}
\usepackage{authblk}
\usepackage{url}
\usepackage[normalem]{ulem}
\usepackage{subcaption}

\def\eq#1{Eq. \!\!\eqref{#1}}
\def\nn{\nonumber\\}

\def\span{\mathrm{span}}
\def\tr{\mathrm{tr}}
\def\set#1{\{#1\}}
\def\proj{\mathrm{proj}}
\def\mat#1{\left(\begin{matrix} #1 \end{matrix}\right)}

\input{Qcircuit}
\def\braket#1#2{\langle #1 | #2 \rangle}
\def\ketbra#1#2{| #1 \rangle \! \langle #2 |}

\def\cC{\mathcal C}

\def\cE{\mathcal E}

\def\cH{\mathcal H}

\def\cL{\mathcal L}

\def\cP{\mathcal P}

\def\bbC{\mathbb C}

\def\bbH{\mathbb H}

\def\bbR{\mathbb R}

\def\bbZ{\mathbb Z}

\theoremstyle{definition}
\newtheorem{definition}{Definition}
\theoremstyle{theorem}
\newtheorem{theorem}[definition]{Theorem}
\newtheorem{lemma}[definition]{Lemma}
\newtheorem{corollary}[definition]{Corollary}
\newtheorem{prop}[definition]{Proposition}
\newtheorem{example}[definition]{Example}
\theoremstyle{definition}
\newtheorem{remark}[definition]{Remark}

\def\Gp{G(\textstyle\frac \pi 4)}

\newcommand\subsetsim{\mathrel{\substack{
  \textstyle\subset\\[-0.2ex]\textstyle\sim}}}

\definecolor{darkgreen}{rgb}{0,0.5,0}

\begin{document}
\title{\bf Quantum simulation from the bottom up: \\
the case of rebits}
\author[1]{Dax Enshan Koh\thanks{daxkoh@mit.edu}}
\author[2]{Murphy Yuezhen Niu\thanks{yzniu@mit.edu}} 
\author[2,3]{Theodore J. Yoder\thanks{tjyoder@mit.edu}}
\affil[1]{\small Department of Mathematics, Massachusetts Institute of Technology, Cambridge, Massachusetts 02139, USA}
\affil[2]{\small Department of Physics, Massachusetts Institute of Technology, Cambridge, Massachusetts 02139, USA}
\affil[3]{\small IBM T.J.~Watson Research Center, Yorktown Heights, NY 10598, USA}
\date{}  

\maketitle
\begin{abstract}
Typically, quantum mechanics is thought of as a linear theory with unitary evolution governed by the Schr\"{o}dinger equation. While this is technically true and useful for a physicist, with regards to \emph{computation} it is an unfortunately narrow point of view. Just as a classical computer can simulate highly nonlinear functions of classical states, so too can the more general quantum computer simulate nonlinear evolutions of quantum states. 
We detail one particular simulation of nonlinearity on a quantum computer, showing how the entire class of \textit{$\mathbb{R}$-unitary} evolutions (on $n$ qubits) can be simulated using a unitary, real-amplitude quantum computer (consisting of $n+1$ qubits in total). These operators can be represented as the sum of a linear and antilinear operator, and add an intriguing new set of nonlinear quantum gates to the toolbox of the quantum algorithm designer. Furthermore, a subgroup of these nonlinear evolutions, called the \textit{$\mathbb{R}$-Cliffords}, can be efficiently classically simulated, by making use of the fact that Clifford operators can simulate non-Clifford (in fact, non-linear) operators. This perspective of using the physical operators that we have to simulate non-physical ones that we do not is what we call bottom-up simulation, and we give some examples of its broader implications.
\end{abstract}

\section{Introduction}

Simulation is a ubiquitous task in the modern world with diverse uses from fundamental research (e.g.~particle physics simulations in the LHC) to entertainment (e.g.~virtual reality headsets). Computers are often associated with simulation due to their wide ranging capabilities as (finite instances of) universal Turing machines. Like classical computers, universal quantum computers are expected to be powerful simulators offering potentially even greater efficiency for some very important quantum tasks, such as chemistry \cite{lanyon2010towards,Aspuru-Guzik1704,babbush2017low} and fermionic simulations \cite{bravyi2002fermionic,bravyi2017tapering}. Moreover, the notion of simulation is not limited to a correspondence between especially disparate systems -- for instance, quantum error-correcting codes can be said to simulate a handful of encoded qubits with many physical ones.

Yet, it is quite common for a simulator $P$ to be developed in terms of end-goals, that is, for the purpose of modeling specific operators $O_L$ on the simulated system $L$ that are deemed interesting. Such a design, which we call top-down, reveals the set of operators $O_P$ on $P$ that are necessary to simulate $O_L$. However, with access to a universal simulator, like a quantum computer, it may be more relevant to start from the bottom with operators we can definitely perform on $P$, and ask what operators on $L$ can be simulated. In contrast to the top-down simulation, this bottom-up simulation by definition takes full advantage of the capability of the simulator.

In this paper, we provide a fleshed-out example of a bottom-up simulation, a nonlinear $n$-qubit quantum computer being simulated by a linear, real-amplitude quantum computer. The simulator consists of $n+1$ \emph{rebits}, which mathematically means that its states are normalized vectors restricted to $\mathbb{R}^{2^{n+1}}$ and it has the ability to perform orthogonal linear operators. Although the top-down version of this simulation has been noted many times in the past beginning with Bennett et.~al.~\cite{bennett1997strengths}, the bottom-up viewpoint, while less often considered (with just a brief mention in McKague, Mosca, and Gisin \cite{mckague2009simulating} and some special-case use in \cite{casanova2011quantum,di2013embedding}), reveals exciting new phenomena. In particular, the $(n+1)$-rebit computer is able to efficiently simulate, not just unitary, but also non-unitary (indeed nonlinear) operators on the $n$-qubit computer. Thus, a major takeaway is that nonlinearity can be simulated by linear systems on a larger space\footnote{At this point, an interested reader may refer to Appendix \ref{sec:simpleExample}, where we present a simple example of such a simulation.}. 

We completely characterize the operators that can be simulated using this bottom-up approach to rebit simulation. It happens that the simulable operators are a subgroup of the so-called $\mathbb{R}$-linear operators, which we call $\mathbb{R}$-unitary. We give universal gate sets for the $\mathbb{R}$-unitaries, which we note can be constructed from just \emph{partial antiunitary} operators, i.e. those that act unitarily on some subspace of the $n$-qubit Hilbert space and antiunitarily on the rest. Furthermore, the (orthogonal) Clifford hierarchy \cite{gottesman1999demonstrating} on rebits maps to a Clifford hierarchy, dubbed the $\bbR$-Clifford hierarchy, contained within the $\mathbb{R}$-unitaries. In the spirit of the Gottesman-Knill theorem \cite{gottesman1998heisenberg}, we show that the second level of the $\bbR$-Clifford hierarchy, which is strictly larger than the Clifford group, is classically efficiently simulable. We also explore the efficiency of our rebit simulation for general $\mathbb{R}$-unitary circuits.

A good reason to consider bottom-up simulation of quantum computers is for algorithm design. Our results show that when designing a quantum algorithm using the circuit model, the designer has at their disposal not just unitary operators, but also the set of $\mathbb{R}$-unitary operators. Examples of this utility are the quantum simulation of the Majorana equation \cite{casanova2011quantum} and measurement of entanglement monotones \cite{di2013embedding}. It is important to note, however, that while the rebit simulator can model non-linear operators, this simulation does not allow us to exceed the power of quantum computers. After all, a rebit simulator is just a special case of a quantum circuit. This conclusion that non-linear operators can be simulated by quantum computers does not contradict the results of \cite{abrams1998nonlinear}, since the $\mathbb{R}$-linear operators are not among the non-linear operators described in \cite{abrams1998nonlinear} that imply polynomial-time solutions for $\mathsf{NP}$-complete and \#$\mathsf{P}$-complete problems. We expect of course that generalizations of our result exist, and more exotic operators (though still not those of the type found in \cite{abrams1998nonlinear}) will become simulable when more rebits (as a function of $n$) are included in the simulator.


\subsection{Two kinds of simulation}
\label{sec:twokinds}
Here we offer a definition of simulation that suits our needs, but also applies to most other uses of the term. A list of examples is provided in Table~\ref{tab:examples}.

A simulation is a tuple $(L,P,\mathcal{P},O_{(\cdot)})$ where $L$ and $P$ are sets of states of the logical (the simulated) and the physical (the simulator) state spaces. The map $\mathcal{P}:L\rightarrow P$ serves to relate the two. Because the simulator should be faithful to the simulated space, we require that $\mathcal{P}$ is injective -- i.e.~$\mathcal{P}(a)=\mathcal{P}(b)$ implies $a=b$. But we do not require it to be surjective -- in general, $\text{Range}(\mathcal{P})\subseteq P$ and there is a partial inverse $\mathcal{L}:\text{Range}(\mathcal{P})\rightarrow L$.

The final element $O_{(\cdot)}$ of the tuple specifies either the operators we want to simulate -- i.e.~$O_{(\cdot)}=O_L:L\rightarrow L$ corresponding to a top-down simulation -- or the operators the simulator can support -- i.e.~$O_{(\cdot)}=O_P:\text{Range}(\mathcal{P})\rightarrow\text{Range}(\mathcal{P})$ corresponding to a bottom-up simulation. Once one set of operators (either $O_L$ or $O_P$) is specified, the other can be determined by using the established maps $\mathcal{P}$ and $\mathcal{L}$.

As an example, the Gottesman-Knill theorem \cite{gottesman1998heisenberg} provides a method for classical computers to simulate a limited, non-universal, $n$-qubit quantum computer. To be exact, the quantum computer is only allowed to perform Clifford gates that act on at most two qubits at a time (e.g.~the gate set $\{H,S,CX\}$ is sufficient) and single-qubit Pauli measurements, which make up the set of operators $O_L$. Furthermore, $L$ is the set of stabilizer states (which is closed under the aforementioned operators $O_L$) and $P$ is the set of $n\times (2n+1)$ binary matrices. A stabilizer state $\ket{\psi}$ of $n$-qubits is special because there are $n$ independent and mutually commuting Pauli operators $p_i$, $i=1,2,\dots,n$ so that $p_i\ket{\psi}=\ket{\psi}$. Since each such $n$-qubit Pauli can be specified using $2n+1$ bits (the last bit is to track a $\pm$ sign), these bit strings form the rows of a full rank matrix in $P$. This describes the map $\mathcal{P}$.

Gottesman-Knill as described is a top-down simulation -- the goal is to simulate the Clifford gates $O_L$. Indeed, the simulator can achieve this efficiently, since each operator in $O_P$ (i.e.~a classical manipulation of the binary matrices in $P$) that corresponds to an operator in $O_L$ takes constant time. Yet, important insights come from the bottom-up viewpoint. For instance, the classical simulation can exactly (and efficiently) calculate the entire probability distribution that the quantum computer can only sample from! The classical simulation is actually \emph{more} powerful than the Clifford quantum computer it was designed to simulate\footnote{There are actually several variants of efficiently-simulable Clifford circuits -- see \cite{jozsa2014classical} and \cite{koh2015further}. To be concrete, pick the variant in which the Clifford circuit is nonadaptive, the input is in the computational basis, and the objective is so-called \emph{strong simulation}, i.e. to calculate the probability a bit string $y$ is observed when measuring any subset of output qubits. Theorem 4 in \cite{jozsa2014classical} shows as a corollary of Gottesman-Knill that this particular strong simulation task is efficient on a classical computer. Furthermore, as shown in Proposition 1 of  \cite{terhal2004adptive}, the ability to strongly simulate a class of circuits implies the ability to classically efficiently sample from it. So, indeed the classical simulation is strictly more powerful than the quantum computation in this case.}.  As this example demonstrates, it is in this bottom-up manner that unexpected phenomena can occur in simulations. For another example, see Appendix~\ref{app:gottesman_knill} where we show that Gottesman-Knill can also simulate complex conjugation of stabilizer states.

\begin{table*}
\centerline{\begin{tabular}{|c|c|c|c|c|}
\hline
Task & Simulated space $L$ & Simulator space $P$ & Mapping $\mathcal{P}$ & operators $O_L$ \\
\hline\hline
Hamiltonian Sim.~\cite{cubitt2017universal} & $n$ qubits & 2D qubit lattice & CMP map \cite{cubitt2017universal} & All Hamiltonians \\\hline
Fermionic Sim.~\cite{bravyi2002fermionic} & $M$ Fermi modes & $Q$ qubits & Bravyi-Kitaev \cite{bravyi2002fermionic} & 2-body Hamiltonian \\\hline
Gottesman-Knill \cite{gottesman1998heisenberg} & Stabilizer states & Binary matrices & Heisenberg repr. & Stabilizer operators \\\hline
Quantum codes \cite{gottesman1997stabilizer} & $k$ qubits & $n$ qubits, $n>k$ & Encoder & Universal gate set \\\hline
Compiling \cite{dawson2005solovay} & $n$ qubits & $n$ qubits + $m$ ancillas & $\ket{\psi}\mapsto\ket{\psi}\otimes\ket{0}^{\otimes m}$ & All unitaries\footnotemark \\\hline
Rebits \cite{mckague2009simulating} & $n$ qubits & $n+1$ rebits & 
Eq.~\eqref{eq:rebit_encoding}
& $\mathbb{R}$-unitaries \\\hline
\end{tabular}}
\caption{\label{tab:examples} Several examples of our definition of simulation.}
\end{table*}
\footnotetext{Compiling looks like a rather trivial simulation until considered from the bottom-up perspective. The nontriviality is that the simulator has a set of operators $O_P$ that is much smaller than $O_L$. For instance, $O_P$ is often finite.}

\subsection{Simulation using rebits} \label{sec:simulationusingrebits}

The task we focus on in this paper is the simulation of $n$ qubits using $n+1$ rebits. In particular, we make use of the single-ancilla rebit encoding of qubits \cite{mckague2009simulating}. Using the definition of simulation given in Section \ref{sec:twokinds}, our rebit simulation is a tuple $(L,P,\mathcal{P},O)$, where $L$ is the set of $n$-qubit states, $P$ is the set of $(n+1)$-rebit states, and the \emph{encoding} map $\mathcal P:L\rightarrow P$ takes $n$-qubit states to $(n+1)$-rebit states as follows:
\begin{equation}\label{eq:rebit_encoding}
\mathcal{P}:\ket{\psi} \mapsto \Re\ket{\psi}\otimes\ket{0}_a+\Im\ket{\psi}\otimes\ket{1}_a,
\end{equation}
where subscripts $a$ indicate the ``ancilla" rebit, and $\Re$ and $\Im$ take the real and imaginary parts, respectively. The inverse of $\mathcal{P}$ is (as we show later) the \emph{decoding} map
\begin{equation}
\mathcal L: \ket \phi \mapsto (\bra 0+i \bra 1)_a \ket\phi.
\end{equation}
Handed only an unknown $n$-qubit state, the encoding operator $\mathcal{P}$ is nonlinear and thus unphysical. However, $n$-qubit states with only real amplitudes can be encoded simply by appending an ancilla $\ket{0}$ to the register. Since quantum algorithms (i.e.~those solving problems in $\mathsf{BQP}$ \cite{bernstein1997quantum}) start on the all-zeroes state $\ket{0}^{\otimes n}$, the inability to start a rebit simulation from an arbitrary, unknown state is not a problem for this standard computational model.\footnote{Indeed, a similar situation is forced upon error-correcting stabilizer codes, for which only certain states (e.g.~stabilizer states and maybe certain magic states depending on the code) are able to be fault-tolerantly prepared.}

Our primary task is to study rebit simulations corresponding to different sets $O$ of operators. Our main results concern 
bottom-up simulation using rebits (i.e. simulations corresponding to sets $O = O_P$), though to draw a contrast, we include a discussion of top-down simulation using rebits (where we choose $O =  O_L$) in Section \ref{sec:top_down}.

\subsubsection{Bottom-up simulation using rebits}
\label{sec:bottomuprebits}

We now introduce some terminology for sets of operators that will be useful for describing the bottom-up simulation using rebits. Rather than specify $O_P$ or $O_L$ directly, we specify instead a physically motivated set of operators $O$ that we can actually implement on $n+1$ qubits in the lab. In general, however, these operators $O$ do not map $P$ to $P$, i.e.~rebits to rebits. Thus, we need to restrict the simulator to operators $O_P$ that do. To do this, take $O_P=O\cap R_{n+1}\subseteq O$, where $R_{n+1}$ is the set of real linear operators on $n+1$ qubits. Correspondingly, there is a set of operators $O_L$ that behave the same way on the $n$-qubit states in $L$ that operators in $O_P$ behave on $P$: i.e.~$A\in O_L$ if and only if there is $B\in O_P$ such that $A\ket{\psi}=\mathcal{L}(B\mathcal{P}(\ket{\psi}))$ for all $\ket{\psi}\in L$. With a convenient abuse of notation, we denote this set $O_L=\mathcal{L}(O_P)$. These are the operators on $L$ being simulated. 





We use the following language to describe these sets $O$, $O_P$, and $O_L$ (see also Figure  \ref{fig:commdiagram}). 
If $O$ is the set of $\texttt S$ operators\footnote{$\texttt S$ is a placeholder for the word used to describe operators in the set $O$. For example, $\texttt S$ could stand for the word `unitary' or `linear'.} (for example, the unitary operators), then we call $O_P$ the set of real $\texttt S$ operators (denoted $O^{\bbR}$) on $P$, and $O_L$ the set of $\bbR$-$\texttt S$ operators (denoted $\mathbb{R}O$) on $L$.
For example, if $O = U_{n+1}$ is the set of $(n+1)$-qubit unitary operators, then $O_P=U^{\bbR}_{n+1}$ is the set of $(n+1)$-rebit \textit{real unitary} operators (also called orthogonal operators) and $O_L= \bbR U_n$ is the set of $n$-qubit $\bbR$\textit{-unitary} operators.\footnote{Note that $\bbR$\textit{-unitary} (which we pronounce as `R unitary') should not be confused with real unitary, which means real and unitary. The same goes for $\bbR$-linear, $\bbR$-Pauli, etc.}

\begin{figure}[h]
\centering
\begin{tikzcd}[column sep=
5em, row sep=7em
]
O \arrow[dr,"(\cdot)\cap R_{n+1}"] \\ \mathbb RO \arrow[r, bend left=15,"\mathcal P"] \arrow[u,leftarrow,"\mathbb R(\cdot)"] 
&  O^{\mathbb R} \arrow[l, bend left=15,"\mathcal L"]
\end{tikzcd}
\caption{Commutative diagram illustrating the relationships between the sets $O$, $\bbR O$ and $O^\bbR$.}
\label{fig:commdiagram}
\end{figure}
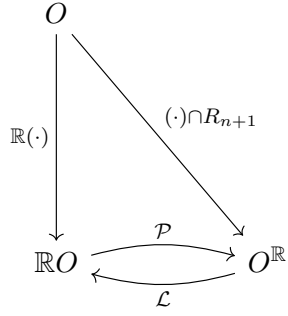


\subsection{Our results} 
\label{sec:results}

In this paper, we study the rebit simulation from a bottom-up perspective. Thus, one important task which we undertake is to determine, for various sets of $\texttt S$ operators, what the corresponding set $\bbR $-$\texttt S$ is. In the case of unitaries, we can then proceed to find universal gate sets for the $\mathbb{R}$-unitaries. We are led naturally to consider the efficiency of simulating these gate sets, and for specific subgroups (the $\mathbb{R}$-Clifford subgroup) we can even find efficient classical simulations.

In this subsection, we aim to highlight these results with a self-contained, albeit brief, presentation. References are provided to theorems and proofs in the main text (i.e.~the following sections) where the details can be perused. See Figure \ref{fig:subsetdiagram} for a summary of the various sets of operators considered in this paper.

\tikzstyle{block} = [rectangle, 
    text width=5em, text centered, minimum height=2em]
\tikzstyle{line} = [draw]

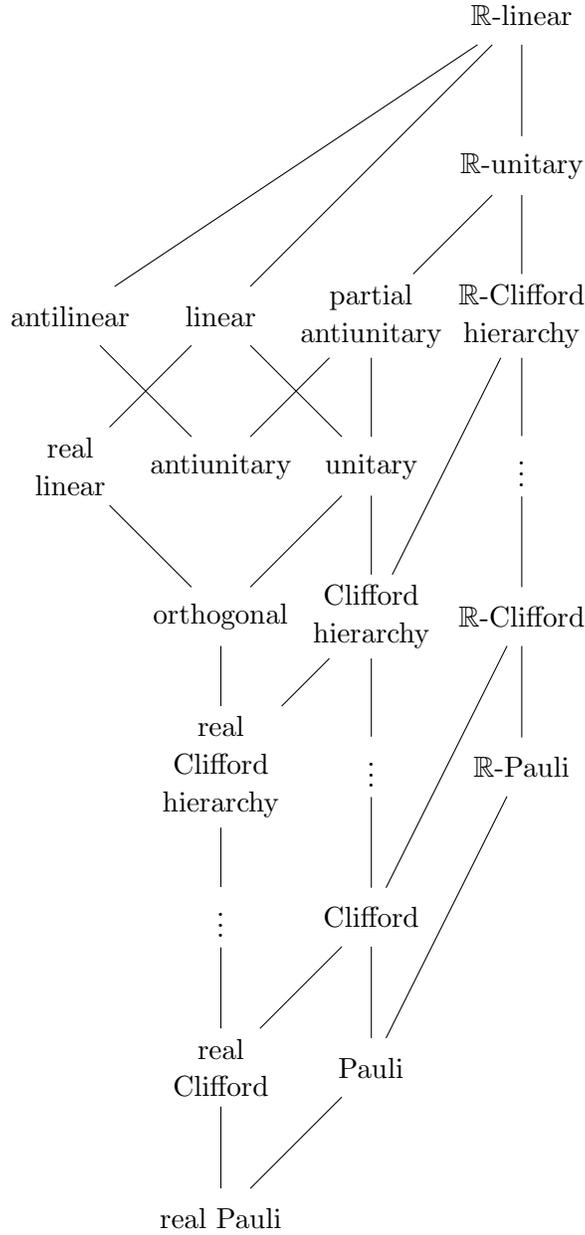
\begin{figure}
\begin{center}
\begin{tikzpicture}[node distance = 2cm, auto]

    \node [block] (a) {$\mathbb R$-linear};
    \node [block,below of=a] (b) {$\mathbb R$-unitary};
    \node [block,below of=b] (c) {$\mathbb R$-Clifford hierarchy};
    \node [block,below of=c] (d) {$\vdots$};
    \node [block,below of=d] (e) {$\mathbb R$-Clifford};
    \node [block,below of=e] (f) {$\mathbb R$-Pauli};
    \node [block,left of=c] (cl) {partial antiunitary};
    \node [block,left of=d] (dl) {unitary};
    \node [block,left of=e] (el) {Clifford hierarchy};
    \node [block,left of=f] (fl) {$\vdots$};
    \node [block,below of=fl] (gl) {Clifford};
    \node [block,below of=gl] (hl) {Pauli};
    \node [block,left of=cl] (cll) {linear};
    \node [block,left of=cll] (clll) {antilinear};
    \node [block,below of=clll] (dlll) {real \\ linear};
    \node [block,left of=dl] (dll) {antiunitary};
    \node [block,below of=dll] (ell) {orthogonal};
    \node [block,below of=ell] (fll) {real Clifford hierarchy};
    \node [block,below of=fll] (gll) {$\vdots$};
    \node [block,below of=gll] (hll) {real Clifford};
    \node [block,below of=hll] (ill) {real Pauli};
    
    \path [line] (a) -- (b);
    \path [line] (b) -- (c);
    \path [line] (c) -- (d);
    \path [line] (d) -- (e);
    \path [line] (e) -- (f);    
    \path [line] (b) -- (cl);  
    \path [line] (cl) -- (dl);
    \path [line] (dl) -- (el);
    \path [line] (el) -- (fl);
    \path [line] (fl) -- (gl);
    \path [line] (gl) -- (hl);
    \path [line] (c) -- (el);
    \path [line] (e) -- (gl);
    \path [line] (f) -- (hl);
    \path [line] (cl) -- (dll);
    \path [line] (dl) -- (ell);
    \path [line] (el) -- (fll);
    \path [line] (gl) -- (hll);
    \path [line] (hl) -- (ill);
    \path [line] (ell) -- (fll);
    \path [line] (fll) -- (gll);
    \path [line] (gll) -- (hll);
    \path [line] (hll) -- (ill);
   \path [line] (a) -- (cll);
   \path [line] (a) -- (clll);
   \path [line] (clll) -- (dll);
   \path [line] (cll) -- (dl);
   \path [line] (cll) -- (dlll);
   \path [line] (dlll) -- (ell);
\end{tikzpicture}
\end{center}
\caption{Diagram illustrating the relationships between different classes of operators considered in this paper. A line from $\mathcal A$ to $\mathcal B$ (where $\mathcal A$ is higher than $\mathcal B$) means that $\mathcal A$ is a proper superset of $\mathcal B$. The ellipses represent the infinite towers of classes corresponding to different levels of the respective Clifford hierarchies.}
\label{fig:subsetdiagram}
\end{figure}

\subsubsection{Characterization of various rebit simulators} \label{sec:characterizationOf}
First, 
we find $\bbR$-$\texttt S$ for the following subsets $\texttt S$: (1) linear operators, (2) unitary operators, (3) Pauli operators, (4) operators in the $k$th level of the Clifford hierarchy. The following theorem summarizes our results.

\begin{theorem} \label{thm:mainTheorem1}
Let $\Gamma$ be an operator on $n$ qubits.
\begin{enumerate}
\item (Theorem~\ref{thm:equiv_defRlinear}) $\Gamma$ is $\bbR$-linear if and only if there exist complex linear operators $A$ and $B$ such that $\Gamma$ can be written as
\begin{equation}
\Gamma = A + BK,
\end{equation}
where $K$ denotes the complex conjugation operator, $K:\ket{\psi}\mapsto\Re\ket{\psi}-i\Im\ket{\psi}$.
\item (Theorem~\ref{thm:OrthogonalABK2}) $\Gamma$ is $\bbR$-unitary if and only if $\Gamma=A+BK$ is $\bbR$-linear and $A$ and $B$ are complex linear operators satisfying the `unitarity' constraints
\begin{eqnarray}
A^\dag A + B^T \bar B = I  \nonumber
\\ 
A^\dag B + B^T \bar A = 0.
\end{eqnarray} 
\item (Theorem \ref{thm:encoded_RPaulis}) $\Gamma$ is $\bbR$-Pauli if and only if it can be written as 
\begin{equation}
\Gamma = PK^b,
\end{equation}
where $P = i^c p_1 \otimes \ldots \otimes p_n$ is a Pauli operator (where $c \in \{0,1,2,3\}$ and $p_j \in \{I,X,Y,Z\}$), and $b\in\{0,1\}$.
\item (Theorem \ref{thm:encoded_hierarchy}) $\Gamma$ is in the $k$th level of the $\mathbb{R}$-Clifford hierarchy $\bbR\cC_k$ if and only if $\Gamma (\bbR \cC_1) \Gamma^\dag \subseteq  \bbR \cC_{k-1}$, where $\bbR \cC_1$ is the set of $\bbR$-Paulis.
\end{enumerate}
\end{theorem}

The reader might recognize that, in operator theory and linear algebra\footnote{See, for example, \cite{huhtanen2011real, huhtanen2012function, eisele2012finding}. Note that our definition of $\bbR$-linearity coincides with their definition of real linearity.}, the term $\bbR$-linear is also used to describe a map $f:V\rightarrow V'$, where $V$ and $V'$ are complex vector spaces, that satisfies
\begin{equation}\label{eq:linealg_rlinear}
f(ax+by) = af(x) + bf(y)
\end{equation}
for all $x,y\in V$ and $a,b\in \bbR$.
It turns out\footnote{More accurately, we chose the terminology so that the two definitions of $\bbR$-linearity coincide.} that this definition is equivalent to that in Part 1 of Theorem \ref{thm:mainTheorem1} (see Theorem~\ref{thm:equiv_defRlinear}). It may not be surprising that $\bbR$-linearity turns out to be the defining characteristic of the simulated system; after all, the encoding map $\mathcal{P}$ defined in Eq.~\eqref{eq:rebit_encoding} is blatantly $\bbR$-linear in the linear algebraic sense of Eq.~\eqref{eq:linealg_rlinear}. More detail on the linear algebraic definition of $\bbR$-linearity is included for reference in Appendices~\ref{app:rlinear} and \ref{app:algebraicProperties}.

The $\bbR$-unitary operators (part 2 of Theorem~\ref{thm:mainTheorem1}) are of special importance. Since they can be simulated by (real) unitary operators acting on rebits, they correspond to the set of operators that can be simulated by physical systems using the rebit encoding. This is an example of bottom-up simulation: we start with the set of operators that we can perform on rebits (i.e. real unitary operators) and derive the set of operators (i.e. $\bbR$-unitaries) that we can simulate on qubits. The rest of our results will concern various properties of the $\bbR$-unitary operators.

\subsubsection{ $\bbR$-unitaries as products of partial antiunitaries}

Our second result concerns the class of partial antiunitary operators, which are a subset (but not a subgroup) of the $\bbR$-unitary operators. In \cite{mckague2009simulating}, McKague, Mosca, and Gisin note that partial antiunitaries -- described as operators ``which act only on a subspace" -- as well as products of partial antiunitaries, can be simulated using the rebit encoding.  However, a precise definition of partial antiunitarity was not provided in \cite{mckague2009simulating}. In this paper, we propose the following definition:

\begin{definition} \label{def1:partialAntiunitary_intro}
Let $S\subseteq \mathcal H$ be a subspace of a complex (finite-dimensional) vector space $\mathcal H$. An operator $\Gamma$ on $\mathcal H$ is a \textit{partial antiunitary} operator with respect to $S$ if 
\begin{enumerate}[(i)]
\item $\Gamma$ is additive, i.e.\ $\Gamma(\psi+\phi) = \Gamma(\psi)+\Gamma(\phi)$ for all $\psi,\phi \in \mathcal H$.
\item $\Gamma$ is unitary on $S^\perp$, i.e.\ $\langle \Gamma(\psi),\Gamma(\phi)\rangle = \langle \psi,\phi\rangle$ for all $\psi, \phi \in S^\perp$.
\item $\Gamma$ is antiunitary on $S$, i.e.\ $\langle \Gamma(\psi),\Gamma(\phi)\rangle = \langle \phi,\psi\rangle$ for all $\psi, \phi \in S$.
\item $\langle \Gamma(\psi),\Gamma(\phi)\rangle = 0$ for all $\psi \in S, \phi \in S^\perp$.	
\end{enumerate}
\end{definition}

Perhaps surprisingly, this definition leads to the following relation between partial antiunitaries and $\bbR$-unitaries. 
\begin{theorem} \label{thm:Runitarypartial}
(Theorem~\ref{thm:logicalActionOfOrthogonals}) For any $\bbR$-unitary operator $\Gamma$, there exists an integer $0<k \leq (n+1)2^n(2^n-1)/2$ and partial antiunitary operators $\Gamma_1,\ldots, \Gamma_{k-1},\Gamma_k$ such that $\Gamma=\Gamma_k\Gamma_{k-1}\dots \Gamma_1$.
\end{theorem}

Theorem \ref{thm:Runitarypartial} tells us that in order to simulate an $\bbR$-unitary, it suffices to just simulate sequences of partial antiunitary operators. Indeed, we can find universal gate sets for the $\mathbb{R}$-linear operators that consist of only finitely many partial antiunitary gates. Some of the simplest are presented here in terms of the partial antiunitaries $CK_i$ (a single-qubit gate on qubit $i$) and $CCK_{ij}$ (a two-qubit gate on qubits $i,j$) defined by (i.e.~they conjugate only the amplitudes of basis states in which the ``control" qubits are 1)
\begin{eqnarray}
CK_i:\sum_{x\in\{0,1\}^n}(a_x+ib_x)\ket{x}&\mapsto&\sum_{\stackrel{x\in\{0,1\}^n}{x_i=0}}(a_x+ib_x)\ket{x}+\sum_{\stackrel{x\in\{0,1\}^n}{x_i=1}}(a_x-ib_x)\ket{x},\\
CCK_{ij}:\sum_{x\in\{0,1\}^n}(a_x+ib_x)\ket{x}&\mapsto&\sum_{\stackrel{x\in\{0,1\}^n}{x_ix_j=0}}(a_x+ib_x)\ket{x}+\sum_{\stackrel{x\in\{0,1\}^n}{x_ix_j=1}}(a_x-ib_x)\ket{x},
\end{eqnarray}
as well as the unitary operators controlled-controlled-$Z$ (CCZ), global phase $G(\theta)\ket{\psi}=e^{i\theta}\ket{\psi}$, and Hadamard gate $H$.
\begin{theorem}
(Proposition~\ref{prop:finitegatesets}) The gate sets $\{H,CCK,G(\pi/4)\}$ and $\{H,CCZ,CK,G(\pi/4)\}$ are able to approximately\footnote{We have a specific notion of approximation in mind. In principle, any metric on $P$ (the simulator space) implies a metric on $L$ (the simulated space), and likewise metrics can be defined for operators on those spaces, even if those operators on $L$ are nonlinear. To make an operator norm on $\bbR$-linears for instance, one instead evaluates the operator norm of their simulations. See the definitions at the end of Section~\ref{sec:R-linear} for more rigor on the $\bbR$-linear operator norm.} simulate any $\mathbb{R}$-unitary operator.
\end{theorem}

\subsubsection{Computational complexity and a generalization of the Gottesman-Knill Theorem}

Third, we explore the efficiency of the rebit encoding with regards to universal gate sets for the top-down \emph{and} bottom-up simulations.
\begin{theorem}
Let $C$ be a depth-$d$ circuit on $n$ qubits.
\begin{enumerate}
\item (Theorem~\ref{thm:topdown_effic}) If $C$ is a unitary circuit consisting of gates from $\{H,T,\text{CNOT}\}$, then $C$ (applied to $\ket{0}^{\otimes n}$) can be simulated using either an orthogonal circuit of depth at most $dn$ on $n+1$ rebits or an orthogonal circuit of depth at most $d$ on $2n$ rebits.
\item (Theorem~\ref{thm:bottomup_effic}) If $C$ is an $\mathbb{R}$-unitary circuit consisting of gates from $\{H,CCZ,CK,G(\pi/4)\}$, then $C$ (applied to $\ket{0}^{\otimes n}$) can be simulated using either an orthogonal circuit of depth at most $dn$ on $n+1$ rebits or an orthogonal circuit of depth at most $d\lceil\log_2n\rceil$ on $2n$ rebits.
\end{enumerate}
\end{theorem}

Further exploration of the $\mathbb{R}$-Clifford circuits reveals an efficient classical simulation, enlarging the scope of the Gottesman-Knill simulation.
\begin{theorem}
If $C$ is an $\mathbb{R}$-Clifford circuit on $n$ qubits, then there is an efficient classical algorithm to sample from $C\ket{0}^{\otimes n}$ in time $O(n^2)$, when the output qubits of the circuit are measured in the computational basis.
\end{theorem}
\noindent The $\mathbb{R}$-Clifford circuits are a strictly larger set of operators than the Clifford circuits. For instance, $CK$ is an $\mathbb{R}$-Clifford gate that is not even linear, much less Clifford. 

\subsection{Related work} \label{subsec:relatedWork}

The use of rebits in quantum computation dates back to the 1990s, where it was shown that real amplitudes (or even rational amplitudes \cite{adleman1997quantum}) suffice for universal quantum computation \cite{bernstein1997quantum}. This was first proven in the quantum Turing machine model \cite{bernstein1997quantum}, before direct proofs of this result for the quantum circuit model were found \cite{kitaev1997quantum, boykin2000new, rudolph20022, shi2002both, aharonov2003simple}. These circuit-model proofs all involve proving the existence of computationally universal gate sets that consist of only gates with real coefficients. A simple example of such a gate set is the set containing Toffoli and Hadamard gates \cite{shi2002both, aharonov2003simple}.

Besides the single-ancilla rebit encoding, other rebit encodings have been proposed. An example is the subsystem-division-respecting encoding introduced by \cite{mckague2009simulating}. Unlike the single-ancilla encoding, local operators remain local under this encoding. As noted in \cite{mckague2013power}, an implication of this result is that no experiment can distinguish between quantum mechanics with real amplitudes and quantum mechanics with complex amplitudes, unless one makes assumptions about the dimensions of systems.

Rebits have also been studied in relation to several different topics in quantum information theory. For example, it was shown that states and measurements on a real Hilbert space are sufficient for the maximal violation of Bell inequalities \cite{pal2008efficiency, mckague2009simulating}. Another example is in computational complexity theory: as we discussed above, choosing to use rebits instead of qubits does not change the power of $\BQP$. It was shown in \cite{mckague2013power} that such a choice also does not change the power of many other quantum complexity classes, like $\QMA$, $\QCMA$ or $\QIP(k)$. 

Finally, we note that the techniques used in analyzing rebit circuits are closely related to that used in analyzing quaternionic circuits. It can be shown, for example, that quaternionic quantum circuits are no more powerful than complex quantum circuits  \cite{fernandez2003quaternionic}, just as complex quantum circuits are no more powerful than real ones.


\subsection{Notation}
\label{sec:notation}

Throughout this paper, all vector spaces are assumed to be finite-dimensional. The complex conjugation operator is denoted by $K$. When $K$ acts on vectors or linear operators, we assume that it acts on them with respect to the computational basis. For a scalar, vector or matrix $A$, we sometimes write $K(A) = \bar A$. The real and imaginary parts of a scalar, vector or matrix $A$ are defined in terms of $K$: the real part of $A$ is given by $\Re A = \frac 12 (A+K(A))$ and the imaginary part of $A$ is given by $\Im A = \frac 1{2i}(A-K(A))$. Hence we could write $I = \Re+ i\Im$ and $K = \Re - i\Im$. We say that $A$ is \textit{real} if $\Re A = A$, and that $A$ is \textit{imaginary} if $\Im A = -i A$. 

The identity matrix is denoted by $I$, and the Pauli matrices by $X$, $Y$ and $Z$. We denote the Hadamard gate by $H$, the phase gate by $S = \mathrm{diag}(1,i)$, and the $T$ gate by  $T = \mathrm{diag}(1,e^{i\pi/4})$. For any unitary gate $G$, we write $CG$ to refer to the controlled-$G$ gate, and $CCG$ to refer to the controlled-controlled-$G$ gate. For example, $CX$ is the CNOT gate, and $CCX$ is the Toffoli gate. The notation $CG_{ij}$ means that the control register is $i$ and the target register is $j$. Similarly, $CCG_{ijk}$ means that the control registers are $i$ and $j$ and the target register is $k$. More controls will be indicated with a superscript on the symbol $C$, i.e.~$C^{h}X_{c_1,c_2,\dots,c_h,t}$ has $h$ control qubits $c_i$ and one target qubit $t$.

We denote the $n$-qubit complex Hilbert space by $\mathcal H_n(\mathbb C)$ and the $n$-rebit real Hilbert space by $\mathcal H_n(\mathbb R)$, i.e., $\mathcal H_n(\mathbb C)$ is a $2^n$-dimensional vector space over the complex numbers and $\mathcal H_n(\mathbb R)$ is a $2^n$-dimensional vector space over the real numbers. 


We denote the set of (complex) linear operators on $\cH_n(\bbC)$ by $L_n$, and the set of real linear operators by $R_n$, i.e. 
\begin{equation}
R_n = \set{\Gamma \in L_n: \Im \Gamma = 0}.
\end{equation} 
The group of unitary operators on $\cH_n(\bbC)$ is denoted by $U_n = \set{U \in L_n: U^\dag U = I}$. The group of orthogonal operators is denoted by $T_n= \set{T \in U_n: \Im T=0}$. Technically, while it is clear that $L_n$ and $U_n$ cannot act on the rebit space $\mathcal{H}_n(\bbR)$ (due to failing closure), it is possible for the operators in $R_n$ and $T_n$ to act on either $\mathcal{H}_n(\bbC)$ or $\mathcal{H}_n(\bbR)$ (i.e.~on qubits or rebits). We generally let context sort this ambiguity out, and moreover, because all these operators can be represented as matrices independent of the vector spaces on which they act, this subtlety should not be an issue.



For operators $A, B$, we write $A\propto B$ if there exists $\theta \in \bbR$ such that $A = e^{i\theta} B$.

\section{Quantum circuits with rebits}
\label{sec:rebitEncoding}

We start by introducing the rebit encoding and decoding maps for quantum circuits. A quantum circuit is in general described in terms of its states, operators  and measurements. We will now describe each of these separately.

\subsection{Rebit encoding and decoding of states} 

We restrict our attention to pure states. Note that we do not lose any generality with this restriction since any mixed state can be purified to a pure state by adding ancilla qubits \cite{nielsen2010quantum}. 
\begin{definition} 
The \textit{rebit encoding} of a quantum state $\ket \psi \in \mathcal H_n(\mathbb C)$ is the state $\mathcal P (\ket \psi) \in \mathcal H_{n+1}(\mathbb R)$ defined by
\begin{eqnarray} \label{eq:defPstates}
\mathcal P (\ket \psi) = \Re \ket \psi \otimes \ket{0}_a + \Im \ket \psi \otimes \ket{1}_a.
\end{eqnarray}
\end{definition}
Here, the real and imaginary parts (denoted by $\Re$ and $\Im$ respectively) of $\ket\psi$ are defined with respect to the computational basis.

Following the terminology in Section \ref{sec:simulationusingrebits}, the space of qubits $\cH_n(\bbC)$ may be thought of as the \textit{logical space} $L$ and the space of rebits $\cH_{n+1}(\bbR)$ may be thought of as the \textit{physical space} $P$ that encodes states in $\cH_n(\bbC)$ through the map $\mathcal P$.

Explicitly, if 
$$\ket\psi = \sum_{i_1,\ldots, i_n \in \{0,1\}^n} \psi_{i_1,\ldots,i_n} \ket {i_1,\ldots,i_n},$$
then 
$$\mathcal P(\ket\psi) = \sum_{i_1,\ldots,i_n,j \in \{0,1\}^{n+1}} \psi_{i_1,\ldots,i_n,j} \ket {i_1,\ldots,i_n,j},$$
where we have used the notation $\psi_{i_1,\ldots,i_n,0} = \Re(\psi_{i_1,\ldots,i_n})$ and $\psi_{i_1,\ldots,i_n,1} = \Im(\psi_{i_1,\ldots,i_n})$.
The normalization condition $\sum_{i_1\ldots i_n \in \{0,1\}^n} |\psi_{i_1,\ldots,i_n}|^2 = 1$ becomes 
\begin{equation}
\sum_{i_1\ldots i_n,j \in \{0,1\}^{n+1}} \psi_{i_1,\ldots,i_n,j}^2 = 1
\end{equation}
in the encoded space. For example, when $n=1$ and each entry of the state $\ket \psi$ is written in terms of its real and imaginary parts, $\mathcal P(\cdot)$ acts on $\ket\psi$ as follows:
\begin{equation}\label{eq:eg1state}
\mathcal P: (a+ib)\ket 0 + (c+id)\ket 1 \mapsto a\ket{00} + b\ket{01} + c\ket{10} + d\ket{11}.
\end{equation}
Incidentally, this discussion shows that $\mathcal{P}$ preserves the $l_2$ norm.
\begin{prop}\label{prop:l2_norm_equal}
Denote the $l_2$ norm for $\ket{\psi}\in\mathcal{H}_n(\bbC)$ as $\|\ket{\psi}\|=|\langle\psi|\psi\rangle|^2$. Then $\|\ket{\psi}\|=\|\mathcal{P}(\ket{\psi})\|$ for all $\ket{\psi}\in\mathcal{H}_n(\bbC)$.
\end{prop}

Next, we show that the map $\mathcal P$ is invertible and that its inverse $\mathcal L: \mathcal H_{n+1}(\mathbb R) \rightarrow \mathcal H_n(\mathbb C)$ is given by
\begin{equation} \label{eq:defLstates}
\mathcal L: \ket \phi \mapsto (\bra 0+i \bra 1)_a \ket\phi.
\end{equation}
We refer to $\mathcal L(\ket \phi)$ as the rebit decoding of the rebit state $\ket \phi$.

\begin{prop}\label{prop:LP_are_inverses}
The maps $\mathcal P$ and $\mathcal L$ defined in \eq{eq:defPstates} and \eq{eq:defLstates} are inverses of each other.
\end{prop}
\begin{proof}
For all $\ket\psi$ and $\ket\phi$,
$$\mathcal L \circ \mathcal P \ket \psi = (\bra 0 + i \bra 1)_a (\Re \ket\psi\otimes \ket{0}_a + \Im\ket\psi \otimes \ket{1}_a) = \Re\ket\psi + i\Im\ket\psi = \ket\psi$$ and $$\mathcal P \circ \mathcal L \ket \phi = \mathcal P(\bra 0_a \ket \phi + i \bra 1_a \ket \phi) = \bra 0_a \ket \phi \otimes \ket 0_a + \bra 1_a \ket \phi \otimes \ket 1_a = (\ketbra 00 + \ketbra 11)_a \ket \phi  = \ket\phi.$$
Hence, $\mathcal L \circ \mathcal P = I$ and $\mathcal P \circ \mathcal L = I$. 
\end{proof}

We noted in Section \ref{sec:characterizationOf} that $\mathcal{P}$ is $\bbR$-linear. Here we also note that $\mathcal{L}$ is linear. We collect these observations in the following proposition.
\begin{prop}\label{prop:prlinear_llinear}
$\mathcal{P}$ is $\bbR$-linear on $\mathcal{H}_n(\bbC)$ and $\mathcal{L}$ is linear on $\mathcal{H}_{n+1}(\bbR)$. That is, for all $\ket{\psi_1},\ket{\psi_2}\in\mathcal{H}_n(\mathbb{C})$ and all $a,b\in\bbR$,
\begin{equation}
\mathcal{P}(a\ket{\psi_1}+b\ket{\psi_2})=a\mathcal{P}(\ket{\psi_1})+b\mathcal{P}(\ket{\psi_2}),
\end{equation}
and for all $\ket{\phi_1},\ket{\phi_2}\in\mathcal{H}_{n+1}(\bbR)$ and $\alpha,\beta\in\bbR$,
\begin{equation}
\mathcal{L}(\alpha\ket{\phi_1}+\beta\ket{\phi_2})=\alpha\mathcal{L}(\ket{\phi_1})+\beta\mathcal{L}(\ket{\phi_2}).
\end{equation}
\end{prop}

\subsection{Rebit encoding and decoding of operators}
\label{sec:rebittransformation}

Having defined how states are encoded, we now proceed to describe how operators on qubits are encoded. Let $\Gamma: \mathcal H_n(\mathbb C) \rightarrow \mathcal H_n(\mathbb C)$ be any operator on the set of $n$-qubit states. We define the rebit encoding of $\Gamma$ to be the map $\mathcal P(\Gamma) : \mathcal H_{n+1}(\mathbb R) \rightarrow \mathcal H_{n+1}(\mathbb R)$ given by
\begin{equation}
\mathcal P(\Gamma): \ket \phi \mapsto \mathcal P(\Gamma \mathcal L(\ket \phi)).
\end{equation} 
Note that we have used the same symbol $\mathcal P$ to denote the rebit encoding of both states and operators.

Let $W: \mathcal H_{n+1}(\mathbb R) \rightarrow \mathcal H_{n+1}(\mathbb R)$ be any operator on the set of $(n+1)$-rebit states. We define the rebit decoding of $W$ to be the map $\mathcal L(W) : \mathcal H_n(\mathbb C) \rightarrow \mathcal H_n(\mathbb C)$ given by
\begin{equation}
\label{eq:LWDefn}
\mathcal L(W): \ket \psi \mapsto \mathcal L(W \mathcal P(\ket \psi)).
\end{equation} 
The above definitions are chosen so that the rebit encoding $\mathcal P$ of operators and the rebit decoding $\mathcal L$ of operators are inverses of each other. 

It is easily established that $\mathcal{P}$ and $\mathcal{L}$ for operators are group homomorphisms.
\begin{prop}\label{prop:PL_homomorphisms}
For all $\Gamma_1,\Gamma_2:\mathcal H_n(\mathbb C) \rightarrow \mathcal H_n(\mathbb C)$, we have $\mathcal{P}(\Gamma_1\Gamma_2)=\mathcal{P}(\Gamma_1)\mathcal{P}(\Gamma_2)$. Likewise, for all $W_1,W_2:\mathcal H_{n+1}(\mathbb R) \rightarrow \mathcal H_{n+1}(\mathbb R)$, we have $\mathcal{L}(W_1W_2)=\mathcal{L}(W_1)\mathcal{L}(W_2)$.
\end{prop}
\begin{proof}
We just prove the first statement here; the second is just as simple an application of the definitions. For all $\ket{\phi}\in\mathcal{H}_n(\bbC)$,
\begin{equation}
\mathcal{P}(\Gamma_1)\mathcal{P}(\Gamma_2)\ket{\phi}=\mathcal{P}(\Gamma_1)\mathcal{P}(\Gamma_2\mathcal{L}(\ket{\phi}))=\mathcal{P}(\Gamma_1\mathcal{L}(\mathcal{P}(\Gamma_2\mathcal{L}(\ket{\phi}))))=\mathcal{P}(\Gamma_1\Gamma_2\mathcal{L}(\ket{\phi}))=\mathcal{P}(\Gamma_1\Gamma_2)\ket{\phi}.
\end{equation}
using Proposition~\ref{prop:LP_are_inverses} in the second-to-last step.
\end{proof}
\noindent The definitions also imply that the encodings of states and operators behave naturally together, and likewise for decodings of the two,
\begin{eqnarray}
\mathcal{P}(\Gamma)\mathcal{P}(\ket{\phi})=\mathcal{P}(\Gamma\ket{\phi}),\\
\mathcal{L}(W)\mathcal{L}(\ket{\psi})=\mathcal{L}(W\ket{\psi}).
\end{eqnarray}
These relations are also easily checked.  Similarly to the case of states, the rebit encoding and decodings of operators have the same linearity properties.
\begin{prop}\label{prop:pop_rlinear_lop_linear}
$\mathcal{P}$ (for operators) is $\bbR$-linear on $L_n$ and $\mathcal{L}$ (for operators) is linear on $R_{n+1}$. That is, for all $\Gamma_1,\Gamma_2\in L_n$ and all $a,b\in\bbR$,
\begin{equation}
\mathcal{P}(a\Gamma_1+b\Gamma_2)=a\mathcal{P}(\Gamma_1)+b\mathcal{P}(\Gamma_2),
\end{equation}
and for all $W_1,W_2\in R_{n+1}$ and all $\alpha,\beta\in\mathbb{R}$,
\begin{equation}
\mathcal{L}(\alpha W_1+\beta W_2)=\alpha\mathcal{L}(W_1)+\beta\mathcal{L}(W_2).
\end{equation}
\end{prop}
\noindent The proof, left to the reader, is a straightforward application of the definitions.
We remark that the above treatment can also be extended in a similar way to superoperators, i.e. operators acting on operators. Such a generalization is used in, for example, Proposition \ref{prop:LofTranspose}.


The treatment so far 
deals with general operators. Since quantum mechanics is a linear (in fact, unitary) theory, we shall henceforth restrict our attention to only those operators on rebits which are linear, and adopt the terminology described in Section \ref{sec:bottomuprebits} to describe various subsets of the linear transformations.

Let $O_n \subseteq L_n$ be any subset of linear operators on $\cH_n(\bbC)$. We are interested in the following two classes:
\begin{enumerate}
\item $O^{\bbR}_n:= O_n \cap R_n$
\item $\bbR O_n = \cL(O^{\bbR}_{n+1})$
\end{enumerate}
which correspond respectively to the real $\texttt S$ and $\mathbb{R}$-$\texttt S$ sets defined in Section \ref{sec:bottomuprebits}. 

In general, the sets $O_n$ need not have additional structure. In this paper though, many of the sets $O_n$ that we consider are also groups. When this is the case, since $R_n$ is a subgroup of $L_n$ and the intersection of subgroups is a subgroup, the set $O_n^\bbR$ is also a subgroup of $L_n$.
Since $\cL$ is a group homomorphism (Proposition \ref{prop:PL_homomorphisms}), and group homomorphisms preserve subgroups, we get the following result:
\begin{prop}
\label{prop:subgroupResult}
Let $O_n \subseteq L_n$ be any subset of linear operators on $\cH_n(\bbC)$. If $O_n$ is a (proper) subgroup of $L_n$, then $\bbR O_n$ is a (proper) subgroup of $\bbR L_n$. 
\end{prop}

In Sections \ref{sec:R-linear} and \ref{sec:Runitaryoperators}, we prove parts 1 and 2 of Theorem \ref{thm:mainTheorem1}, thereby characterizing the sets $O_n^\bbR$ and $\bbR O_n$, when (i) $O_n = L_n$ is the set of linear operators, and when (ii) $O_n = U_n$ is the set of unitary operators. We defer the proof of the rest of Theorem~\ref{thm:mainTheorem1} to Section~\ref{sec:RClifford_hierarchy} where the $\mathbb{R}$-Clifford hierarchy is discussed.

\subsubsection{$\bbR$-linear operators}\label{sec:R-linear}
Consider the set $L_n$ of linear operators. In this section, we characterize the sets (i) $L_n^\bbR$ (the real linear operators) and (ii) $\bbR L_n$ (the $\bbR$-linear operators). 
Part (i) is straightforward, since $L^{\bbR}_n=L_n \cap R_n=R_n$ is the set of real linear operators. Part (ii) is more involved. To begin, we prove the following lemma.

\begin{lemma} \label{lem:imageWunderL}
Let $W: \mathcal H_{n+1}(\mathbb R) \rightarrow \mathcal H_{n+1}(\mathbb R)$ be a real linear operator. Then 
\begin{equation}
\label{eq:imageWunderL}
\mathcal L(W) = \left(\frac 12 \tr_a [W(I_a - iX_aZ_a)]  \right) + \left(\frac 12 \tr_a [W(Z_a + i X_a)]  \right) K .
\end{equation}
\end{lemma}
\begin{proof}

We consider the evolution of an arbitrary state $\ket\psi$ under $\cL(W) = \cL\circ W \circ \cP$:
\begin{eqnarray}
\ket \psi &\xrightarrow{\cP} & \Re\ket\psi\ket 0_a + \Im \ket\psi\ket 1_a \nn
&\xrightarrow{W} & W \ket 0_a \Re\ket\psi + W\ket 1_a \Im\ket\psi \nn
&\xrightarrow{\cL} &  [(\bra 0_a + i \bra 1_a) W \ket 0_a \Re +  (\bra 0_a + i \bra 1_a) W \ket 1_a \Im ] \ket\psi .
\end{eqnarray}

Hence,
\begin{eqnarray}
\cL(W) &=& (\bra 0_a + i \bra 1_a) W \ket 0_a \Re +  (\bra 0_a + i \bra 1_a) W \ket 1_a \Im \nn
&=& \frac 12 (\bra 0_a + i \bra 1_a) W \ket 0_a (I+K) + \frac 1{2i} (\bra 0_a + i \bra 1_a) W \ket 1_a (I-K) \nn
&=& \frac 12[ (\bra 0_a W \ket 0_a + i \bra 1_a W \ket 0_a - i \bra 0_a W \ket 1_a + \bra 1_a W \ket 1_a ) I \nn
&&+ (\bra 0_a W \ket 0_a + i \bra 1_a W \ket 0_a + i \bra 0_a W \ket 1_a - \bra 1_a W \ket 1_a ) K] \nn 
&=& \frac 12 \{\tr_a[W(\ketbra 00 + i\ketbra 01 -i \ketbra 10 + \ketbra 11)_a] I  \nn  
&&+  \tr_a[W(\ketbra 00 + i\ketbra 01 + i \ketbra 10 - \ketbra 11)_a] K \} \nn
&=&  \left(\frac 12 \tr_a [W(I_a -iX_aZ_a)]  \right) + \left(\frac 12 \tr_a [W(Z_a + i X_a)]  \right) K .
\end{eqnarray}

\end{proof}

The above proposition shows that the $\bbR$-linear operators can always be written in the form $A+BK$. It happens the converse is also true: any operator of the form $A+BK$ is an $\bbR$-linear operator, as the following lemma proves.
\begin{lemma} \label{lem:imageABunderK}
Let $A,B : \mathcal H_n(\mathcal C) \rightarrow \mathcal H_n(\mathcal C)$ be complex linear operators. Then,
\begin{equation} \label{eq:PABK}
\mathcal P(A+BK) = \Re A \otimes I + \Im A \otimes XZ + \Re B \otimes Z + \Im B \otimes X,
\end{equation}
which is a real linear operator.
\end{lemma}
\begin{proof}

We consider the evolution of an arbitrary state $\ket\phi$ under $\cP(A+BK) = \cP\circ (A+BK) \circ \cL$. Writing $A+BK = C\Re + D\Im$, where $C = A+B$ and $D = i(A-B)$, we get
\begin{eqnarray}
\ket \phi & \xrightarrow{\cL} & \bra 0_a \cdot \ket \phi + i \bra 1_a \cdot \ket \phi \nn
& \xrightarrow{A+BK} & C \bra 0_a \cdot \ket \phi + D  \bra 1_a \cdot \ket\phi \nn
& \xrightarrow{\cP} & (\Re C \bra 0_a \cdot \ket \phi + \Re D \bra 1_a \cdot \ket\phi) \otimes \ket 0_a + 
(\Im C \bra 0_a \cdot \ket \phi + \Im D \bra 1_a \cdot \ket\phi) \otimes \ket 1_a \nn
&&= ( \Re C \otimes \ketbra 00 + \Im C \otimes \ketbra 10 + \Re D \otimes \ketbra 01 + \Im D \otimes \ketbra 11) \ket\phi .
\end{eqnarray}

Hence,
\begin{eqnarray}
\cP(A+BK) &=& \Re (A+B) \otimes \ketbra 00 + \Im (A+B) \otimes \ketbra 10 \nn &&+ \Re (i(A-B)) \otimes \ketbra 01 + \Im (i(A-B)) \otimes \ketbra 11 \nn
&=& \Re A \otimes (\ketbra 00 + \ketbra 11) + \Re B \otimes (\ketbra 00 - \ketbra 11) \nn &&+ \Im A \otimes (\ketbra 10 - \ketbra 01) + \Im B \otimes (\ketbra 10 + \ketbra 01) \nn
&=& \Re A \otimes I + \Im A \otimes XZ + \Re B \otimes Z + \Im B \otimes X .
\end{eqnarray}
\end{proof}
\noindent Note that \eq{eq:PABK} is a generalization of Eq.\ (9) of \cite{fernandez2003quaternionic}. Indeed, when $B=0$, we obtain an expression for the rebit encoding $\mathcal P$ of the linear operator $A$ in terms of its real and imaginary parts\footnote{See Appendix \ref{app:rebitEncodingLinear} for equivalent expressions of $\cP(A)$.}:
\begin{equation} \label{eq:rebitEncodingLinear}
\mathcal{P}(A) = \Re A \otimes I + \Im A \otimes XZ.
\end{equation}

Lemmas \ref{lem:imageWunderL} and \ref{lem:imageABunderK} tell us that an operator $\Gamma$ is $\bbR$-linear if and only if it can be written in the form $\Gamma = A+BK$, where $A$ and $B$ are complex linear operators. It turns out that operators of the form $A+BK$ have the following alternative characterization (see Appendix \ref{app:rlinear}):

\begin{theorem} \label{thm:Rlinearcharacterization}
Let $V$ and $V'$ be complex vector spaces, and $f:V\rightarrow V'$ be a function on $V$. Then,
there exist linear maps $A$ and $B$ such that $f = A+BK$ if and only if \begin{equation}
f(ax+ by) = a f(x) + bf(y)
\end{equation} 
for all $a,b \in \mathbb R$ and $x,y\in V$. 
\end{theorem}
\begin{proof}
See Theorem \ref{thm:RlinearcharacterizationRestated} of Appendix \ref{app:rlinear}.
\end{proof}

Hence, we obtain the following equivalent notions of $\bbR$-linearity.

\begin{theorem}\label{thm:equiv_defRlinear}
Let $\Gamma: \cH_n(\bbC) \rightarrow \cH_n(\bbC)$ be a $n$-qubit operator. Then the following three statements are equivalent.
\begin{enumerate}
\item $\Gamma \in \bbR L_n$, i.e.~$\Gamma$ is $\bbR$-linear.
\item There exists complex linear operators $A$ and $B$ such that $\Gamma = A+BK$.
\item $\Gamma (a \ket\psi + b \ket\phi) = a \Gamma(\ket\psi) + b \Gamma(\ket\phi)$ for all $a,b \in \bbR$ and $\ket\psi, \ket\phi \in \cH_n(\bbC)$.
\end{enumerate}
\end{theorem}

Before moving on, we address the notion of an operator norm for $\bbR$-linear operators. We start by recalling the operator norm for linear operators.
\begin{definition}\label{defn:linear_norm}
For $\Gamma\in L_n$, the operator norm, denoted $\|\Gamma\|$ is defined as
\begin{equation}
\|\Gamma\|=\text{inf}\{\epsilon\ge0:\|\Gamma\ket{\psi}\|\le\epsilon\|\ket{\psi}\|,\forall\ket{\psi}\in\mathcal{H}_{n}(\bbC)\}.
\end{equation}
\end{definition}
\noindent Equivalently, this is the largest singular value of $\Gamma$. Recall $\|\ket{\psi}\|$ is the $l_2$ norm from Proposition~\ref{prop:l2_norm_equal}.

For $\bbR$-linear operators, we advocate a norm built from the norm for linear operators and the encoding map $\mathcal{P}$.
\begin{definition}\label{defn:rlinear_norm}
For $\Gamma\in\bbR L_n$, let $\|\Gamma\|=\|\mathcal{P}(\Gamma)\|$ be the operator norm.
\end{definition}
\noindent This norm satisfies the triangle inequality due to the $\bbR$-linearity of $\mathcal{P}$ from Proposition~\ref{prop:pop_rlinear_lop_linear}. To justify using the same notation in Definitions~\ref{defn:linear_norm} and \ref{defn:rlinear_norm}, we do have to check that they coincide when $\Gamma\in L_n$. We verify this in Proposition~\ref{prop:norms_coincide} in Appendix~\ref{app:equivalenceNorm}. With the operator norm on $\bbR$-linears, we also gain the ability to define $\epsilon$-approximations, and we follow \cite{shi2002both} in the inclusion of an ancilla system.
\begin{definition}
For $\tilde n\ge n$, we say $\tilde \Gamma\in\bbR L_{\tilde n}$ $\epsilon$-approximates $\Gamma\in\bbR L_n$ with ancilla $\ket{\psi}\in\mathcal{H}_{\tilde n-n}(\bbC)$ if
\begin{equation}
\|\bra{\psi}\tilde \Gamma\ket{\psi}-\Gamma\|=\|\mathcal{P}(\bra{\psi}\tilde \Gamma\ket{\psi})-\mathcal{P}(\Gamma)\|\le\epsilon,
\end{equation}
where the $\bbR$-linearity of $\mathcal{P}$ (Proposition~\ref{prop:pop_rlinear_lop_linear}) justifies the equality.
\end{definition}

\subsubsection{$\bbR$-unitary operators}
\label{sec:Runitaryoperators}

Consider the set $U_n$ of unitary operators on $n$-qubits. Then $U^{\bbR}_n := U_n\cap R_n = T_n$ is the set of real unitary (i.e. orthogonal) operators.
To find necessary and sufficient conditions for a operator to be $\bbR$-unitary, we compute the image of the set of orthogonal operators under $\mathcal L$:
\begin{theorem} \label{thm:OrthogonalABK}
Let $A+BK \in \bbR L_n$ be an $\mathbb R$-linear operator. Then $\mathcal P(A+BK)\in T_{n+1}$ if and only if
\begin{eqnarray} \label{eq:OrthogonalABK}
A^\dag A + B^T \bar B = I \nn
A^\dag B + B^T \bar A = 0 .
\end{eqnarray}
\end{theorem}
\begin{proof}
$\mathcal P(A+BK) = \Re A \otimes I + \Im A \otimes XZ + \Re B \otimes Z + \Im B \otimes X $ is orthogonal if and only if
\begin{eqnarray}
I \otimes I &=& (\Re A \otimes I + \Im A \otimes XZ + \Re B \otimes Z + \Im B \otimes X )^T \nn
&& \qquad (\Re A \otimes I + \Im A \otimes XZ + \Re B \otimes Z + \Im B \otimes X ) \nn
&=& (\Re A^T \Re A + \Im A^T \Im A + \Re B^T \Re B + \Im B^T \Im B) \otimes I \nn
&&+ (\Re A^T \Im B - \Im A^T \Re B - \Re B^T \Im A + \Im B^T \Re A) \otimes X \nn
&&+ (\Re A^T \Im A - \Im A^T \Re A - \Re B^T \Im B + \Im B^T \Im B) \otimes XZ \nn
&&+ (\Re A^T \Re B + \Im A^T \Im B + \Re B^T \Re A + \Im B^T \Im A) \otimes Z \nn
&=& (\Re (A^\dag A) + \Re(B^\dag B)) \otimes I + (\Im(A^\dag B) - \Im(B^\dag A))\otimes X \nn
&&+ (\Im(A^\dag A) - \Im(B^\dag B))\otimes XZ + (\Re (A^\dag B)+ \Re(B^\dag A) )\otimes Z
\end{eqnarray}
which holds if and only if
\begin{eqnarray}
\Re (A^\dag A) + \Re(B^\dag B) = I\\
\Im(A^\dag B) - \Im(B^\dag A) = 0 \\
\Im(A^\dag A) - \Im(B^\dag B) = 0 \\
\Re (A^\dag B)+ \Re(B^\dag A) = 0
\end{eqnarray}
which holds if and only if
\begin{eqnarray} 
A^\dag A + B^T \bar B = I \nn
A^\dag B + B^T \bar A = 0 .
\end{eqnarray}
\end{proof}

In Appendix \ref{app:algebraicProperties} (see Proposition \ref{prop:unitaryElement}), we show that an $\bbR$-linear operator satisfies \eq{eq:OrthogonalABK} if and only if it is a unitary element of the $\bbR$-linear group with respect to the dagger operator $\dag$ defined by
\begin{equation} \label{eq:dagdef}
(A+BK)^\dag = A^\dag + B^T K.
\end{equation}
This, together with Theorem~\ref{thm:OrthogonalABK}, implies the following theorem. 
\begin{theorem} \label{thm:OrthogonalABK2}
Let $A+BK \in \bbR L_n$ be an $\mathbb R$-linear operator. Then $A+BK$ is $\bbR$-unitary  if and only if
$A+BK$ is a unitary element with respect to the $\dag$ operator defined by \eq{eq:dagdef}.
\end{theorem}

Since the unitary operators are a proper subgroup of the linear operators, Proposition \ref{prop:subgroupResult} implies that the $\bbR$-unitaries are a proper subgroup of the $\bbR$-linear operators. In particular, there exist $\bbR$-linear operators, like $I+K$, which are not $\bbR$-unitary. Note that since Eqs.~\eqref{eq:OrthogonalABK} are not scale invariant, any $\bbR$-unitary $\Gamma$ gives rise to a family of non-$\bbR$-unitary
elements $c\Gamma$, for $c \in \bbC$, $|c| \neq 1$.

We note that in the proof of Theorem \ref{thm:OrthogonalABK}, we used the property that a matrix $W$ is orthogonal if and only if $W^T W = I$. But this is equivalent to the condition that $W W^T = I$. 
In Appendix \ref{app:altFormulation}, we present an alternative formulation of Theorem \ref{thm:OrthogonalABK} that uses the latter criterion for orthogonality.

Finally, we notice that the group of $\bbR$-unitary operators has a center $\mathcal{Z}(\bbR U_n)=\{\pm I\}=\mathbb{R}\mathcal{Z}(U_n)$ in contrast to the center of the unitary group $\mathcal{Z}(U_n)=\{e^{i\phi}I:\phi\in\mathbb{R}\}$, consisting of all global phases. Since arbitrary global phases do not commute with the rest of the $\bbR$-unitary group, we might expect to actually observe them. After setting up a framework for encoding and decoding measurements in the next section, we show how such a measurement of the global phase works in Section~\ref{sec:bottom_up_tomography}.

\subsection{Rebit encoding and decoding of measurements} \label{sec:enc_and_dec_meas}

The final circuit ingredient is measurement. Recall that measurements of quantum systems may be described by a collection of measurement operators $\{M_m\}$ satisfying \cite{nielsen2010quantum}
\begin{equation}
\label{eq:completenessrelation}
\sum_m M_m^\dag M_m = I ,
\end{equation}
where the probability that the result $m$ occurs when a state $\ket \psi$ is measured is given by 
\begin{equation}
|| M_m\ket\psi ||^2.
\end{equation}

Let $\{M_m\}$ be a collection of measurement operators on $\cH_{n+1}(\bbR)$. 
We assume that each $M_m$ is real, i.e. $\Im M_m = 0$. Let $\{F_m\}$ be a collection of operators on $\cH_n(\bbC)$. We say that a $\{F_m\}$ is a \textit{rebit decoding} of $\{M_m\}$, or that $\{M_m\}$ is a \textit{rebit encoding} of $\{F_m\}$, if for all $\ket \phi \in \cH_n(\bbR)$, we have
\begin{equation}\label{eq:rebit_enc_meas}
||F_m \cL \ket \phi||^2 = ||M_m \ket \phi||^2 .
\end{equation}
The collection $\{F_m \}$ can hence be thought of a set of ``$\bbR$-measurement operators" on the space $\cH_n(\bbC)$. The probability that the result $m$ occurs when a state $\ket \psi \in \cH_n(\bbC)$ is measured is given by
\begin{equation}
||F_m \ket\psi||^2 .
\end{equation}

Our next proposition relates $\{M_m\}$ to a rebit decoding of it.

\begin{prop}\label{prop:rebit_decoding}
$\{\cL(M_m)\}$ is a rebit decoding of $\{M_m\}$.
\end{prop}
\begin{proof}
This follows from noting that
\begin{equation}
||\cL(M_m) \cL \ket \phi||^2 = ||\cL( M_m \ket \phi) ||^2 = ||M_m \ket \phi||^2 .
\end{equation}
\end{proof}

It turns out that $\cL(M_m)$ also obeys the completeness relation given in \eq{eq:completenessrelation}:

\begin{prop}
Let $F_m = \cL(M_m)$. Then, 
$$\sum_m F_m^\dag F_m = I.$$
\end{prop}

\begin{proof}
To prove this, consider
\begin{eqnarray}
\sum_m F_m^\dag F_m &=& \sum_m \cL(M_m)^\dag \cL(M_m) \nn
&=& \sum_m \cL(M_m^T) \cL(M_m),\quad \mbox{by Proposition \ref{prop:LofTranspose}} \nn
&=& \cL\left(\sum_m M_m^T M_m\right) \nn
&=& \cL(I) = I.
\end{eqnarray}
\end{proof}

We now give a top-down example. How is a computational basis measurement on qubits simulated using the rebit encoding? The following proposition tells us that 
this may be done by simply performing a computational basis measurement on the first $n$ rebits in the encoded circuit, and discarding the ancilla rebit. 


\begin{prop} \label{prop:rebitMeasurement}
The rebit encoding of the computational basis measurement described by measurement operators $\set{\ketbra mm}_m$ is given by $\{\ketbra mm \otimes I_a\}_m$.
\end{prop}
\begin{proof}

Substituting $W = \ketbra mm \otimes I_a$ into
\eq{eq:imageWunderL} and using the fact that the Pauli matrices $X$, $Y$ and $Z$ are traceless, we obtain 
\begin{equation}
\cL(\ketbra mm \otimes I_a) = \ketbra mm.
\end{equation}
Using Proposition~\ref{prop:rebit_decoding} completes the proof.
Alternatively, we may verify directly that
\begin{eqnarray*}
||(\ketbra mm \otimes I_a)\mathcal P(\ket\psi)||^2 &=& ||(\bra m \otimes I_a) (\Re \ket \psi \otimes \ket{0}_a + \Im \ket \psi \otimes \ket{1}_a)   ||^2 \\
&=& || \bra m \Re \ket \psi \ket{0}_a + \bra m \Im \ket \psi \ket{1}_a||^2 \\
&=& \bra m \Re \ket \psi^2 + \bra m \Im \ket \psi^2 \\
&=& |\bra m \Re \ket \psi + i \bra m \Im \ket \psi|^2 \\
&=& |\braket m \psi|^2 \\
&=& ||(\ketbra mm) \ket\psi||^2.
\end{eqnarray*}
\end{proof}

We now consider the following bottom-up example.
What does it mean to measure the ancilla qubit in the computational basis? The following proposition reveals that this corresponds to measuring the eigenvalue of the complex conjugation operator $K$, thereby projecting an $n$-qubit quantum state to its real or imaginary part.
\begin{prop}
The rebit decoding of the set of measurement operators $\{(I+Z_a)/2,(I-Z_a)/2\} = \{ \ketbra 00_a , \ketbra 11_a \}$ is
$\{(I+K)/2,(I-K)/2\} = \{\Re, \Im\}$.
\end{prop}

\begin{proof}
By \eq{eq:imageWunderL},
\begin{equation}
\label{eq:LZa}
\mathcal L(Z_a) = \left(\frac 12 \tr_a [Z_a(I_a - iX_aZ_a)]  \right) + \left(\frac 12 \tr_a [Z_a(Z_a + i X_a)]  \right) K = K ,
\end{equation}
since the Pauli matrices $X$, $Y$ and $Z$ are all traceless. Using Proposition~\ref{prop:rebit_decoding} completes the proof.

\end{proof}

\subsection{Bottom-up tomography}\label{sec:bottom_up_tomography}
Given a pure $n+1$ rebit state $\ket{\phi}$, we sketch a procedure to find the amplitudes up to a global $\pm1$ sign using only orthogonal operators and single-rebit measurements in the computational basis. This is the natural form of tomography on the simulation space $P$ and translates through the decoding map $\mathcal{L}$ to a tomographic procedure on $n$-qubit states that makes use of $\mathbb{R}$-unitary operators, computational basis measurements, and projection onto the eigenspaces of complex conjugation $K$.

\begin{theorem}\label{thm:rebit_tomog}
An $n$-rebit density matrix is determined exactly by measurement of the $(4^n+2^n)/2$ observables 
\begin{equation}
\mathcal{O}=\{p=p_1\otimes p_2\otimes\dots\otimes p_{n}:p_j\in\{I,X,Y,Z\},p=\bar p\}.
\end{equation}
Moreover, any given $p\in\mathcal{O}$ can be measured using single-rebit computational basis measurements, along with $H$ and $CZ$ gates.
\end{theorem}
\begin{proof}
Let $\rho=\ket{\phi}\bra{\phi}$ be the real density matrix corresponding to $\ket{\phi}$, a generic $n$-rebit state. Then, because $\rho=\rho^\dag=\bar\rho=\rho^T$, there is a decomposition
\begin{equation}
\rho=\sum_{p\in\mathcal{O}}a_pp
\end{equation}
for some real coefficients $a_p\in\mathbb{R}$. Since $\tr(\rho)=1$ and all $p\neq I^{\otimes n}$ satisfy $\tr(p)=0$, we have $a_{I^{\otimes n}}=2^{-n}$. Evidently then, learning the values of $a_p=\tr(p\rho)$ by repeatedly measuring the observable $p$ effectively learns $\rho$. Given $\rho$, $\ket{\phi}$ is determined up to a sign $\pm1$.

The size of $\mathcal{O}$ can be calculated by noticing that any $p\in\mathcal{O}$ must have an even number of $Y$s in its tensor product due to the reality condition $p=\bar p$. Thus,
\begin{equation}
|\mathcal{O}|=3^n+3^{n-2}\binom{n}{2}+3^{n-4}\binom{n}{4}+\dots+3^{n-2\lfloor n/2\rfloor}\binom{n}{2\lfloor n/2\rfloor}=(4^n+2^n)/2.
\end{equation}
Notice that this is strictly greater than the number of rebit measurements that measure a subset of rebits each in the $X$-basis and another disjoint subset each in the $Z$-basis -- there are $3^n$ such measurements, corresponding to observables in
\begin{equation}
\mathcal{O}'=\{p_1\otimes p_2\otimes\dots\otimes p_n:p_j\in\{I,X,Z\}\}.
\end{equation}
Since $|\mathcal{O}'|=3^n<(4^n+2^n)/2=|\mathcal{O}|$ for all $n>1$, there is no way to satisfy tomographic locality \cite{hardy2011reformulating} in rebit tomography. That is, there is no way to completely learn a rebit state $\ket{\phi}$ using only single-rebit orthogonal gates and computational basis measurements. We must therefore use at least a two-rebit gate, and it turns out this (in particular, the $CZ$ gate) is sufficient, as we show next. Therefore, real quantum mechanics satisfies a slightly looser axiom of tomographic \emph{2-locality}.

To show tomographic 2-locality, we need only show how observables in $\mathcal{O}$ can be measured using the gates $H$ and $CZ$ and single-rebit measurements in the computational basis (i.e.~measurements of $Z$ on a single rebit). Since any element of $\mathcal{O}$ is made up of a tensor product of $X,Z,$ and $Y\otimes Y$, we need only show how to measure these using the components given. The circuits in Fig.~\ref{fig:tomo_2loc} demonstrate this.
\end{proof}

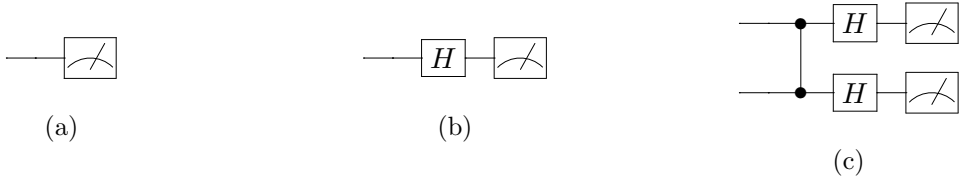
\begin{figure}
\centering
\begin{subfigure}{0.3\textwidth}
\begin{equation*}
\Qcircuit @C=1em @R=1em {
& \qw & \meter  \\
} 
\end{equation*}
\caption{}
\end{subfigure}
\begin{subfigure}{0.3\textwidth}
\begin{equation*}
\Qcircuit @C=1em @R=1em {
& \qw   & \gate H    &  \meter \\
} 
\end{equation*}
\caption{}
\end{subfigure}
\begin{subfigure}{0.3\textwidth}
\begin{equation*}
\Qcircuit @C=1em @R=1em {
& \qw  & \ctrl 1 & \gate H & \meter \\ 
& \qw & \control \qw & \gate H & \meter \\
} 
\end{equation*}
\caption{}
\end{subfigure}

\caption{\label{fig:tomo_2loc} The proof of 2-locality of rebit tomography using $H$, $CZ$, and single-rebit computational basis measurements. (a) Measuring $Z$, (b) Measuring $X$, and (c) Measuring $Y\otimes Y$.}
\end{figure}

Because Theorem~\ref{thm:rebit_tomog} allows us to learn an $(n+1)$-rebit state $\ket{\phi}$ up to a $\pm$ sign, we can also determine $\ket{\psi}=\mathcal{L}(\ket{\phi})$ up to a $\pm$ sign, including the (typically unobservable) global phase of $\ket{\psi}$. Mapped to the simulated space, the operators used for rebit tomography in the proof of Theorem~\ref{thm:rebit_tomog} become $\bbR$-unitaries and measurements of observables $Z$ and $K$ (using the encoding and decoding of measurements results in Section~\ref{sec:enc_and_dec_meas}).

\section{Partial antiunitarity} \label{sec:partialAntiunitarity}
In the previous section, we characterized the set of operators that a bottom-up rebit simulation can simulate as the $\mathbb{R}$-unitary operators. In this section, we pay special attention to a subset of the $\mathbb{R}$-unitaries that we call \textit{partial antiunitary}. These operators are mathematically an interpolation between unitary and antiunitary operators, which are each special cases. Partial antiunitaries, as we show in the Section~\ref{sec:universalgateset}, are also sufficient to densely generate the entire group of $\mathbb{R}$-unitaries. We find it convenient to start by defining the partial complex conjugation operator, which is partial antiunitary itself and is central to the study of the entire set of partial antiunitaries.

\subsection{Partial complex conjugation} 

Let $L\subseteq \{0,1\}^n$ be a language. We define the \textit{partial complex conjugation} operator with respect to $L$ to be $K_L: \mathcal H_n(\bbC) \rightarrow \mathcal H_n(\bbC)$, where
\begin{equation} \label{eq:KLdef}
K_L : \ket\psi \mapsto \sum_{x\notin L} \braket x\psi \ket x + \sum_{x\in L} \braket \psi x \ket x.
\end{equation}
This generalizes the notion of complex conjugation. Indeed, when $L = \{0,1\}^n$, we get $K_L = K$. Note that the identity operator $I$ is another special case of $K_L$: when $L =\emptyset$, we obtain $K_L = I$. Two other examples of $K_L$ that will be important in this paper are the \textit{controlled complex conjugation} operator $CK$ and the \textit{controlled-controlled complex conjugation} operator $CCK$, which are defined as follows:
The $CK$ operator on the $j$th register, denoted $CK_j$, is the partial complex conjugation operator with $L = \{x\in \{0,1\}^n|x_j = 1\}$, i.e.,
\begin{equation}
CK_j : \ket\psi \mapsto \sum_{x:x_j=0} \braket x\psi \ket x + \sum_{x:x_j=1} \braket \psi x \ket x,
\end{equation}
and the CCK operator on the $i$th and $j$th registers, where $i \neq j$, denoted $CCK_{ij}$ is the partial complex conjugation operator with $L = \{x\in \{0,1\}^n|x_i = x_j = 1\} = \{x\in \{0,1\}^n|x_i x_j = 1\}$, i.e.,
\begin{equation}
CCK_{ij} : \ket\psi \mapsto \sum_{x: x_i x_j=0} \braket x\psi \ket x + \sum_{x:x_i x_j = 1} \braket \psi x \ket x.
\end{equation}

We can also express $K_L$ in terms of orthogonal projections (see Appendix \ref{app:OrthogonalProjections}). By denoting 
\begin{eqnarray} \label{def:SL}
\bbH(L) = \span\set{\ket x: x\in L}, \quad \bbH^\perp(L) = \span\set{\ket x: x\notin L},
\end{eqnarray}
and
\begin{equation}
\Pi_L=\proj_{\bbH(L)}=\sum_{x\in L} \ketbra xx, \quad \Theta_L =\proj_{\bbH^\perp(L)}=\sum_{x\notin L} \ketbra xx,
\end{equation}
we may express $K_L$ as follows:

\begin{prop} \label{prop:KL}
\begin{equation} \label{eq:KL}
K_L = \Theta_L + K \Pi_L = \Theta_L + \Pi_L K .
\end{equation}
\end{prop}

\begin{proof}
First, we note that since $\Pi_L$ is a matrix with only real entries, $K\Pi_L = \Pi_L K$. Now,
\begin{eqnarray}
(\Theta_L + K \Pi_L) \ket\psi &=& 
\sum_{x\not\in L} \ket x \braket x \psi + K \sum_{x\in L} \ket x \braket x \psi \nn
&=& \sum_{x\notin L} \braket x\psi \ket x + \sum_{x\in L} \braket \psi x \ket x = K_L\ket\psi,
\end{eqnarray}
which completes the proof.
\end{proof}

Its image under $\cP$ is given by

\begin{theorem}
\label{thm:PKLeqn}
\begin{equation} \label{eq:PKLeq}
\cP (K_L) = \Theta_L \otimes I_a + \Pi_L \otimes Z_a .
\end{equation}
\end{theorem}
\begin{proof}
By using \eq{eq:PABK},
\begin{eqnarray}
\mathcal P(K_L) &=& \Re \Theta_L \otimes I_a + \Im \Theta_L \otimes X_aZ_a + \Re \Pi_L \otimes Z_a + \Im \Pi_L \otimes X_a \nn
&=& \Theta_L \otimes I_a + \Pi_L \otimes Z_a .
\end{eqnarray}
\end{proof}

\begin{remark}
\label{rem:KLunitary}
An immediate consequence of \eq{eq:PKLeq} is that $K_L$ is not just $\bbR$-linear but also $\bbR$-unitary, as the RHS of \eq{eq:PKLeq} is orthogonal. 
\end{remark}


\subsection{Partial antiunitary operators} 

Let $\mathcal H$ be a complex vector space with inner product $\langle \cdot, \cdot \rangle$. A \textit{unitary} operator on $\mathcal H$ is a function $f:\mathcal H\rightarrow \mathcal H$ for which 
$$\langle f(x),f(y)\rangle = \langle x,y \rangle, \quad \mbox{ for all } x,y\in \mathcal H.$$ An \textit{antiunitary} operator on $V$ is a function $g:\mathcal H\rightarrow \mathcal H$ for which $$\langle g(x),g(y)\rangle = \langle y,x \rangle, \quad \mbox{ for all } x,y\in \mathcal H.$$ Note that a consequence of these definitions is that unitary operators are linear and antiunitary operators are antilinear\footnote{Let $\mathcal H$ be a complex vector space. A function $\mathcal H\rightarrow \mathcal H$ is an antilinear operator if 
$$f(ax+by) = \bar a f(x) +\bar b f(y)$$ 
for all $x,y\in \mathcal H$ and $a,b\in \mathbb C$. To show that unitarity implies linearity, consider the expression $$||f(ax + by) - a f(x) -b f(y)||^2 = \langle f(ax + by) - a f(x) -b f(y),f(ax + by) - a f(x) -b f(y) \rangle$$ and show that it is equal to zero. By definition of a norm, this implies that $f(ax + by) = a f(x) +b f(y)$. The proof that antiunitarity implies antilinearity proceeds similarly.}.


We shall now generalize the above definitions of unitarity and antiunitarity. Let $\Xi\subseteq \mathcal H$ be a subspace of $\mathcal H$. Let $f,g :\mathcal H \rightarrow \mathcal H$ be operators on $\mathcal H$. We say that $f$ is \textit{unitary on} $\Xi$ if 
$$\langle f(x),f(y)\rangle = \langle x,y \rangle, \quad \mbox{ for all } x,y\in \Xi.$$ We say that $g$ is \textit{antiunitary on} $\Xi$ if $$\langle g(x),g(y)\rangle = \langle y,x \rangle, \quad \mbox{ for all } x,y\in \Xi.$$ Hence, a unitary (antiunitary) operator is one that is unitary (antiunitary) on $\mathcal H$. We may now define a partial antiunitary operator\footnote{The term \textit{partial antiunitary} appears on Page 2 of \cite{mckague2009simulating}, where it was used as an example of an operator that can be simulated by the rebit encoding. However, an exact definition was not given in \cite{mckague2009simulating}. Also, \cite{mckague2009simulating} does not address whether the rebit encoding allows for the simulation of more than just products of partial antiunitaries. In this paper, we give a precise definition of partial antiunitarity and address the above question.} as follows:
\begin{definition} \label{def1:partialAntiunitary}
Let $\Xi \subseteq \mathcal H$ be a subspace of a complex (finite-dimensional) vector space $\mathcal H$. An operator $\Gamma$ on $\mathcal H$ is a \textit{partial antiunitary} operator with respect to $\Xi$ if 
\begin{enumerate}[(i)]
\item $\Gamma$ is additive, i.e.\ $\Gamma(\psi+\phi) = \Gamma(\psi)+\Gamma(\phi)$ for all $\psi,\phi \in \mathcal H$.
\item $\Gamma$ is unitary on $\Xi^\perp$, i.e.\ $\langle \Gamma(\psi),\Gamma(\phi)\rangle = \langle \psi,\phi\rangle$ for all $\psi, \phi \in \Xi^\perp$.
\item $\Gamma$ is antiunitary on $\Xi$, i.e.\ $\langle \Gamma(\psi),\Gamma(\phi)\rangle = \langle \phi,\psi\rangle$ for all $\psi, \phi \in \Xi$.
\item $\langle \Gamma(\psi),\Gamma(\phi)\rangle = 0$ for all $\psi \in \Xi$ and $\phi \in \Xi^\perp$.	
\end{enumerate}
\end{definition}
The partial antiunitary operators include unitary and antiunitary operators as special cases: a unitary operator on $\mathcal H$ is a partial antiunitary operator with respect to $\set{0}$, and an antiunitary operator on $\mathcal H$ is a partial antiunitary operator with respect to $\mathcal H$.

The above definition of partial antiunitarity could also be phrased in terms of orthogonal projections:
\begin{prop} \label{def2:partialAntiunitary}
An operator $\Gamma$ on $\mathcal H$ is a partial antiunitary operator with respect to a subspace $\Xi \subseteq \cH$ if and only if for all $\psi,\phi\in \mathcal H$,
\begin{enumerate}[(i')]
\item $\Gamma(\psi+\phi) = \Gamma(\psi)+\Gamma(\phi)$.
\item $\langle \Gamma\Theta\psi,\Gamma\Theta\phi \rangle = \langle \Theta\psi, \Theta\phi\rangle$.
\item $\langle \Gamma\Pi\psi,\Gamma\Pi\phi\rangle = \langle \Pi\phi,\Pi\psi\rangle$.
\item $\langle \Gamma\Pi\psi,\Gamma\Theta\phi\rangle = 0$.	
\end{enumerate}
where $\Pi=\proj_{\Xi}$ and $\Theta=\proj_{\Xi^\perp}$.
\end{prop}
\begin{proof}
Equivalence holds since $\psi = \Theta \psi$ and $\phi = \Theta \phi$ for all $\psi, \phi \in \Xi^\perp$, and $\psi = \Pi \psi$ and $\phi = \Pi \phi$ for all $\psi, \phi \in \Xi$.
\end{proof}

We will now give a third characterization of partial antiunitary operators.
\begin{theorem} \label{thm3:partialAntiunitary}
An operator $\Gamma$ on $\mathcal H$ is a partial antiunitary operator with respect to a subspace $\Xi \subseteq \cH$ if and only if for all $\psi,\phi\in \mathcal H$,
\begin{equation} \label{def3:partialAntiunitary}
\langle \Gamma\psi,\Gamma\phi\rangle = \langle \Theta\psi,\Theta\phi\rangle + \langle \Pi\phi,\Pi\psi\rangle,
\end{equation}
where $\Pi=\proj_{\Xi}$ and $\Theta=\proj_{\Xi^\perp}$.
\end{theorem}
\begin{proof} \hfill
\begin{itemize}
\item[$(\implies)$] Assume that $\Gamma$ is a parital antiunitary with respect to $\Xi$. Let $\psi,\phi\in \mathcal H$. Then we could write $\psi = \Pi\psi + \Theta\psi$ and $\phi = \Pi\phi + \Theta\phi$. This implies that
\begin{eqnarray*}
\langle \Gamma\psi,\Gamma\phi\rangle &=& \langle \Gamma(\Pi\psi + \Theta\psi), \Gamma( \Pi\phi + \Theta\phi)\rangle \\
&=& \langle \Gamma\Pi\psi + \Gamma\Theta\psi, \Gamma \Pi\phi + \Gamma\Theta\phi\rangle, \qquad\mbox{by additivity} \\
&=& \langle \Gamma\Pi\psi, \Gamma\Pi\phi \rangle + \langle \Gamma\Pi\psi, \Gamma\Theta\phi \rangle + \langle \Gamma\Theta\psi, \Gamma\Pi\phi \rangle + \langle \Gamma\Theta\psi, \Gamma\Theta\phi \rangle  \\
&=&  \langle \Pi\phi,\Pi\psi\rangle + 0+ 0 + \langle \Theta\psi,\Theta\phi\rangle, \qquad\mbox{by (ii'), (iii'), (iv') of Proposition \ref{def2:partialAntiunitary}}\\
&=& \langle \Theta\psi,\Theta\phi\rangle + \langle \Pi\phi,\Pi\psi\rangle.
\end{eqnarray*}
\item[$(\impliedby)$] Assume that \eq{def3:partialAntiunitary} holds.
\begin{enumerate}[(i')]
\item We first show that \eq{def3:partialAntiunitary} implies that $\Gamma$ is additive. Consider the expression
\begin{eqnarray*}
&&||\Gamma(\psi+\phi)-\Gamma(\psi)-\Gamma(\phi)||^2 \\ &=& \langle \Gamma(\psi+\phi)-\Gamma(\psi)-\Gamma(\phi), \Gamma(\psi+\phi)-\Gamma(\psi)-\Gamma(\phi)\rangle \\
&=& \langle \Gamma(\psi+\phi),\Gamma(\psi+\phi)\rangle + \langle \Gamma(\psi),\Gamma(\psi)\rangle + \langle \Gamma(\phi),\Gamma(\phi)\rangle \\ &&+ [- \langle \Gamma(\psi+\phi),\Gamma(\psi)\rangle - \langle \Gamma(\psi+\phi),\Gamma(\phi)\rangle + \langle \Gamma(\psi),\Gamma(\phi)\rangle + c.c.] \\
&=& \langle \Theta(\psi+\phi),\Theta(\psi+\phi)\rangle + \langle \Theta(\psi),\Theta(\psi)\rangle + \langle \Theta(\phi),\Theta(\phi)\rangle \\ &&+ [- \langle \Theta(\psi+\phi),\Theta(\psi)\rangle - \langle \Theta(\psi+\phi),\Theta(\phi)\rangle + \langle \Theta(\psi),\Theta(\phi)\rangle + c.c.] + (\Theta\leftrightarrow\Pi) \\
&=& \langle \Theta(\psi+\phi)-\Theta(\psi)-\Theta(\phi), \Theta(\psi+\phi)-\Theta(\psi)-\Theta(\phi)\rangle + (\Theta\leftrightarrow\Pi) \\  &=& 0,
\end{eqnarray*}
where $c.c.$ stands for complex conjugate, and where the last line follows because both $\Pi$ and $\Theta$ are linear. By the properties of a norm, $\Gamma(\psi+\phi)=\Gamma(\psi)+\Gamma(\phi)$.
\item $$\langle \Gamma\Theta\psi, \Gamma\Theta\phi\rangle = \langle \Theta^2 \psi, \Theta^2 \phi \rangle + \langle\Pi\Theta\phi, \Pi\Theta \psi\rangle = \langle \Theta \psi, \Theta \phi \rangle.$$
\item $$\langle \Gamma\Pi\psi, \Gamma\Pi\phi\rangle = \langle\Theta\Pi \psi, \Theta\Pi \phi \rangle + \langle \Pi^2\phi, \Pi^2 \psi\rangle = \langle \Pi \phi, \Pi \psi \rangle.$$
\item $$\langle \Gamma\Pi\psi, \Gamma\Theta\phi\rangle = \langle \Theta\Pi \psi, \Theta^2 \phi \rangle + \langle \Pi\Theta\phi, \Pi^2 \psi\rangle = 0.$$
\end{enumerate}
Note that we used the facts that $\Pi^2 = \Pi$, $\Theta^2 = \Theta$ and $\Theta\Pi = \Pi\Theta = 0$.
\end{itemize}
\end{proof}
Note that we could have defined a notion of \textit{partial unitarity} analogously. It would then follow that an operator $\Gamma$ is a partial unitary with respect to $\Xi$ if and only if $\Gamma$ is a partial antiunitary with respect to $\Xi^\perp$, since $(\Xi^\perp)^\perp = \Xi$ for finite-dimensional vector spaces. Since we could easily relate the two notions, we will (somewhat arbitrarily) choose to phrase all subsequent results in terms of partial antiunitaries.

We now give an example of a partial antiunitary operator that is neither unitary nor antiunitary in general: 
\begin{prop} \label{prop:KLispartial}
$K_L$ is a partial antiunitary operator with respect to $\bbH(L)$.
\end{prop}
\begin{proof}
We shall make use of Theorem \ref{thm3:partialAntiunitary}. Then,
\begin{eqnarray*}
\langle K_L \psi, K_L \phi \rangle &=& \langle (\Theta_L+K\Pi_L)\psi, (\Theta_L+K\Pi_L)\phi \rangle, \qquad\mbox{by \eq{eq:KL}} \\ &=& \langle \Theta_L\psi, \Theta_L\phi\rangle + \langle K\Pi_L\psi, \Theta_L\phi\rangle + \langle \Theta_L\psi, K\Pi_L\phi\rangle + \langle K\Pi_L\psi, K\Pi_L\phi\rangle.
\end{eqnarray*}
But $\langle K\Pi_L\psi, \Theta_L\phi\rangle = \langle \Pi_L K\psi, \Theta_L\phi\rangle = \langle \Pi_L^\dag \Pi_L K\psi, \Theta_L\phi\rangle = \langle \Pi_L K\psi, \Pi_L\Theta_L\phi\rangle = 0$, where we used the identities $\Pi_L K = K\Pi_L$, $\Pi_L = \Pi_L^\dag \Pi_L$ and $\Pi_L\Theta_L= 0$, as well as the definition of the adjoint. Likewise, $\langle \Theta_L\psi, K\Pi_L\phi\rangle = 0$. Finally, $\langle K\Pi_L\psi, K\Pi_L\phi\rangle = K(\langle \Pi_L\psi, \Pi_L\phi\rangle) = \langle \Pi_L\phi, \Pi_L\psi\rangle$. 
Hence, $\langle K_L \psi, K_L \phi \rangle = \langle \Theta_L\psi, \Theta_L\phi \rangle + \langle \Pi_L\phi,\Pi_L\psi\rangle$, which means that $K_L$ is a partial antiunitary operator with respect to $\bbH(L)$.
\end{proof}

Next, we show that multiplying a partial antiunitary with respect to $\bbH(L)$ with $K_L$ produces a unitary operator.
\begin{prop} \label{prop:KLwithPartialUnitary}
Let $\Gamma$ be a partial antiunitary (on $\mathcal H$) with respect to $\bbH(L)$. Then $\Gamma K_L$ is a unitary operator on $\mathcal H$.
\end{prop}
\begin{proof}
Assume that $\Gamma$ is a partial antiunitary with respect to $\bbH(L)$. Then for all $\psi,\phi \in \mathcal H$,  
\begin{eqnarray*}
\langle \Gamma K_L \psi, \Gamma K_L \phi \rangle &=& \langle \Theta_L K_L \psi, \Theta_L K_L \psi\rangle +\langle \Pi_L K_L \phi, \Pi_L K_L \phi\rangle \\
&=& \langle \Theta_L(\Theta_L+\Pi_L K) \psi, \Theta_L(\Theta_L+\Pi_L K) \psi\rangle +\langle \Pi_L(\Theta_L+\Pi_L K) \phi, \Pi_L(\Theta_L+\Pi_L K) \phi\rangle \\
&=& \langle \Theta_L\psi, \Theta_L \phi \rangle + \langle \Pi_L K\phi, \Pi_L K \psi\rangle.
\end{eqnarray*}

But the latter term in the sum is equal to $\langle \Pi_L K\phi, \Pi_L K \psi\rangle = \langle K\Pi_L\phi, K\Pi_L \psi\rangle = K(\langle \Pi_L\phi, \Pi_L \psi\rangle) = \langle \Pi_L\psi, \Pi_L\phi\rangle$.
Hence, 
\begin{eqnarray*}
\langle \Gamma K_L \psi, \Gamma K_L \phi \rangle &=& \langle \Theta_L \psi ,\Theta_L \phi\rangle +\langle \Pi_L \psi ,\Pi_L \phi\rangle  \\
&=& \langle \Theta_L \psi ,\Theta_L\phi\rangle +\langle \Theta_L \psi ,\Pi_L\phi\rangle+\langle \Pi_L\psi ,\Theta_L\phi\rangle +\langle \Pi_L\psi ,\Pi_L\phi\rangle \\
&=&\langle \Theta_L\psi + \Pi_L\psi ,\Theta_L\phi + \Pi_L\phi \rangle = \langle \psi,\phi \rangle.
\end{eqnarray*}
Since this holds for all $\phi, \psi \in \mathcal H$, $\Gamma K_L$ is unitary.
\end{proof}

While the product\footnote{We define the product of two partial antiunitary operators $A$ and $B$ to be their composition, i.e. $AB := A\circ B$.} of two partial antiunitary operators may not be partial antiunitary (we show this later in Theorem~\ref{thm:partial_antiunitaries_not_group}), the product of a partial antiunitary operator with either a unitary or antiunitary operator is always partial antiunitary, as the following theorem shows.

\begin{theorem}
\label{thm:prodPartialAntiunitary}
Let $\Xi \subseteq \mathcal H$ be a subspace of a complex vector space $\mathcal H$. Let $U$ be a unitary operator on $\mathcal H$ and $V$ be an antiunitary operator on $\mathcal H$. Then the following statements are equivalent.\footnote{Here, we use $U^\dag (\Xi)$ to denote the set $\set{U^\dag \phi |\phi \in \Xi}$.}
\begin{enumerate}[(I)]
\item $\Gamma$ is a partial antiunitary with respect to $\Xi$.
\item $\Gamma U$ is a partial antiunitary with respect to $U^\dag(\Xi)$.
\item $U \Gamma$ is a partial antiunitary with respect to $\Xi$.
\item $\Gamma V$ is a partial antiunitary with respect to $(V^\dag(\Xi))^\perp$.
\item $V\Gamma$ is a partial antiunitary with respect to $\Xi^\perp$.
\end{enumerate}
\end{theorem}
\begin{proof}
We first show that if $U$ is either unitary or antiunitary, then 
\begin{equation} \label{eq:commutprojU}
\proj_\Xi U = U \proj_{U^\dag(\Xi)}.
\end{equation}
To see this, let $\set{\ket{a_i}}_{i=1}^s$ be a basis for $\Xi$. Then $\set{U^\dag\ket{a_i}}_{i=1}^s$ is a basis for $U^\dag(\Xi)$. Hence, by Proposition \ref{prop:orthogonalprojection} in Appendix \ref{app:OrthogonalProjections},
\begin{eqnarray*}
\proj_\Xi U &=& \sum_{i=1}^s \ketbra{a_i}{a_i} U \\
&=& U \sum_{i=1}^s U^\dag \ketbra{a_i}{a_i} U \\
&=& U \proj_{U^\dag (\Xi)}.
\end{eqnarray*}

Now, (I) holds if and only if for all $\psi,\phi \in \mathcal H$,
\begin{equation}
\label{eq:Gammastar}
\langle \Gamma\psi,\Gamma\phi\rangle = \langle \proj_{\Xi^\perp}\psi,\proj_{\Xi^\perp} \phi\rangle + \langle \proj_\Xi \phi, \proj_\Xi \psi\rangle. \qquad 
\end{equation}

Since
\begin{eqnarray*}
(*) \iff  \langle \Gamma U\psi,\Gamma U\phi\rangle &=& \langle \proj_{\Xi^\perp}U\psi,\proj_{\Xi^\perp} U\phi\rangle + \langle \proj_\Xi U\phi, \proj_\Xi U\psi\rangle \\ &=& \langle U \proj_{U^\dag(\Xi^\perp)}\psi, U\proj_{U^\dag(\Xi^\perp)} \phi\rangle + \langle U\proj_{U^\dag(\Xi)} \phi,  U\proj_{U^\dag(\Xi)} \psi\rangle \\ 
&=& \langle  \proj_{U^\dag(\Xi)^\perp}\psi, \proj_{U^\dag(\Xi)^\perp} \phi\rangle + \langle \proj_{U^\dag(\Xi)} \phi,  \proj_{U^\dag(\Xi)} \psi\rangle.
\end{eqnarray*}
Therefore, (I) is equivalent to (II).

Since $\langle U\Gamma\psi,U\Gamma \phi\rangle = \langle V\Gamma \phi, V\Gamma\psi\rangle = \langle \Gamma\psi,\Gamma\phi\rangle$, (I) is equivalent to (III) and (V), by using \eq{eq:Gammastar}.

Finally, since
\begin{eqnarray*}
\mathrm{Eq.}\, \eqref{eq:Gammastar} \iff  \langle \Gamma V\psi,\Gamma V\phi\rangle &=& \langle \proj_{\Xi^\perp}V\psi,\proj_{\Xi^\perp} V\phi\rangle + \langle \proj_\Xi V\phi, \proj_\Xi V\psi\rangle \\ &=& \langle V \proj_{V^\dag(\Xi^\perp)}\psi, V\proj_{V^\dag(\Xi^\perp)} \phi\rangle + \langle V\proj_{V^\dag(\Xi)} \phi,  V\proj_{V^\dag(\Xi)} \psi\rangle \\ 
&=& \langle  \proj_{V^\dag(\Xi)^\perp}\phi, \proj_{V^\dag(\Xi)^\perp} \psi\rangle + \langle \proj_{V^\dag(\Xi)} \psi,  \proj_{V^\dag(\Xi)} \phi\rangle,
\end{eqnarray*} 
(I) is equivalent to (IV).
\end{proof}

We now use Theorem \ref{thm:prodPartialAntiunitary} to prove the following corollary. 
\begin{corollary} \label{cor:partialExistenceUnitary}
Let $\Xi$ be a subspace of the $n$-qubit Hilbert space $\mathcal H_n(\bbC)$. If $\Gamma$ is a partial antiunitary with respect to $\Xi$, then there exists a language $L\subseteq \set{0,1}^n$, with $|L| = \dim \Xi$, and a unitary operator $U$ mapping $\bbH(L)$ to $\Xi$, such that $\Gamma U$ is a partial antiunitary with respect to $\bbH(L)$.
\end{corollary}
\begin{proof}
We are given that $\bbH(L) = U^\dag(\Xi)$. By (I)$\iff$(II) of Theorem~\ref{thm:prodPartialAntiunitary}, $\Gamma U$ is a partial antiunitary with respect to $U^\dag(\Xi) = \bbH(L)$ since $\Gamma$ is a partial antiunitary with respect to $\Xi$.
\end{proof}

We are now ready to combine the result above from Corollary~\ref{cor:partialExistenceUnitary} with Proposition~\ref{prop:KLwithPartialUnitary} to show that any partial antiunitary operator can be written as a product of $K_L$ with unitary operators.
\begin{theorem} \label{thm:main1}
Let $\Xi$ be a subspace of the $n$-qubit Hilbert space $\mathcal H_n(\bbC)$. If $\Gamma$ is a partial antiunitary operator with respect to $\Xi$, then there exists a language $L\subseteq \set{0,1}^n$, with $|L| = \dim \Xi$, and unitary operators $U$ and $V$, with $V$ mapping $\Xi$ to $\bbH(L)$, such that 
\begin{equation}
\Gamma = U K_L V.
\end{equation}
\end{theorem}

\begin{proof}
By Corollary~\ref{cor:partialExistenceUnitary}, $\Gamma$ being a partial antiunitary with respect to $\Xi$ implies that $\Gamma V^\dag$ is a partial antiunitary with respect to $V(\Xi) = \bbH(L)$. By Proposition~\ref{prop:KLwithPartialUnitary},  $\Gamma V^\dag K_L = U$ for some unitary $U$. Since $K_L$ is its own inverse, this implies that $\Gamma= U K_L V$, which completes the proof of the theorem.
\end{proof}

Theorem \ref{thm:main1} tells us that if we wanted to simulate any arbitrary partial antiunitary operator, it would suffice to just use products of unitary operators and partial complex conjugation. Next, we show in the next two theorems that partial antiunitary operators are a special case of $\bbR$-linear operators, and that not every $\bbR$-linear operator is partial antiunitary. In fact, the partial antiunitaries, unlike the $\bbR$-unitaries, are not a subgroup of the $\bbR$-linear operators.

\begin{theorem}\label{thm:partialAntiunitaryIsUnitaryElement}
If $\Gamma$ is a partial antiunitary operator, then it is $\mathbb{R}$-unitary.
\end{theorem}

\begin{proof}
By Theorem \ref{thm:main1}, we could write $\Gamma = UK_L V$, where $L \subseteq \{0,1\}^n$ and  $U$ and $V$ are unitaries. By Remark \ref{rem:KLunitary}, $K_L$ is $\bbR$-unitary. Since the $\bbR$-unitaries are closed under multiplication, $\Gamma = UK_L V$ is also $\bbR$-unitary.

\end{proof}



Our characterizations of partial antiunitary operators thus far (Definition \ref{def1:partialAntiunitary}, 
Proposition \ref{def2:partialAntiunitary}, and Theorem \ref{thm3:partialAntiunitary}) all involve universal quantifiers and do not provide us with an algorithm to decide if a given $\bbR$-linear operator is partial antiunitary. In the following theorem -- our fourth characterization of partial antiunitary operators -- we give necessary and sufficient conditions, which can be checked efficiently, that $A$ and $B$ must satisfy, in order for $\Gamma=A+BK$ to be partial antiunitary.

\begin{theorem}
\label{thm:characterizationfourth}
An operator $\Gamma$ on $\mathcal H$ is a partial antiunitary operator with respect to a subspace $\Xi \subseteq \cH$ if and only if $\Gamma = A+BK$ for complex operators $A$ and $B$ satisfying
\begin{eqnarray} \label{eq:fourthCharac1}
A^\dag B &=& 0 \\
\label{eq:fourthCharac2}
A^\dag A &=& \Theta \\
\label{eq:fourthCharac3}
B^\dag B &=& \bar\Pi
\end{eqnarray}
where $\Theta = \proj_{\Xi^\perp}$ and $\Pi = \proj_{\Xi}$.
\end{theorem}
\begin{proof}
We first prove the forward direction. Let $\Gamma$ be a partial antiunitary operator with respect to $\Xi$. Then by Theorem \ref{thm:main1}, there exist unitaries $U$ and $V$, and a language $L$ satisfying $V(\Xi) = \bbH(L)$ such that 
\begin{eqnarray}
\Gamma &=& UK_L V \nn
&=& U(\Theta_L + \Pi_L K)V \nn
&=& (U\Theta_L V) + (U \Pi_L \bar V) K \nn
&\equiv & A+BK ,
\end{eqnarray}
where $A = U \Theta_L V$ and $B = U \Pi_L \bar V$.

Next, we check that the conditions \eq{eq:fourthCharac1}--\eqref{eq:fourthCharac3} are satisfied.

\begin{eqnarray}
A^\dag B &=& (U \Theta_L V)^\dag (U \Pi_L \bar V) \nn
&=& 
V^\dag \Theta_L \Pi_L \bar V \nn
&=& 0 .
\end{eqnarray}
\begin{eqnarray}
A^\dag A &=& (U \Theta_L V)^\dag (U \Theta_L V) \nn
&=& 
V^\dag \Theta_L V \nn
&=& 
\Theta .
\end{eqnarray}
\begin{eqnarray}
B^\dag B &=& (U \Pi_L \bar V)^\dag (U \Pi_L \bar V) \nn
&=& 
V^T \Pi_L \bar V \nn
&=& 
\bar{V^\dag \Pi_L V} \nn
&=& \bar\Pi .
\end{eqnarray}
where we used the following identities that follow from \eq{eq:commutprojU}:
\begin{equation}
\Theta_L V = \proj_{\bbH^\perp(L)} V = V\proj_{V^\dag(\bbH^\perp(L))} = V \proj_{\xi^\perp} = V\Theta . 
\end{equation}
and
\begin{equation}
\Pi_L V = \proj_{\bbH(L)} V = V\proj_{V^\dag(\bbH(L))} = V \proj_{\xi} = V\Pi .
\end{equation}

We now prove the backward direction. Let $\Gamma = A+BK$, with $A$ and $B$ satisfying \eq{eq:fourthCharac1}--\eqref{eq:fourthCharac3}. Let $\psi$ and $\phi$ be arbitrary. Then,
\begin{eqnarray} 
\langle \Gamma \psi, \Gamma \phi \rangle &=& \langle (A+BK) \psi, (A+BK) \phi \rangle \nn
&=& \langle A \psi, A \phi\rangle + \langle A \psi, B \bar\phi \rangle + \langle B \bar \psi, A\phi \rangle + \langle B \bar \psi, B \bar \phi\rangle \nn
&=& \langle  \psi, A^\dag A \phi\rangle + \langle  \psi, A^\dag B \bar\phi \rangle + \langle  \bar \psi, B^\dag A\phi \rangle + \langle \psi, B^\dag B \bar \phi\rangle \nn 
&=& \langle \psi, \Theta \phi\rangle + 0+0+ \langle \bar \psi, \bar \Pi \bar \phi \rangle \nn
&=& \langle \Theta \psi, \Theta \phi \rangle + \langle \Pi \phi, \Pi \psi \rangle.
\end{eqnarray}
\end{proof}

Theorem \ref{thm:characterizationfourth} gives us an efficient algorithm (in terms of the dimensions of the matrices $A$ and $B$) for deciding if a given $\bbR$-linear operator $A+BK$ is partial antiunitary. We simply need to compute the matrices $s_1 = A^\dag B$, $s_2 = A^\dag A$ and $s_3= B^\dag B$, and check that $s_1=0$, $s_2 = s_2^2 = s_2^\dag$ (which is the definition of an orthogonal projection) and $s_3 = \overline{1-s_2}$ (since $\Theta + \Pi = I$). The theorem also gives us the following corollary.

\begin{corollary}
Let $\Gamma = A+BK$ be $\bbR$-unitary. Then $\Gamma$ is partial antiunitary with respect to a subspace $\Xi$ if and only if 

\begin{equation}
\label{eq:condcorr}
A^\dag A = \Theta \quad \mbox{and} \quad  A^\dag B = 0 ,
\end{equation}
where $\Theta = \proj_{\Xi^\perp}$.
\end{corollary}
\begin{proof}
The forward direction follows immediately from \eq{eq:fourthCharac1} and \eq{eq:fourthCharac2}. Next, we prove the backward direction. Assume that \eq{eq:condcorr} holds. Then, \eq{eq:fourthCharac1} and \eq{eq:fourthCharac2} are immediately satisfied. Since $\Gamma$ is $\bbR$-unitary,
\begin{equation}
A^\dag A + B^T \bar B = I.
\end{equation}
Hence,
\begin{equation}
B^\dag B = \overline{(B^T \bar B)} = \overline{I - A^\dag A} = \overline{I-\Theta} = \bar \Pi.
\end{equation}

\end{proof}

We conclude this section by proving some closure properties of the set of partial antiunitaries. First, we show that the partial antiunitaries are closed under inverses (i.e. under $\dag$).
\begin{theorem}
If $\Gamma$ is partial antiunitary, then  $\Gamma^\dag$ is also partial antiunitary.
\end{theorem}
\begin{proof}
If $\Gamma$ is partial antiunitary, then by Theorem \ref{thm:main1}, there exists a language $L$ and unitaries $U$ and $V$ such that $\Gamma = U K_L V$. Hence, $\Gamma^\dag = V^\dag K_L U^\dag$. Since partial antiunitaries are closed under multiplication by unitaries (by Theorem \ref{thm:prodPartialAntiunitary}), $\Gamma^\dag$ is partial antiunitary.\footnote{If $\Gamma$ is partial antiunitary with respect to $\Xi$ and $\Gamma^\dag$ is partial antiunitary with respect to $\Sigma$, one may wonder about the relation between $\Xi$ and $\Sigma$. Since, by Theorem~\ref{thm:characterizationfourth}, $\text{proj}_{\Xi^\perp}=A^\dag A$ and $\text{proj}_{\Sigma^\perp}=AA^\dag$, it is clear that at the very least $\dim\Xi=\dim\Sigma$, and thus these spaces are related by a unitary. To actually find the unitary, perform the polar decomposition of $A=UP$ into a unitary $U$ and positive semi-definite $P$. Then $\text{proj}_{\Xi^\perp}=P^2$ and $\text{proj}_{\Sigma^\perp}=U\text{proj}_{\Xi^\perp}U^\dag$.}
\end{proof}

Next, we show that the partial antiunitaries are not closed under multiplication. As a consequence, they are not a subgroup of the $\bbR$-unitaries, i.e. there exist $\bbR$-unitary operators that are not partial antiunitary.

\begin{theorem}\label{thm:partial_antiunitaries_not_group}
There exist partial antiunitaries $\Gamma$ and $\Delta$ such that $\Gamma \Delta$ is not partial antiunitary.
\end{theorem}

\begin{proof}
On a single-qubit, let $K_1$ be the partial complex  conjugation operator with respect to the language $L = \{1\}$ (equivalently, $K_1 = CK_1$ is the controlled complex conjugation operator on the first (and only) register) and consider single-qubit operators $\Gamma = K_1 H$ and $\Delta = SH K_1$. These are both partial antiunitary by Theorem \ref{thm:prodPartialAntiunitary}. 
Their product, however, is
\begin{equation}
\Gamma \Delta=K_1 HSH K_1=\frac{1}{\sqrt2}\left(\begin{array}{cc}1&0\\0&k\end{array}\right)\left(\begin{array}{cc}e^{i\pi/4}&e^{-i\pi/4}\\e^{-i\pi/4}&e^{i\pi/4}\end{array}\right)\left(\begin{array}{cc}1&0\\0&k\end{array}\right)=\frac{1}{\sqrt2}e^{i\pi/4}\left(\begin{array}{cc}1&-ik\\k&-i\end{array}\right),
\end{equation}
where $k$ is the complex conjugation operator\footnote{In this proof, we represented $\bbR$-linear operators by matrices. Note that the elements of these matrices are themselves $\bbR$-linear operators, and do not belong to a field. We discuss matrix representations of $\bbR$-linear operators further in Appendix \ref{app:matrixrep}.}.

Hence, $\Gamma \Delta = A+BK$, where 
\begin{equation}
A = \frac 1{\sqrt 2} e^{i \pi/4} \mat{1 & 0\\ 0& -i}, \quad B = \frac 1{\sqrt 2} e^{i \pi/4} \mat{0& -i\\ 1& 0} ,
\end{equation}
from which it follows that
\begin{equation}
A^\dag A = \frac 12 I, \quad A^\dag B = \frac 12 Y \neq 0 ,
\end{equation}
which does not satisfy \eq{eq:condcorr}. Hence, $\Gamma \Delta$ is not partial antiunitary.\footnote{Alternatively, we can show more generally that any $\Omega=\left(\begin{smallmatrix}a&bk\\ck&d\end{smallmatrix}\right)$ is not partial antiunitary when either $a$ and $b$ are both nonzero or $c$ and $d$ are both nonzero. This is because setting $\Omega=A+BK$ and calculating $A^\dag B=\left(\begin{smallmatrix}\bar a&0\\0&\bar d\end{smallmatrix}\right)\left(\begin{smallmatrix}0&b\\c&0\end{smallmatrix}\right)=\left(\begin{smallmatrix}0&\bar ab\\\bar dc&0\end{smallmatrix}\right)\neq\left(\begin{smallmatrix}0&0\\0&0\end{smallmatrix}\right)$ in contradiction with Theorem~\ref{thm:characterizationfourth} and in particular Eq.~\eqref{eq:fourthCharac1}.}

\end{proof}

\section{Simulating partial antiunitary operators} \label{sec:rebitPartial}

In the previous section, we studied various properties of the set of partial antiunitary operators. In this section, we consider examples of these operators and
study their rebit simulation. We do this in two parts. First, we consider partial antiunitaries which are also unitary. This is in line with the top-down approach taken by the initial use of rebit simulation \cite{bennett1997strengths}, where it is shown that real-amplitude quantum computers are just as powerful as those with complex amplitudes. Second, we consider partial antiunitaries which are neither unitary nor linear -- operators for which a bottom-up simulation is necessary.

\subsection{Rebit simulation of unitaries: top-down perspective}\label{sec:top_down}

We study the rebit encoding of unitary operators from a circuit point of view. The encoding function $\mathcal P: U_n\rightarrow T_{n+1}$ is a group homomorphism (see Proposition~\ref{prop:PL_homomorphisms}), and hence $\mathcal P(UV) = \mathcal P(U) \mathcal P(V)$ for all unitary operators $U,V \in U_n$. This means that if $U$ is a unitary operator that is implemented by a circuit $C$ consisting of unitary gates $G_1,\ldots,G_k$, then $\mathcal P(U)$ is an orthogonal operator that can be implemented by the circuit $\mathcal P(C)$, where $\mathcal P(C)$ is the circuit formed from $C$ by replacing each gate $G_i$ in $C$ with $\mathcal P(G_i)$.

We make use of the following conventions. Let $G$ be a gate contained in an $n$-qubit circuit. We write $G_{i_1,\ldots,i_s}$ to mean that $G$ acts on registers $i_1,\ldots,i_s \in \set{1,\ldots,n}$. The encoded gate $\mathcal P(G)$ is now a gate in an $(n+1)$-rebit circuit. We will continue to use the labels $1,\ldots,n$ to denote the first $n$ registers, and will use the subscript $a$ to denote the last (ancilla) register. 


We will now give examples of various common gates $G_i$ and their rebit encodings $\mathcal P(G_i)$. 
\begin{prop}  \label{prop:exampleslinear}
Under the encoding $\mathcal P(\cdot)$, the gates below transform as follows
\begin{enumerate}[(1)]
\item $X_i \mapsto X_i$
\item $Y_i \mapsto X_i \cdot X_a \cdot Z_i \cdot Z_a=-Y_i\cdot Y_a$
\item $Z_i \mapsto Z_i$
\item $H_i \mapsto H_i$
\item $S_i \mapsto CX_{ia} \cdot CZ_{ia} = H_a \cdot CZ_{ia}\cdot H_a\cdot CZ_{ia}$
\item $T_i \mapsto CH_{ia}\cdot CZ_{ia}$
\item $CX_{ij} \mapsto CX_{ij}$
\item $CZ_{ij} \mapsto CZ_{ij}$
\item $CS_{ij} \mapsto CCX_{ija} \cdot CCZ_{ija} = H_a \cdot CCZ_{ija} \cdot H_a \cdot CCZ_{ija}$ \label{item:CSij}
\item $CCZ_{ijk} \mapsto CCZ_{ijk}$
\item $Y(\theta)_i:=\cos(\theta/2)I-i\sin(\theta/2)Y_i \mapsto Y(\theta)_i$
\item $e^{i\theta/2}Z(\theta)_i:=e^{i\theta/2}\left(\cos(\theta/2)I-i\sin(\theta/2)Z_i\right) \mapsto CY(2\theta)_{ia}$
\end{enumerate}
\end{prop}
\begin{proof}
Gates $G_i$ in (1), (3), (4), (7), (8), (10), and (11) have only real entries, and hence they map to themselves under $\mathcal P$. For the other cases, we make use of \eq{eq:rebitEncodingLinear}. 
\begin{itemize}
\item For (2), $Y_i = i X_i Z_i \mapsto (XZ)_i \otimes (XZ)_a = X_i X_a Z_i Z_a$. 
\item For (5), $S_i = \ketbra 00 + i \ketbra 11 \mapsto \ketbra 00_i \otimes I_a + \ketbra 11_1\otimes X_a Z_a = C(XZ)_{ia} = CX_{ia} CZ_{ia} = H_a CZ_{ia} H_a CZ_{ia}$. 
\item For (6), $T = \ketbra 00 + e^{i \pi/4} \ketbra 11 = \ketbra 00 + \frac 1{\sqrt 2} \ketbra 11 + \frac 1{\sqrt 2} i \ketbra 11 \mapsto 
( \ketbra 00 + \frac 1{\sqrt 2} \ketbra 11)_i \otimes I_a + \frac 1{\sqrt 2} \ketbra 11_i \otimes X_a Z_a \mapsto  \ketbra 00_i  \otimes I_a + \frac 1{\sqrt 2} \ketbra 11_i \otimes (I+X Z)_a = C(\frac 1{\sqrt 2} (I+XZ))_{ia} = C(HZ)_{ia} = CH_{ia} CZ_{ia}$. 
\item For (9), $CS_{ij} = \ketbra 00_i \otimes I_j + \ketbra 11_i \otimes S_j = (\ketbra {00}{00} + \ketbra {01}{01}+\ketbra {10}{10} + i \ketbra {11}{11})_{ij} \mapsto (\ketbra {00}{00} + \ketbra {01}{01}+\ketbra {10}{10})_{ij}\otimes I_a + \ketbra {11}{11}_{ij} \otimes X_a Z_a = CC(XZ)_{ija} = CCX_{ija} CCZ_{ija} = H_a CCZ_{ija} H_a CCZ_{ija}$.
\item For (12), $e^{i\theta/2}Z(\theta)_i=\ket{0}\bra{0}_i+e^{i\theta}\ket{1}\bra{1}_i\mapsto (\ket{0}\bra{0}+\cos(\theta)\ket{1}\bra{1})_i\otimes I_a-i\sin(\theta)\ket{1}\bra{1}_i\otimes Y_a=\ket{0}\bra{0}_i\otimes I_a+\ket{1}\bra{1}_i\otimes(\cos(\theta)I_a-i\sin(\theta)Y_a)=CY(2\theta)_{ia}$.
\end{itemize}
\end{proof}

In quantum mechanics, the global phase of a quantum state does not play any physical role. Therefore, two unitary operators $U_1$ and $U_2$ that differ by a global phase (i.e. $U_1 = e^{i\theta} U_2$ for some $\theta \in \mathbb R$) are physically equivalent. How does this equivalence manifest in the rebit encoding? To start answering this question, we first define $G(\theta) = e^{i\theta} I$ to be the global phase operator with angle $\theta$. Since $G(\theta)$ is unitary, we could find its image under the rebit encoding $\mathcal P(\cdot)$. As we shall now see, the image of $G(\theta)$ under $\mathcal P$ is the rotation matrix $R(\theta) = \left(\begin{matrix}
\cos \theta & -\sin\theta \\ \sin\theta & \cos\theta
\end{matrix}\right)$ acting on the ancilla register.

\begin{prop} \label{prop:globalphase}
$\mathcal P:G(\theta) \mapsto R(\theta)_a$.
\end{prop}
\begin{proof} We compute the action of $\mathcal P$ on $G(\theta)$.
\begin{eqnarray*}
\mathcal P:G(\theta) = (\cos\theta + i \sin\theta)I &\mapsto & (\cos\theta) I  \otimes \left(\begin{matrix}
1 & 0 \\ 0 & 1
\end{matrix}\right)_a + (\sin\theta) I \otimes \left(\begin{matrix}
0 & -1 \\ 1 & 0
\end{matrix}\right)_a  \\
&=& I \otimes \left(\begin{matrix}
\cos \theta & -\sin\theta \\ \sin\theta & \cos\theta
\end{matrix}\right)_a = R(\theta)_a.
\end{eqnarray*} 
\end{proof}

Since the encoding of the global phase operator is restricted to the ancilla, it follows that measurements on the other $n$ rebits will not give information about the global phase. By Proposition~\ref{prop:rebitMeasurement} these are exactly the measurements of rebits in $P$ that correspond to qubit measurements in the simulated space $L$. On the other hand, measurements on the ancilla can yield information about the global phase of the simulated state, as we described in Section~\ref{sec:bottom_up_tomography}.



\subsection{Rebit simulation of non-unitaries: bottom-up perspective}

We now give examples of the rebit simulation of various partial antiunitaries that are not unitary (and not linear).
As noted earlier, these are therefore operators for which a bottom-up simulation is necessary. We then discuss the partial complex conjugation operator $K_L$ for an arbitrary language $L\subseteq\{0,1\}^n$ and make a connection between the complexity of deciding $L$ and the complexity of simulating $K_L$.

As before, we use the indices $i, j, k$ to refer to any of the first $n$ registers, and $a$ to refer to the ancilla register. Given that we are allowed to perform any orthogonal gate from $T_{n+1}$ on the $n+1$ rebits in the simulator, to find non-unitary simulations we can just try various orthogonal gates and see if any decode to non-unitary operators on the $n$ qubit system being simulated. But notice, if an orthogonal gate $T\in T_n$ does not act on the ancilla register, its image under $\mathcal L$ is itself, since $\mathcal L (T \otimes I_a) = T$. Hence, any gate not in $\mathcal P(U_n)$ must necessarily act nontrivially on the $a$th register. For instance,
\begin{prop}  \label{prop:examplesNonlinear}
Under the rebit decoding $\mathcal L(\cdot)$, the gates below transform as follows
\begin{enumerate}[(I)]
\item \label{enum:Ha} $H_a \mapsto G(\frac \pi 4) K$
\item \label{enum:Za} $Z_a \mapsto K$
\item \label{enum:CZia} $CZ_{ia} \mapsto CK_i$
\item \label{enum:CCZija} $CCZ_{ija} \mapsto CCK_{ij}$
\item \label{enum:ChZ} $C^{h}Z_{c_1,c_2,\dots,c_h,a}\mapsto C^hK_{c_1,c_2,\dots,c_h}$ .
\end{enumerate}
\end{prop}
\begin{proof}

Note that (\ref{enum:Za}) was proved in \eq{eq:LZa}. So next,
we prove (\ref{enum:Ha}): First, note that 
$$ H = \frac 1{\sqrt 2} \left(\begin{matrix} 1 & 1 \\ 1 & -1 \end{matrix}\right) = \left(\begin{matrix} \cos\frac\pi 4 & -\sin\frac\pi 4 \\ \sin\frac\pi 4 & \cos\frac\pi 4 \end{matrix}\right)\left(\begin{matrix} 1 & 0 \\ 0 & -1 \end{matrix}\right) = R\left(\frac\pi 4\right) Z,$$
which implies that
\begin{equation}
\mathcal L(H_a) = \mathcal L\left(R\left(\frac \pi 4\right) Z_a\right) = \mathcal L\left(R\left(\frac \pi 4\right)\right)\mathcal L(Z_a) = G\left(\frac \pi 4\right) K.
\end{equation}
where we used the fact that $\mathcal L$ is a homomorphism for the second equality, and \eq{prop:globalphase} and \eq{eq:LZa} for the third equality. \newline
For (\ref{enum:CZia}), we make use of \eq{eq:LWDefn}. By denoting $\psi_{x0} = \Re \psi_x$ and $\psi_{x1} = \Im \psi_x$, we get
\begin{eqnarray*}
\mathcal L\left(CZ_{ia}\right) \left(\sum_x \psi_x \ket x\right) 
&=& \mathcal L\left( CZ_{ia} \sum_{xa} \psi_{xa} \ket{xa} \right) \\
&=& \mathcal L \left( \sum_{xa} (-1)^{x_i a} \psi_{xa} \ket{xa}\right) \\
&=& \mathcal L \left( \sum_x\ket x \left(\psi_{x0} \ket 0 + (-1)^{x_i} \psi_{x1}\ket 1\right) \right) \\
&=& \sum_x \ket x (\psi_{x0} + i (-1)^{x_i} \psi_{x1}) \\
&=& \sum_{x:x_i = 0} (\psi_{x0}+i \psi_{x1})\ket x + \sum_{x:x_i = 1} (\psi_{x0} - i \psi_{x1})\ket x \\
&=& \sum_{x:x_i=0} \psi_x \ket x + \sum_{x:x_i=1} \bar{\psi}_x \ket x \\
&=& CK_i  \left(\sum_x \psi_x \ket x\right).
\end{eqnarray*}
Hence, $\mathcal L(CZ_{ia}) = CK_i$. For (\ref{enum:CCZija}) and (\ref{enum:ChZ}) the arguments are analogous and straightforward.
\end{proof}

\subsubsection{Simulation of the partial complex conjugation operator}
Next, we consider the partial complex conjugation operator $K_L$ defined in \eq{eq:KLdef}, where $L \subseteq \set{0,1}^n$ is some language. We start by expressing the image of $K_L$ under $\mathcal P$ (characterized, for example, in Theorem \ref{thm:PKLeqn}) in terms of the indicator function of the language $L$.

\begin{prop}
Let $L\subseteq \set{0,1}^n$. Then,
\begin{equation}
\label{eq:outputofPKL}
\mathcal P(K_L) = \sum_{x\in \set{0,1}^n}\sum_{a\in \set{0,1}} (-1)^{aL(x)} \ketbra{xa}{xa},
\end{equation}
where $L(x)$ is the indicator function of $L$.
\end{prop}
\begin{proof}
We make use of Theorem \ref{thm:PKLeqn}.

\begin{eqnarray}
\mathcal P(K_L) &=& \Theta_L \otimes I + \Pi_L \otimes Z \nn
&=& \sum_{x\not\in L} \ketbra xx \otimes \sum_{a\in \set{0,1}} \ketbra aa + \sum_{x\in L} \ketbra xx \otimes \sum_{a\in \set{0,1}} (-1)^a \ketbra aa \nn
&=& \sum_{x\in \set{0,1}^n}\sum_{a\in \set{0,1}} (-1)^{aL(x)} \ketbra{xa}{xa} .
\end{eqnarray}

\end{proof}

We now show that if $L\subseteq \set{0,1}^n$ is decidable by some quantum circuit $C_L$, then we can construct a quantum circuit that implements the operator $\mathcal P(K_L)$. Here, we say that $C_L$ decides $L$ if when given input $\ket x$, $C_L$ outputs $\ket{L(x)}$. More precisely, taking into account all ancilla registers, the action of $C_L$ may be described by
\begin{equation} \label{eq:actionofCL}
C_L  \ket x_{1,\ldots,n} \ket 0_\alpha = \ket{L(x)}_1\ket{j(x)}_{2,\ldots,n,\alpha},
\end{equation}
where $j(x)$ are junk bits.
\begin{prop} \label{prop:PKL1}
Let $L\subseteq \set{0,1}^n$. Let $C_L=(C_L)_{1,\ldots,n,a}$ be a quantum circuit that acts on basis states according to \eq{eq:actionofCL}. Then 
\begin{equation} \label{eq:PKL1}
\mathcal P(K_L)_{1,\ldots,n,a} = \bra 0_\alpha (C_L^\dag)_{1,\ldots,n,a} CZ_{1a} (C_L)_{1,\ldots,n,a} \ket 0_\alpha.
\end{equation}
\end{prop}
\begin{proof}
Starting with the input state $\ket\phi = \sum_{xa} \psi_{xa}\ket{xa}_{1,\ldots,n,a}$, and appending an ancilla register $\alpha$ initialized to $\ket 0$, the system evolves as follows:
\begin{eqnarray*}
\sum_{xa} \psi_{xa} \ket{xa}_{1,\ldots,n,a}\ket 0_\alpha &\xrightarrow{C_L} & \sum_{xa} \psi_{xa} \ket{L(x)}_1 \ket{j(x)}_{2,\ldots,n,\alpha} \ket a_a \\
&\xrightarrow{CZ_{1a}} & \sum_{xa} \psi_{xa} (-1)^{a L(x)} \ket{L(x)}_1 \ket{j(x)}_{2,\ldots,n,\alpha} \ket a_a \\
&\xrightarrow{C_L^\dag} & \sum_{xa} \psi_{xa} (-1)^{a L(x)} \ket{x}_{1,\ldots,n} \ket a_a \ket 0_\alpha.
\end{eqnarray*}
Hence, 
\begin{eqnarray}
\bra 0_\alpha (C_L^\dag)_{1,\ldots,n,a} CZ_{1a} (C_L)_{1,\ldots,n,a} \ket 0_\alpha  \left(\sum_{xa} \psi_{xa} \ket{xa}\right) &=& \sum_{xa} \psi_{xa} (-1)^{a L(x)} \ket{x}_{1,\ldots,n} \ket a_a \\ &=& \mathcal P(K_L) \sum_{xa} \psi_{xa} \ket {xa},
\end{eqnarray}
where the last line follows from \eq{eq:outputofPKL}. Therefore, \eq{eq:PKL1} is true since the above equality holds for all $\ket \phi$.
\end{proof}

So far, our discussion has dealt with the case where $n$ is fixed. We now consider the case where $n$ is allowed to grow arbitrarily. 
\begin{corollary} \label{cor:uniformfamily}
Let $L \in \P$. Then $\set{\mathcal P(K_L)_{1,\ldots,n,a}}_n$ can be implemented by a uniform family of polynomial-sized quantum circuits $\set{Q_n}_n$ that comprise only orthogonal gates.
\end{corollary}
\begin{proof}
$L\in \P$ means that there exists a uniform family of polynomial-sized classical circuits $\set{C_n}_n$, where each $C_n$ comprises only reversible gates, say Toffoli gates and NOT gates, such that $C_n(x,0) = (L(x),j(x))$. Let $\tilde C_n$ be a quantum circuit formed from $C_n$ by replacing each classical gate by its quantum counterpart so that $\tilde C_n\ket x_{1,\ldots,n} \ket 0_\alpha = \ket{L(x)}_1\ket{j(x)}_{2,\ldots,n,\alpha}$. We then construct the circuit $(C_L^\dag)_{1,\ldots,n,a} CZ_{1a} (C_L)_{1,\ldots,n,a}$. By Proposition \ref{prop:PKL1}, the circuit $Q_n = (\tilde C_L^\dag)_{1,\ldots,n,a} CZ_{1a} (\tilde C_L)_{1,\ldots,n,a}$, after discarding ancilla rebits, implements $\mathcal P(K_L)$. Furthermore, since $C_n$ is polynomial-sized, $\tilde C_n$ and hence $Q_n$ are also polynomial-sized. Finally, both the quantum Toffoli gate and the quantum NOT gate (i.e. the $X$ gate) are orthogonal gates. Hence, the set $\set{\mathcal P(K_L)_{1,\ldots,n,a}}_n$ has the desired properties.
\end{proof}

\begin{corollary} \label{cor:uniformfamily2}
Let $L \in \P$. Then $\set{\mathcal (K_L)_{1,\ldots,n,a}}_n$ can be implemented by a uniform family of polynomial-sized $\mathbb{R}$-unitary quantum circuits $\set{Q_n}_n$ that comprise only $CCZ$, $H$ gates and exactly one $CK$ gate.
\end{corollary}
\begin{proof}
The circuit $Q_n$ constructed in the proof of Corollary \ref{cor:uniformfamily} is $(\tilde C_L^\dag)_{1,\ldots,n,a} CZ_{1a} (\tilde C_L)_{1,\ldots,n,a}$, and it implements the operator $\mathcal P(K_L)$. By applying the rebit decoding operator $\mathcal L$ to this, we find that $K_L = \mathcal L(\tilde C_L^\dag)_{1,\ldots,n,a} \mathcal L(CZ_{1a}) \mathcal L(\tilde C_L)_{1,\ldots,n,a} = L(\tilde C_L^\dag)_{1,\ldots,n,a} CK_1 (\tilde C_L)_{1,\ldots,n,a}$, where we used the fact that the circuits $\tilde C_L$ and its inverse are unchanged by the unencoding (this holds since they do not act on the ancilla register). Now, the Toffoli and $X$ gate can be simulated by $H$ and $CCZ$. Hence, $\set{Q_n}_n$ can be simulated by only the gates $CCZ$ and $H$ and a single $CK$ gate.
\end{proof}
Let us now specialize \eq{eq:PKL1} to the controlled-controlled complex conjugation operator $CCK$, which is a special case of $K_L$. While we know from Proposition~(\ref{prop:examplesNonlinear}) that $CCK$ can be simulated using one $CCZ$ gate, the following equivalent simulation makes use of the previous discussion and shows that the ancilla need only be operated on by a single two-qubit gate.
\begin{prop}
\begin{equation} \label{eq:PKL2}
\mathcal P(CCK_{ij}) = \bra 0_\alpha CCX_{ij \alpha} CZ_{\alpha a} CCX_{ij\alpha} \ket 0_\alpha.
\end{equation}
\end{prop}
\begin{proof}
We shall use the construction of the circuit described in Proposition \ref{prop:PKL1}, with a few modifications to the labels. Recall that $CCK$ corresponds to the language $L = \set{x: x_i x_j = 1}$, i.e. $L(x) = L(x_i x_j)$. Note that the circuit $C_L = CCX_{ij\alpha}$ maps $\ket x_{1,\ldots, n }\ket 0_\alpha \mapsto \ket{x_i x_j}_\alpha \ket x_{1,\ldots, n } = \ket{L(x)}_\alpha \ket x_{1,\ldots, n }$, which is of the form given in \eq{eq:actionofCL}, except that the labels 1 and $\alpha$ are switched. Hence, by replacing the labels $1a$ in the $CZ$ gate in \eq{eq:PKL1} with $\alpha a$, we obtain \eq{eq:PKL2}.
\end{proof}

\section{Universal gate sets for $\mathbb{R}$-unitaries} \label{sec:universalgateset}

In Section \ref{sec:rebitPartial}, we found several examples of nonunitary operators, like $K$, $CK$ and $CCK$, that can be simulated by orthogonal quantum circuits via the rebit encoding. The goal of this section is to find universal sets of gates for the $\bbR$-unitaries. We first introduce some definitions. 
\begin{definition}
Let $\mathcal G_1$ and $\mathcal G_2$ be two gate sets\footnote{gate set here refers to a set of gates, which may be either finite or infinite.}. We say that $\mathcal G_2$ \textit{exactly simulates} $\mathcal G_1$ if for all $G \in \mathcal G_1$, there exists a circuit $C$ formed using gates in $\mathcal G_2$ and ancilla registers (that can be initialized to any computational basis state) such that $\bra a G\ket b = \bra a C \ket b$ for all vectors $a,b$.  We say that $\mathcal G_2$ \textit{approximately simulates} $\mathcal G_1$ if any gate in $\mathcal G_1$ can be approximated in the operator norm (e.g.~Definition~\ref{defn:rlinear_norm}) to within arbitrary accuracy by a sequence of gates from $\mathcal G_2$. 
\end{definition}


Write $\mathcal G_1 \leq \mathcal G_2$ if $\mathcal G_2$ exactly simulates $\mathcal G_1$; and $\mathcal G_1 \lesssim \mathcal G_2$ if $\mathcal G_2$ approximately simulates $\mathcal G_1$. If $\mathcal G_1 \leq \mathcal  G_2$ and $\mathcal G_2 \leq \mathcal G_1$, i.e.\ if the gate sets exactly simulate each other, we write $\mathcal G_1 \equiv \mathcal G_2$, and say that $\mathcal G_1$ are $\mathcal G_2$ are \textit{exact-simulation equivalent} to each other. If $\mathcal G_1 \lesssim \mathcal  G_2$ and $\mathcal G_2 \lesssim \mathcal G_1$, i.e.\ if the gate sets approximately simulate each other, we write $\mathcal G_1 \cong \mathcal G_2$, and say that $\mathcal G_1$ are $\mathcal  G_2$ are \textit{approximate-simulation equivalent} to each other. Note that exact-simulation equivalence is a special case of approximate-simulation equivalence, and that approximate-simulation is equivalent to strict universality defined in \cite{aharonov2003simple}. Note also that $\leq$ is a transitive relation, i.e.\ if $\mathcal G_1 \leq \mathcal G_2$ and $\mathcal G_2 \leq \mathcal G_3$, then $\mathcal G_1 \leq \mathcal G_3$. If either one or both of the first two $\leq$ signs in the previous sentence is changed to $\lesssim$, then $\mathcal G_1 \lesssim \mathcal G_3$. 

For a gate set $\mathcal G = \set{g_1,\ldots,g_s}$, we say that $g_1,\ldots,g_s$ generate $\mathfrak G$ if $\mathfrak G$ is the set of all operators that can be written as circuits using gates from $\mathcal G$. We emphasize that we choose to specify generators $g_j$ independently from their support, which can be chosen arbitrarily. That is, we view each $g_j$ as a gate that may act on any set of qubits in the circuit. We denote the generated set as $\mathfrak G=\langle \mathcal G\rangle = \langle g_1,\ldots,g_s\rangle$.

Examples of gate sets that approximately simulate the special unitary group on $n$ qubits $SU(2^n)$ include the Clifford+T set $\set{H, CZ, T}$ \cite{boykin1999universal}, and Kitaev's gate set $\set{H,CS}$ \cite{kitaev1997quantum} (also, see Theorem 1 of \cite{aharonov2003simple}). An example of gate set that approximately simulates the orthogonal operators $T_n$ is $\set{CCX,H}$ (Theorem 3.2 of \cite{shi2002both}). Since $CCZ_{ijk} = H_k CCX_{ijk} H_k$, it follows that $CCZ$ and $H$ can simulate $CCX$. Hence, $\set{CCZ,H}$ is also capable of approximately simulating $T_n$.

Other orthogonal gate sets, or even single gates, can simulate all the unitaries $U_n$ via the rebit encoding. Rudolph and Grover \cite{rudolph20022} provide an example: the controlled-Y rotation $CY(\theta)$ for any $\theta$ that is an irrational multiple of $\pi$ (e.g.~$\theta=\pi/e$). More specifically, $CY$ can be approximated to any desired accuracy by some power of $CY(\theta)$ and is orthogonal, so it can be applied in the rebit encoding without the ancilla. Single-qubit unitaries can be compiled (to any accuracy) from $CY(\theta)$ using parts (11) and (12) of Proposition~\ref{prop:exampleslinear}.

Our main interest in this section is simulating $\bbR U_n$ using the rebit encoding. We now show that the set $\mathbb G$ defined by 
\begin{equation}
\mathbb G := \set{H,CCZ,CCK,G(\textstyle\frac \pi 4)K}
\end{equation} 
can approximately simulate $\bbR U_n$.

\begin{theorem} \label{thm:universalgateset}
$\mathbb G$ approximately simulates $\bbR U_n$.
\end{theorem}
\begin{proof}
Theorem~3.2 of \cite{shi2002both} shows that $\set{CCX,H}$ approximately simulates $T_{n+1}$. Hence, the image of this set under $\mathcal L$ (which is bijective) gives an approximate simulation for $\bbR U_n$. There are four cases corresponding to the gates acting on different sets of wires that we need to consider, namely, (i) $CCZ$ gate acting on three numbered registers, (ii) $H$ gate acting on a numbered register, (iii) $CCZ$ gate acting on 2 numbered registers and the ancilla register, (iv) $H$ gate acting on the ancilla register. From Propositions \ref{prop:exampleslinear} and \ref{prop:examplesNonlinear}, the corresponding gates under $\mathcal L$ are as follows:
\begin{eqnarray}
CCZ_{ijk} &\mapsto & CCZ_{ijk} \\
H_i &\mapsto & H_i \\
CCZ_{ija} &\mapsto & CCK_{ij} \\
H_a &\mapsto & G(\textstyle\frac \pi 4)K.
\end{eqnarray}
Hence, the set $\mathbb G = \set{H,CCZ,CCK,G(\textstyle\frac \pi 4)K}$ approximately simulates $\bbR U_n$.
\end{proof}

Since the gates in $\mathbb G$ are contained in $\bbR U_n$, Theorem \ref{thm:universalgateset} tells us that $\mathbb G \cong \bbR U_n$. We now give some more examples\footnote{Whenever a gate set consists of exactly one gate, we drop the curly braces and denote the set by the element it contains. For example, if $G$ is a gate such that $\set{G}$ is exactly simulated by a gate set $\mathcal G$, we write $G \leq \mathcal G$ rather than $\set{G} \leq \mathcal G$.} of exact and approximate simulation with the goal of using these relations to find other universal gates sets starting from $\mathbb G$.
\begin{lemma}\label{lem:gatesetsimexamples}
\begin{eqnarray}
(i) && CCX \leq \set{H,CCZ} \label{eq:CCXHCCZ} \\
(ii) && K \leq CK \leq CCK \label{eq:KleqCK} \\
(iii) && CS \leq \set{G(\textstyle\frac \pi 4) , CCK} \label{eq:CSGCCK}\\
(iv) && \set{H,CCZ} \lesssim \set{H,CS} \leq \set{H,G(\textstyle\frac \pi 4),CCK}  \label{eq:HCCZ} \\
(v) && CCK \leq \set{CCZ, H, CK} \label{eq:CCKCCZHCK}\\
(vi) && \mbox{If } L \in \P, \mbox{ then } K_L \lesssim \set{H,CCZ,CK} \lesssim \set{H,CCK,\Gp}. \label{eq:KLHCCZ}
\end{eqnarray}
\end{lemma}

\begin{proof} 
To show that $G_1 \leq G_2$ or $G_1 \lesssim G_2$, it suffices to show that each gate in $G_1$ can be approximately simulated by a circuit consisting of gates in $G_2$. (i) follows from the fact that
\begin{equation}
CCX_{ijk} = H_k CCZ_{ijk} H_k.
\end{equation}
(ii) follows from the facts that
\begin{equation}
K = \bra 1_b CK_b \ket 1_b, \quad CK_i = \bra 1_b CCK_{ib} \ket 1_b.
\end{equation}
For (iii), taking $\mathcal L$ on both sides of Item \ref{item:CSij} of Proposition \ref{prop:exampleslinear} gives:
\begin{equation}
CS_{ij} = \mathcal L(H_a \cdot CCZ_{ija} \cdot H_a \cdot CCZ_{ija}) = \Gp \cdot K \cdot CCK_{ij} \cdot \Gp \cdot K \cdot CCK_{ij} .
\end{equation}
Hence, $CS \leq \set{\Gp,K,CCK}$. But $K\leq CCK$ from \eq{eq:KleqCK}. Hence, $CS \leq \set{\Gp,CCK}$. \newline
For (iv), $\set{H,CS}$ can approximately simulate any unitary \cite{kitaev1997quantum} including $\set{H,CCZ}$. Hence, $\set{H,CCZ} \lesssim \set{H,CS}$. Combining this result with \eq{eq:CSGCCK} produces (iii). \newline
For (v), by taking $\mathcal L$ on both sides of \eq{eq:PKL2}, we obtain 
\begin{eqnarray}
CCK_{ij} &=& \mathcal L(\bra 0_\alpha CCX_{ij \alpha} CZ_{\alpha a} CCX_{ij\alpha} \ket 0_\alpha) \\
&=& \bra 0_\alpha CCX_{ij\alpha} CK_\alpha CCX_{ij\alpha} \ket 0_\alpha.
\end{eqnarray}
Hence $CCK \leq \set{CCX,CK}$. By combining this result with \eq{eq:CCXHCCZ}, we get \eq{eq:CCKCCZHCK}. \newline
For (vi), we use Corollary \ref{cor:uniformfamily2}, which says that if $L \in \P$, then $K_L$ can be implemented by a uniform family of polynomial-sized quantum circuits that comprise only $CCZ$, $H$ gates and exactly one $CK$ gate. Hence, $K_L \lesssim \set{H,CCZ,CK}$. By using \eq{eq:KleqCK} and \eq{eq:HCCZ}, $\set{H,CCZ,CK} \lesssim \set{H,CCK,\Gp}$.
\end{proof}

Using Lemma~\ref{lem:gatesetsimexamples}, we now give examples of various finite gate sets which are approximate-simulation equivalent to $\mathbb G$ and therefore also to $\bbR U_n$. We start with finite gate sets.
\begin{prop} \label{prop:finitegatesets}
The following finite gates sets are all exact-simulation equivalent or approximate-simulation equivalent to one another. Hence, they are all approximate-simulation equivalent to $\bbR U_n$.
\begin{enumerate}[(i)]
\item $\set{H,CCZ,CCK, \Gp K, K}$
\item $\set{H,CCZ,CCK, K,\Gp}$
\item $\set{H,CCZ,CCK,\Gp}$
\item $\set{H,CS,CCK,\Gp}$
\item $\set{H,CCK,\Gp}$
\item $\set{H,CCZ,CK,\Gp}$.
\end{enumerate}
\end{prop}

\begin{proof} \hfill
\begin{itemize}
\item (i) $\equiv$ (ii): Clearly, $\Gp K \leq \set{K,\Gp}$, so (i) $\leq$ (ii). Also, $\Gp = \Gp K\cdot K$, $\Gp \leq \set{\Gp K, K}$, so (ii) $\leq$ (i). 
\item (ii) $\equiv$ (iii): Since $K\leq CCK$, by \eq{eq:KleqCK}, so (ii) $\leq$ (iii). Conversely, (iii) $\subset$ (ii), so (iii) $\leq$ (ii)
\item (iii) $\cong $ (iv) $\cong$ (v):  By \eq{eq:HCCZ}, $\set{H,CCZ} \lesssim \set{H,\Gp,CCK}$. Hence, (iii) $\lesssim$ (iv). By \eq{eq:CSGCCK}, $CS \leq \set{\Gp , CCK}$, so (iv) $\leq$ (v). But (v) $\subset$ (iii), so (v) $\leq$ (iii).
\item (v) $\cong $ (vi): By \eq{eq:CCKCCZHCK},  $CCK \leq \set{CCZ, H, CK} $, so (v) $\leq$ (vi). By \eq{eq:HCCZ}, $\set{H,CCZ} \lesssim \set{H,\Gp,CCK}$, and by \eq{eq:KleqCK}, $CK \leq CCK$, so (vi) $\leq$ (v).
\end{itemize}
\end{proof}

We can also find various infinite gate sets which that are approximate-simulation equivalent to $\bbR U_n$.
\begin{prop} \label{prop:infinitegatesets}
The following infinite gates sets are approximate-simulation equivalent to $\bbR U_n$.
\begin{enumerate}
\item[(vii)] $\set{H,\Gp} \cup \set{K_L:L\subseteq\P}$
\item[(viii)] $\set{H,CS,\Gp} \cup \set{K_L:L\subseteq\P}$
\item[(ix)] all the above gate sets in this list as well as in Proposition \ref{prop:finitegatesets} with $\Gp$ replaced by $\set{G(\theta): \theta \in [0,2\pi)}$.
\end{enumerate}
\end{prop}

\begin{proof} We'll continue the numbering from Proposition \ref{prop:finitegatesets}.
\begin{itemize}
\item (vii) $\cong$ (v): By \eq{eq:KLHCCZ}, $K_L \lesssim \set{H,CCK,\Gp}$. Hence, (vii) $\lesssim$ (v). Also, (v) $\subseteq$ (vii), so (v) $\leq$ (vii).
\item (vii) $\equiv$ (viii): (vii) $\subseteq$ (viii), so (vii) $\leq$ (viii). Using $CS \leq \set{G(\textstyle\frac \pi 4) , CCK}$ (from \eq{eq:CSGCCK}) and $CCK \leq K_L$, we get (viii) $\leq$ (vii).
\item We first show that $LHS:=\bbR U_n \cong RHS:=\set{H,CCZ,CCK,K} \cup \set{G(\theta):\theta\in [0,2\pi)}$. Since $\Gp\in \set{G(\theta):\theta\in [0,2\pi)}$, $LHS = \bbR U_n \cong \set{H,CCZ,CCK, K,\Gp} \leq RHS$. Conversely, the gates in RHS are all either unitary or are $K$ or $CCK$. By Proposition \ref{prop:examplesNonlinear}, these are all images of orthogonal matrices under $\mathcal L$. Hence, $RHS \leq LHS$, which completes the proof of $LHS \cong RHS$. Next, notice that the set RHS is identical to the set (ii) in Proposition \ref{prop:finitegatesets}, except that $\Gp$ is replaced by $G(\theta)$. Hence, making this replacement in all the above proofs, we get (ix).
\end{itemize}
\end{proof}


Finally, we give a set of operators that exactly simulates $\bbR U_n$. Denote the set of operators which can be expressed as a product of partial antiunitary operators by 
\begin{eqnarray} \label{eq:partialantiunitaries}
 \langle \mbox{partial antiunitaries}\rangle &:=& \langle V: V \mbox{ is a partial antiunitary operator}\rangle \nonumber \\
 &=& \set{W: \exists \mbox{ partial antiunitaries } W_1,\ldots, W_k \nonumber\\ &&\qquad\mbox{ such that } 
 W = W_1\ldots W_k}.
\end{eqnarray}
We show that the set in Eq.~\eqref{eq:partialantiunitaries} is exact-simulation equivalent to the $\bbR$-unitary operators $\bbR U_n$. That is, any operator $\Gamma\in\bbR U_n$ can be written as the product $W_kW_{k-1}\dots W_1$ of some partial antiunitaries $W_j$. Indeed, we can even take $W_j\in\bbR U_n$ as well, implying we need no extra ancilla qubits. Since the $\bbR$-unitaries are the image of the orthogonal operators under the decoding map $\mathcal{L}$, it is useful to have a lemma relating to the compiling of orthogonal operators.
\begin{lemma}\label{lem:orthog_compiling}
Let $W\in T_n$ be an $n$-qubit orthogonal operator. Then $W$ can be written as the product of single-qubit orthogonal operators and multiply-controlled $Z$ operators $C^{h}Z$ on the same $n$-qubits.
\end{lemma}
\begin{proof}
The proof is essentially the realization that the compilation scheme for unitaries in Chapter 4 of \cite{nielsen2010quantum} works for compiling orthogonal gates into products of the claimed orthogonal gates as well. We complete the proof in Appendix~\ref{app:orthog_compiling}.
\end{proof}

\begin{theorem}\label{thm:logicalActionOfOrthogonals}
\begin{equation}
\bbR U_n \equiv \langle \mathrm{partial} \mbox{ } \mathrm{antiunitaries}\rangle .
\end{equation}
\end{theorem}
\begin{proof}\hfill
\begin{itemize}
\item[$(\leq)$] 
This direction follows from the compiling lemma, Lemma~\ref{lem:orthog_compiling}. Let $\Gamma\in\bbR U_n$ and $\mathcal{P}(\Gamma)=W\in T_{n+1}$. Then the lemma provides us with a sequence of orthogonal gates $V_j\in T_{n+1}$ such that $V_kV_{k-1}\dots V_1=W$. Thus, $\mathcal{L}(V_k)\mathcal{L}(V_{k-1})\dots\mathcal{L}(V_1)=U$. Since $V_j$ is orthogonal, if it is not supported on the rebit ancilla then $\mathcal{L}(V_j)=V_j$. If the ancilla is in its support there are two cases: either (1) $V_j$ is a $C^{h}Z$ gate and so $\mathcal{L}(W_j)$ is a $C^{h}K$ gate or (2) $V_j$ is a single-qubit orthogonal gate and so $\mathcal{L}(W_j)$ is a global phase gate. Therefore, $\Gamma=\mathcal{L}(V_k)\mathcal{L}(V_{k-1})\dots\mathcal{L}(V_1)$ is indeed a sequence of unitaries alternating with partial complex conjugation operators, i.e.
\begin{equation}\label{eq:compile_rlinear}
\Gamma=U_{k'}K_{L_{k'}}U_{k'-1}K_{L_{k'-1}}U_{k'-2}\dots U_1K_{L_1}U_0.
\end{equation}
Since $U_j$ and $K_{L_j}$ are both partial antiunitary for all $j$, we have found a product of partial antiunitaries making $\Gamma$.
\item[$(\geq)$] This direction follows from Theorem~\ref{thm:partialAntiunitaryIsUnitaryElement}, which says that all partial antiunitaries are $\mathbb{R}$-unitary.
Products of partial antiunitaries, like those found in $\langle \mathrm{partial} \mbox{ } \mathrm{antiunitaries}\rangle$, must also be $\mathbb{R}$-unitary because it is a group.
\end{itemize}
\end{proof}

To conclude this section, we discuss how efficient the rebit encoding is for simulating (1) an arbitrary unitary circuit (top-down simulation) and (2) an arbitrary $\bbR$-unitary circuit (bottom-up simulation). In the top-down case, we consider the universal gate set $\{H,T,CX\}$ \cite{nielsen2010quantum}. In the bottom-up case, using Proposition~\ref{prop:finitegatesets}, we focus on the universal gate set $\{H,CCZ,CK,G(\pi/4)\}$ because it is relatively simple. Similar analyses could be done with any other universal sets of gates.

Let us start with the top-down case. Say we have a depth-$d$ circuit consisting of gates from $\{H,T,CX\}$ on $n$-qubits. What is the depth and width of a rebit circuit required to simulate it? We provide two approaches trading off depth and width.
\begin{theorem}\label{thm:topdown_effic}
Let $C$ be an $n$-qubit, depth-$d$ unitary circuit using gates from $\{H,T,CX\}$. Then $C$ (applied to $\ket{0}^{\otimes n}$) can be simulated (i.e.~we can make $\mathcal{P}(C\ket{0}^{\otimes n})$) using either
\begin{enumerate}
\item an orthogonal circuit of depth at most $dn$ on $n+1$ rebits
\item or an orthogonal circuit of depth at most $d$ on $2n$ rebits.
\end{enumerate}
\end{theorem}
\begin{proof}
To show the first statement, we proceed using the $\mathcal{P}$ mappings of the gates $\{H,T,CX\}$ from Proposition~\ref{prop:exampleslinear}. Since $T$ gates that happen in parallel (at most $n$ at once) must access the ancilla simultaneously, we get a depth blowup by a factor of $n$ of the orthogonal circuit over the unitary circuit it is simulating.

The following rebit encoding was used in \cite{mckague2009simulating} to make rebit simulations local. Here we use it to prove the second statement, removing the depth blowup of the first statement at the cost of using more rebits. The idea is quite simple: encode the ancilla rebit so that some logical Pauli operator (in this case, logical $Y$) is accessible $n$ times in parallel. The obvious code for this is the classical redundancy code. Let $\ket{\bar0},\ket{\bar1}$ be the encoded $\ket{0}$ and $\ket{1}$. These should be $+1$-eigenstates of the stabilizers $Y_jY_{j+1}$ for all $j=1,2,\dots,n-1$ and the $\pm1$-eigenstates of the encoded $Z$ operator $\bar Z=Z_1Z_2\dots Z_n$. Working it out, the states are
\begin{equation}
\ket{\bar0}=\frac{1}{\sqrt{2^{n-1}}}\sum_{\stackrel{x\in\{0,1\}^n}{|x|\text{ even}}}(-1)^{|x|/2}\ket{x},\quad \ket{\bar1}=\frac{1}{\sqrt{2^{n-1}}}\sum_{\stackrel{x\in\{0,1\}^n}{|x|\text{ odd}}}(-1)^{(|x|-1)/2}\ket{x}.
\end{equation}
where $|x|$ is the Hamming weight of the string $x$, i.e. the number of 1's in $x$. From Proposition~\ref{prop:exampleslinear}, the simulation of $T_i$ for any $i=1,2,\dots,n$ is
\begin{equation}
CH_{ia}CZ_{ia}=\frac12(I+Z_i)+\frac12(I-Z_i)\frac{1}{\sqrt2}(I-iY_a).
\end{equation}
Since encoded $Y$ is $\bar Y=Y_{ai}$ for any qubit $ai$ in the encoded ancilla, we can modify the simulation of $T_i$ to
\begin{equation}
CH_{ia}CZ_{ia}=\frac12(I+Z_i)+\frac12(I-Z_i)\frac{1}{\sqrt2}(I-iY_{ai}).
\end{equation}
With this modification, the simulations of $T_i$ and $T_j$ for $i\neq j$ are orthogonal operators with disjoint support and can be performed in parallel.
\end{proof}

Next we discuss efficiency of the rebit bottom-up simulation in the same manner, i.e.~we are concerned with the simulation of an $\bbR$-unitary circuit on $n$ qubits with depth $d$. The notion of depth is not immediately obvious for circuits constructed from the gates $\{H,CCZ,CK,G(\pi/4)\}$, but we can use the following definition. This gives us a well-defined notion of depth that leads to a theorem similar to Theorem~\ref{thm:topdown_effic}. 
\begin{definition}\label{def:Rcirc_depth}
An $n$-qubit, depth-$1$ $\bbR$-unitary circuit consists of gates $G_i$, $i=1,2,\dots,s$ such that $[G_i,G_j]=0$ (i.e.~all gates mutually commute) and all $q\in\{1,2,\dots,n\}$ is in the support of exactly one gate $G_i$. A depth-$d$ $\bbR$-unitary circuit equals the product of $d$ depth-$1$ $\bbR$-unitary circuits.
\end{definition}
\noindent Gates $\{H,CCZ,CK,G(\pi/4)\}$ have supports of sizes $1,3,1,0$ respectively, and, when all supports are disjoint, only $CK$ and $G(\pi/4)$ do not commute. Definition~\ref{def:Rcirc_depth} generalizes the notion of circuit depth from the unitary to the $\bbR$-unitary case, because two unitary gates having disjoint support implies that they commute.

We have the following theorem.
\begin{theorem}\label{thm:bottomup_effic}
Let $C$ be an $n$-qubit, depth-$d$ $\bbR$-unitary circuit using gates from $\{H,CCZ,CK,G(\pi/4)\}$. Then $C$ (applied to $\ket{0}^{\otimes n}$) can be simulated (i.e.~we can make $\mathcal{P}(C\ket{0}^{\otimes n})$) using either
\begin{enumerate}
\item a circuit of depth at most $dn$ on $n+1$ rebits
\item or a circuit of depth at most $d\lceil\log_2n\rceil$ on $2n$ rebits.
\end{enumerate}
\end{theorem}
\begin{proof}
The first part follows the same reasoning as the proof of the first part of Theorem~\ref{thm:topdown_effic}. In this case, the $1$-qubit $CK$ gates are the problem. At most $n$ $CK$ can occur in parallel, but to simulate each we need access to the ancilla ($CK_i$ is simulated by $CZ_{ia}$ by Proposition~\ref{prop:examplesNonlinear}).

The second part also follows similar reasoning to that of Theorem~\ref{thm:topdown_effic}. We encode the rebit in a redundancy code, although this time one in which the encoded $Z$ operator is accessible in parallel. The code states are simply the redundancy states
\begin{equation}
\ket{\bar0}=\ket{0}^{\otimes n},\quad\ket{\bar1}=\ket{1}^{\otimes n}.
\end{equation}
With this, the encoded $Z$ is $\bar Z=Z_{ai}$ for any qubit in the ancilla $ai$ and encoded $Y$ is $\bar Y=X_{a1}X_{a2}\dots X_{a,n-1}Y_{an}$. Thus, $CK_i$ and $CK_j$ gates (for $i\neq j$) can be simulated in parallel. However, $G(\pi/4)$ requires depth $\lceil\log_2n\rceil$ to simulate: decode the states (i.e.~$\ket{\bar 0}\rightarrow\ket{0}^{\otimes n}$ and $\ket{\bar 1}\rightarrow\ket{1}\ket{0}^{\otimes n-1}$ by $\lceil\log n\rceil$ timesteps of $CX$ gates), apply a $Y$-rotation to the first qubit, and re-encode.
\end{proof}

\section{The $\bbR$-Clifford hierarchy}\label{sec:RClifford_hierarchy}
An intriguing consequence of rebit simulation is an extension of the standard Clifford hierarchy of unitary operators into the more general $\mathbb{R}$-unitaries. At the second level of this $\mathbb{R}$-Clifford hierarchy, we obtain an extension of the famous Gottesman-Knill theorem, allowing us to efficiently classically simulate $\mathbb{R}$-linear quantum circuits of a restricted class that is analogous to but larger than the standard Clifford circuits.

We begin by defining the Clifford hierarchy and the $\mathbb{R}$-Clifford hierarchy. 
The standard Pauli group (on $n$-qubits) is
\begin{equation}\label{eq:def_Paulis}
\mathcal{C}_n(1)=\{e^{i\alpha}\big(p_1\otimes p_2\otimes\dots\otimes p_n\big):p_j\in\{I,X,Y,Z\},\alpha\in\mathbb{R}\}.
\end{equation}
An analogous Pauli group incorporating complex conjugation may be defined as
\begin{equation}\label{eq:def_RPaulis}
\mathcal{C}'_n(1)=\{i^c\big(p_1\otimes p_2\otimes\dots\otimes p_n\big)K^b:p_j\in\{I,X,Y,Z\},c\in\{0,1,2,3\},b\in\{0,1\}\},
\end{equation}
which we, for now, call simply the primed Paulis. Shortly, we show the primed Paulis are actually the $\bbR$-Paulis.

Notice that the two definitions have different global phases --- in the case of $\mathcal{C}_n(1)$ the phase is arbitrary, while for $\mathcal{C}'_n(1)$ it is restricted to powers of $i$. This is intentional and should be expected, because the rebit encoding tracks global phases.
Note that the definition of the Pauli group in \eq{eq:def_Paulis} differs from that defined in \cite{nielsen2010quantum}:
\begin{equation}\label{eq:def_Paulis_ikemike}
G_n(1)=\{i^c\big(p_1\otimes p_2\otimes\dots\otimes p_n\big):p_j\in\{I,X,Y,Z\},c \in \{0,1,2,3\}\}.
\end{equation}
We define the Pauli group differently simply to ease some of the arguments below (specifically, Lemma~\ref{lem:UUT_is_pauli}). Allowing arbitrary phases via definition $\mathcal{C}_n(1)$ is also more consistent with operators in higher levels of the hierarchy having arbitrary phases as well.

Now, we appeal to the discussion of bottom-up simulation in Section~\ref{sec:bottomuprebits}, taking the Paulis $\mathcal{C}_n(1)$ as the set of operators $\texttt S$. The real Paulis on $n$ rebits is the set of orthogonal Paulis $\mathcal{C}^\mathbb{R}_n(1)=\mathcal{C}_n(1)\cap R_n=\mathcal{C}_n(1)\cap T_n$, and the $\bbR$-Paulis are $\bbR\mathcal{C}_n(1)=\mathcal{L}(\mathcal{C}^{\bbR}_n(1))$. The following theorem shows the primed Paulis are the $\bbR$-Paulis, thus establishing the Pauli-like description, Eq.~\eqref{eq:def_RPaulis}, of the set of $\bbR$-Paulis as appropriate.
\begin{theorem}\label{thm:encoded_RPaulis}
$\mathcal{L}(\mathcal{C}^{\bbR}_{n+1}(1)):=\bbR\mathcal{C}_n(1)=\mathcal{C}'_n(1).$ 
\end{theorem}
\begin{proof}

First, note that $p\in\mathcal{C}_{n+1}(1)$ is orthogonal if and only if $p$ contains an even number of Pauli $Y$s and a real phase (i.e.~$e^{i\alpha}$ from Eq.~\eqref{eq:def_Paulis} is $\pm1$) or $p$ contains an odd number of Pauli $Y$s and an imaginary phase ($\pm i$). Since $\mathcal{L}$ is a homomorphism, we need only consider its action on a basis set of orthogonal Paulis, namely $\{X_i,iY_i,Z_i,X_a,iY_a,Z_a\}$ where $i=1,2,\dots,n$ indicates a data qubit and $a$ indicates the rebit ancilla. We find
\begin{eqnarray}
\label{eq:first_pauli_mapping}
\mathcal{L}(X_i)=&X_i,\\
\mathcal{L}(iY_i)=&iY_i,\\
\mathcal{L}(Z_i)=&Z_i,\\
\mathcal{L}(X_a)=& iK, \\
\mathcal{L}(iY_a)=& -iI,\\
\mathcal{L}(Z_a)=&K,
\label{eq:last_pauli_mapping}
\end{eqnarray}
all of which are elements of $\mathcal{C}'_n(1)$. This shows $\mathcal{L}(\mathcal{C}^{\bbR}_{n+1})\subseteq\mathcal{C}'_n(1)$. However, the reverse direction, $\mathcal{L}(\mathcal{C}^{\bbR}_{n+1}(1))\supseteq\mathcal{C}'_n(1)$ follows from Eqs.~(\ref{eq:first_pauli_mapping}-\ref{eq:last_pauli_mapping}) as well. Let $p=i^cX^{\vec x}Z^{\vec z}K^b\in\mathcal{C}'_n(1)$, where $X^{\vec x}$ for $\vec x\in\{0,1\}^n$ means $\bigotimes_{j=1}^{n}X_j^{\vec x_j}$ and likewise for $Z^{\vec z}$. Then,
\begin{equation}
p=\mathcal{L}(X^{\vec x})\mathcal{L}(Z^{\vec z})\mathcal{L}(iY_a)^c\mathcal{L}(Z_a^b)=\mathcal{L}(X^{\vec x}Z^{\vec z}(iY_a)^cZ_a^b).
\end{equation}
Since $X^{\vec x}Z^{\vec z}(iY_a)^cZ_a^b$ is an orthogonal Pauli, this proves $\mathcal{C}'_n(1)\subseteq\mathcal{L}(\mathcal{C}^{\bbR}_{n+1}(1))$.
\end{proof}
The $\mathbb{R}$-Paulis are a group just as the standard Paulis are.
\begin{corollary}
$\bbR\mathcal{C}_n(1)$ is a group for all $n$.
\end{corollary}
\begin{proof}
Follows from Proposition \ref{prop:subgroupResult}.

\end{proof}

The upper levels of the standard Clifford hierarchy are defined recursively
\begin{equation}\label{eq:clifford_hierarchy}
\mathcal{C}_n(k)=\{U\in U_n:U\mathcal{C}_n(1)U^\dag\subseteq\mathcal{C}_n(k-1)\}.
\end{equation}
Note that the first level is the Pauli group and the second level is the Clifford group. In Appendix~\ref{app:p_sets_phases}, we show how using Pauli sets with different allowed global phases (e.g.~$G_n(1)$ instead of $\mathcal{C}_n(1)$) leads to the same Clifford hierarchy for $k\ge2$.

We would like a similar recursion to Eq.~\eqref{eq:clifford_hierarchy} to hold for the $\bbR$-Clifford hierarchy. Thus, we proceed similarly to the Pauli case above and define the \emph{primed Clifford hierarchy} as
\begin{equation}
\mathcal{C}'_n(k)=\{U\in\bbR L_n:U\left(\mathcal{C}'_n(1)\right)U^\dag\subseteq\mathcal{C}'_n(k-1)\}.
\end{equation}

Our goal now is to show that the primed Clifford hierarchy is equivalent to the $\bbR$-Clifford hierarchy. That is, $\mathcal{C}'_n(k)$ is exactly $\bbR\mathcal{C}_n(k):=\mathcal{L}(\mathcal{C}_n(k)\cap T_n)$. Because $\mathbb{R}$-unitaries are mapped to orthogonal operators in the physical rebit encoding, we find it natural to define the orthogonal Clifford hierarchy
\begin{eqnarray}
&\mathcal{D}_n(1)=\mathcal{C}_n(1)\cap T_n,\\
&\mathcal{D}_n(k)=\{U\in T_n:U\mathcal{D}_n(1)U^T\subseteq\mathcal{D}_n(k-1)\}.
\end{eqnarray}
Then, it is worth noting the following definition of the orthogonal hierarchy as the real Cliffords.
\begin{lemma}
$\mathcal{D}_n(k)=\mathcal{C}_n(k)\cap T_n:=\mathcal{C}^{\bbR}_n(k)$ for all $k$.
\end{lemma}
\begin{proof}
The second equality is a definition of notation. To prove the first, we proceed inductively, with $k=1$ already satisfying the claim by definition. Let $U\in\mathcal{C}_n(k)\cap T_n$. Then $U\mathcal{D}_n(1)U^T\subseteq U\mathcal{C}_n(1)U^T\subseteq\mathcal{C}_n(k-1)$ and $U\mathcal{D}_n(1)U^T\subseteq T_n$ implying that $U\mathcal{D}_n(1)U^T\subseteq\mathcal{C}_n(k-1)\cap T_n=\mathcal{D}_n(k-1)$ and thus, $U\in\mathcal{D}_n(k)$ by definition.

In the other direction, let $U\in\mathcal{D}_n(k)$. We notice that for all $p\in\mathcal{C}_n(1)$, there exists a phase $e^{-i\alpha}$ such that $e^{-i\alpha}p\in\mathcal{C}_n(1)\cap T_n=\mathcal{D}_n(1)$. This is because $p=e^{i\alpha}X^{\vec x}Z^{\vec z}$ implies $pp^T=e^{i2\alpha}I$. Thus, $U\mathcal{D}_n(1)U^T\subseteq\mathcal{D}_n(k-1)=\mathcal{C}_n(k-1)\cap T_n$ implies $U\mathcal{C}_n(1)U^T\subseteq\mathcal{C}_n(k-1)$. So $U\in\mathcal{C}_n(k)$ by definition, and thus $U\in\mathcal{C}_n(k)\cap T_n$.
\end{proof}

Now we are in a position to prove that the primed hierarchy is indeed the $\bbR$-Clifford hierarchy.
\begin{theorem}\label{thm:encoded_hierarchy}
$\mathcal{L}(\mathcal{C}^{\bbR}_{n+1}(k)):=\bbR \cC_n(k)=\mathcal{C}'_n(k)$ for all $k$.
\end{theorem}
\begin{proof}
The $k=1$ case is proven in Theorem~\ref{thm:encoded_RPaulis}. For the rest, we proceed inductively. 

Let $U\in\mathcal{C}^{\bbR}_{n+1}(k)$. Then $U(\mathcal{C}_{n+1}^{\bbR}(1))U^T\subseteq\mathcal{C}^{\bbR}_{n+1}(k-1)$. Taking $\mathcal{L}$ of both sides, we get
\begin{equation}
\mathcal{L}(U)\mathcal{C}'_n(1)\mathcal{L}(U)^\dag=\mathcal{L}(U)\mathcal{L}(\mathcal{C}^{\bbR}_{n+1}(1))\mathcal{L}(U)^\dag=\mathcal{L}(U\mathcal{C}^{\bbR}_{n+1}(1)U^T)\subseteq\mathcal{L}(\mathcal{C}^{\bbR}_{n+1}({k-1}))=\mathcal{C}'_n(k-1)
\end{equation}
using the inductive hypothesis. Thus, $\mathcal{L}(U)\in\mathcal{C}'_n(k)$, showing $\mathcal{L}(\mathcal{C}^{\bbR}_{n+1}(k))\subseteq\mathcal{C}'_n(k)$.

For the other direction, let $U\in\mathcal{C}'_n(k)$. Then,
\begin{equation}
U\mathcal{L}(\mathcal{C}^{\bbR}_{n+1}(1))U^\dag= U\mathcal{C}'_n(1)U^\dag\subseteq\mathcal{C}'_n(k-1)=\mathcal{L}(\mathcal{C}^{\bbR}_{n+1}(k-1)).
\end{equation}
Taking $\mathcal{P}$ of both sides, we find $\mathcal{P}(U)\mathcal{C}^{\bbR}_{n+1}(1)\mathcal{P}(U)^T\subseteq\mathcal{C}^{\bbR}_{n+1}(k-1)$, which implies $\mathcal{P}(U)\in\mathcal{C}^{\bbR}_{n+1}(k)$. Thus, $U\in\mathcal{L}(\mathcal{C}^{\bbR}_{n+1}(k))$ and showing $\mathcal{C}'_n(k)\subseteq\mathcal{L}(\mathcal{C}^{\bbR}_{n+1}(k))$.
\end{proof}




There are some corollaries of Theorem~\ref{thm:encoded_hierarchy}. For instance, just as standard Cliffords form a group, so do the $\mathbb{R}$-Cliffords.
\begin{corollary}
$\bbR\mathcal{C}_n(2)$ is a group.
\end{corollary}
\begin{proof}
Follows from Proposition \ref{prop:subgroupResult}.



\end{proof}

While the group $G_n(1)$ defined in \eq{eq:def_Paulis_ikemike} is a subgroup of $\cC_n'(1) = \bbR \cC_n(1)$, the groups $\cC_n(1)$ and $\bbR \cC_n(1)$ are incomparable:

\begin{prop}
\label{prop:C1notcomparable}
$\cC_n(1) \not\subseteq \bbR \cC_n(1)$ and $\bbR \cC_n(1) \not\subseteq \cC_n(1)$.
\end{prop}
\begin{proof}
The first noninclusion follows from the fact that the operator $e^{i\pi/4} I$ is in $\cC_n(1)$ but not in $\bbR \cC_n(1)$.
The second noninclusion follows from the fact that $K$ is in $\bbR \cC_n(1)$ but not in $\cC_n(1)$.

\end{proof}

Furthermore, the groups $\cC_n(2)$ and $\bbR \cC_n(2)$ are also incomparable:

\begin{prop}
\label{prop:C2notcomparable}
$\cC_n(2) \not\subseteq \bbR \cC_n(2)$ and $\bbR \cC_n(2) \not\subseteq \cC_n(2)$.
\end{prop}
\begin{proof}
The first noninclusion follows from the fact that the operator $e^{i\pi/8} I$ is in $\cC_n(2)$ but not in $\bbR \cC_n(2)$. To see the latter, note that $e^{i\pi/8} K e^{-i\pi/8} = e^{i\pi/4} K \not\in \bbR \cC_n(1)$.
The second noninclusion follows from the fact that $K$ is in $\bbR \cC_n(2)$ but not in $\cC_n(2)$.

\end{proof}

However, Propositions \ref{prop:C1notcomparable} and \ref{prop:C2notcomparable} notwithstanding, it turns out that if we disregard global phases, then the groups $\cC_n(k)$ and $\bbR \cC_n(k)$ (for $k=1,2$) are no longer incomparable. More precisely, for at least the first two levels, 
the $\mathbb{R}$-Clifford hierarchy $\bbR \cC_n(k)$ is strictly larger than the standard hierarchy $\cC_n(k)$. A lemma regarding some structure of the standard Pauli and Clifford groups helps us show this.

\begin{lemma}\label{lem:UUT_is_pauli}
$\{UU^T:U\in\mathcal{C}_n(2)\}=\{p:p=p^T\in\mathcal{C}_n(1)\}\subseteq\mathcal{C}_n(1).$
\end{lemma}
\begin{proof}
That any symmetric Pauli $p=p^T\in\mathcal{C}_n(1)$ can be written as $UU^T$ for some $U\in\mathcal{C}_n(2)$ is not hard to see. Let $p=e^{i\alpha}X^{\vec x}Z^{\vec z}$ with $\vec x\cdot\vec z=0$, enforcing the symmetry of $p$. Let $J_x=\{j:\vec x_j=1-\vec z_j=1\}$, $J_y=\{j:\vec x_j=\vec z_j=1\}$, $J_z=\{j:1-\vec x_j=\vec z_j=1\}$. Since $|J_y|$ is even, we can partition it into two equal sized subsets $J_y^1$ and $J_y^2$ with a one-to-one mapping $\sigma:J_y^1\rightarrow J_y^2$. Let
\begin{equation}
U=e^{i\alpha/2}\left(\prod_{i\in J_y^1}CX_{i\sigma(i)}\right)\left(\prod_{j\in J_x\cup J_y^1} H_j\right)
\left( \prod_{k \in J_x \cup J_y \cup J_z} \hspace{-5pt} S_k \right).
\end{equation}
One can now calculate that 
 \begin{eqnarray}
 UU^T &=& e^{i\alpha}\left(\prod_{i\in J_y^1}CX_{i\sigma(i)}\right)\left(\prod_{j\in J_x\cup J_y^1} H_j\right)
 \left( \prod_{k \in J_x \cup J_y \cup J_z} \hspace{-5pt} Z_k \right) \left(\prod_{j\in J_x\cup J_y^1} H_j\right) \left(\prod_{i\in J_y^1}CX_{i\sigma(i)}\right) \nn
&=& 
 e^{i\alpha} \left(\prod_{i\in J_y^1}CX_{i\sigma(i)}\right)\left(\prod_{j\in J_x \cup J_y^1} X_j \right)
 \left(\prod_{k\in J_y^2 \cup J_z} Z_k \right)
\left(\prod_{i\in J_y^1}CX_{i\sigma(i)}\right) \nn
&=& 
e^{i\alpha} \left(\prod_{j\in J_x \cup J_y} X_j \right)
 \left(\prod_{k\in J_y \cup J_z} Z_k \right) \nn
 &=& e^{i\alpha} X^{\vec x} Z^{\vec z} = p.
\end{eqnarray}
For the other direction, we need $UU^T\in\mathcal{C}_n(1)$ for any $U\in\mathcal{C}_n(2)$, but it suffices\footnote{Here, we show that if $VpV^\dag = \pm p$ for all $p \in \cC_n(1)$, then $V \in \cC_n(1)$. Indeed, since the Paulis form a basis, we may write 
$V=\sum_{q\in P_n}\alpha_qq$ for complex coefficients $\alpha_q$, where 
$P_n \ =\{p_1\otimes p_2\otimes\dots\otimes p_n :p_j\in\{I,X,Y,Z\} \}$ is the set of Pauli operators without global phases.
Let $p\in\mathcal{C}_n(1)$ be supported on one qubit (i.e.~$p=X_i,Y_i,Z_i$ for qubit $i$). Then $VpV^\dag=\pm p$ implies that exactly half of the coefficients $\alpha_q$ are zero (e.g.~in the case of $+$, all $\alpha_q$ such that $\{p,q\}=0$ are zero). Repeating for all qubits $i$ and all $X_i,Y_i,Z_i$, all but one coefficient are zeroed, thus implying $V\in\mathcal{C}_n(1)$.} to show that for all $p\in\mathcal{C}_n(1)$, $(UU^T)p(UU^T)^\dag=\pm p$ where the choice of sign may depend on $p$.
In fact, it suffices to show that 
\begin{equation} \label{eq:toshowUUT}
(UU^T)p(UU^T)^\dag=ap
\end{equation}
for any complex number $a$. This is because, squaring both sides and using the fact that $p^2\propto I$ for all $p\in\mathcal{C}_n(1)$, we get $(UU^T)p^2(UU^T)^\dag=p^2=a^2p^2$, implying $a=\pm1$.

We now prove \eq{eq:toshowUUT}. Notice that for any $q \in \cC_n(1)$, 
\begin{equation} \label{eq:pproptoq}
q \propto \bar q.
\end{equation}
Also, notice that since $\cC_n(2)$ is closed under taking inverses, $U^\dag \in \cC_n(2)$, which implies that $U^\dag p U \in \cC_n(1)$.

Hence,
\begin{eqnarray}
(UU^T)p(UU^T)^\dag &=&
UU^T p\bar UU^\dag \nn
&\propto & UU^T \bar p\bar UU^\dag ,\quad \mbox{by applying \eq{eq:pproptoq} to $p \in \cC_n(1)$.} 
\nn
&=& U \left(\overline{U^\dag p U}\right) U^\dag \nn
& \propto & U(U^\dag p U)U^\dag ,\quad \mbox{by applying \eq{eq:pproptoq} to $U^\dag p U\in \cC_n(1)$} 
\nn
&=& p.
\end{eqnarray}

\end{proof}

We write $A\subsetsim B$ if for all $a\in A$, there exists $\alpha\in\mathbb{R}$ such that $e^{i\alpha}a\in B$.

\begin{prop}
\label{prop:inclusion_properties_Clifford}
$G_n(1)\cup G_n(1)K=\bbR\mathcal{C}_n(1)$, $\mathcal{C}_n(1)\subsetsim\bbR\mathcal{C}_n(1)$, and $\mathcal{C}_n(2)\subsetsim\bbR\mathcal{C}_n(2)$.
\end{prop}
\begin{proof}
The first and second statements follow easily from definitions Eqs.~(\ref{eq:def_Paulis}), (\ref{eq:def_RPaulis}), and (\ref{eq:def_Paulis_ikemike}). It remains to prove the third inclusion. Let $U \in \cC_n(2)$. Then by 
Lemma~\ref{lem:UUT_is_pauli}, $UU^T$ is Pauli, i.e. 
\begin{equation}
UU^T = e^{i\beta} X^{\vec x}Z^{\vec z}
\end{equation}
for some $\beta \in \bbR$ and $\vec x,\vec z \in \{0,1\}^n$. To complete the proof, it suffices to find some $\alpha \in \bbR$ such that $e^{i\alpha} U \in \bbR \cC_n(2)$. To this end, we choose $\alpha = -\beta/2$, and show that $U' := e^{-i\beta/2} U \in \bbR \cC_n(2)$.

Let $\zeta = i^c X^{\vec x}Z^{\vec z} K^b \in \bbR \cC_n(1)$ be arbitrary. If $b=0$, then
\begin{eqnarray}
\label{eq:b_equal_0_case}
U' \zeta U'^\dag &=& e^{-i\beta/2} U \big(i^c X^{\vec x}Z^{\vec z}\big) e^{i\beta/2} U^\dag \nn
&=&
U \big(i^c X^{\vec x}Z^{\vec z}\big) U^\dag
\nn
&\in & G_n(1) \subseteq \bbR \cC_n(1),
\end{eqnarray}
where, since $U\in\mathcal{C}_n(2)$ is Clifford and $\zeta=i^cX^{\vec x}Z^{\vec z}\in G_n(1)$, Theorem~\ref{thm:cliffordhierarchyindependence} in Appendix~\ref{app:p_sets_phases} shows $U\zeta U^\dag\in P_n$. If $b=1$, then
\begin{eqnarray}
\label{eq:b_equal_1_case}
U' \zeta U'^\dag &=& e^{-i\beta/2} U \big(i^c X^{\vec x}Z^{\vec z} K\big) e^{i\beta/2} U^\dag \nn
&=&
e^{-i \beta} U \big(i^c X^{\vec x}Z^{\vec z}\big) U^T K
\nn
&=&
e^{-i \beta} U \big(i^c X^{\vec x}Z^{\vec z}\big) U^\dag (U U^T)K
\nn
&=&
e^{-i \beta} U \big(i^c X^{\vec x}Z^{\vec z}\big) U^\dag e^{i\beta} X^{\vec x}Z^{\vec z} K
\nn
&=&
U \big(i^c X^{\vec x}Z^{\vec z}\big) U^\dag X^{\vec x}Z^{\vec z} K
\nn
&\in & G_n(1)K \subseteq \bbR \cC_n(1)
\end{eqnarray}
where the last line follows because $U \big(i^c X^{\vec x}Z^{\vec z}\big) U^\dag \in G_n(1)$ and $X^{\vec x}Z^{\vec z} \in G_n(1)$. Since $G_n(1)$ is closed under multiplication, the expression in the second-to-last line of \eq{eq:b_equal_1_case} is of the form \eq{eq:def_RPaulis}.
\end{proof}




A final corollary of Theorem~\ref{thm:encoded_hierarchy} corollary is a Gottesman-Knill-esque efficient classical simulation of $\bbR\mathcal{C}_n(2)$ circuits and a generating set for them.
\begin{corollary}
\label{cor:GKforRcircuits}
Let $U\in\bbR\mathcal{C}_n(2)$ be an $n$-qubit $\mathbb{R}$-linear operator. Then $U$ can be constructed from $O(n^2)$ gates from $\{H,S,K,CX,CK\}$. Moreover, a classical computer can sample from $U\ket{0}^{\otimes n}$ in time $O(n^2)$.
\end{corollary}
\begin{proof}
The first statement can be proved by compiling $V=\mathcal{P}(U)\in\mathcal{C}^{\bbR}_{n+1}(2)$ using gates from the orthogonal Clifford group $\{H,Z,CX\}$. That this is possible with the requisite number of gates follows the same Clifford compiling argument from \cite{nielsen2010quantum} (see their Theorem 10.6). 

The first step is to argue that $\{H,Z\}$ generate all single-qubit orthogonal Cliffords. A single-qubit orthogonal Clifford is uniquely specified up to a phase (which for orthogonal operators is just $\pm1$) by its action on the Paulis $\{X,iY,Z\}$. However, since $\det(X)=\det(Z)=-1$ and $\det(iY)=1$, orthogonal Cliffords must map $\{X,Z\}$ to $\{\pm X,\pm Z\}$ and $\{iY\}$ to $\{\pm iY\}$. That appropriate sequences of $H$ and $Z$ can achieve all these mappings can be checked directly by enumeration. Moreover, the phase $\pm1$ can be provided by $(HZ)^4=-I$.

The second part of the proof is inductive on the number of rebits $n+1$. There are $O(n)$ recursive steps, each using $O(n)$ gates from our gate set $\{H,Z,CX\}$. Say $VX_1V^T=g$ and $VZ_1V^T=h$ for $g,h\in\mathcal{C}^{\bbR}_{n+1}(1)$. There are simplifications that can be made without loss of generality, however. First, we note that $g$ and $h$ anticommute, which means that on some qubit $j$, the Paulis there locally anticommute. We can apply a SWAP $P$ between qubits $1$ and $j$ (which can be constructed from three $CX$ gates), so that we get $V'X_1V'^T=p_1\otimes g'$ and $V'Z_1V'^T=p_2\otimes h'$ with $p_1,p_2\in\mathcal{C}_1(1)$, $g',h'\in\mathcal{C}_n(1)$, $\{p_1,p_2\}=0$, $[g',h']=0$, and instead compile $V'=PV$. Second, by applying $H$ before or after $V'$, we can switch the roles of $X_1$ and $Z_1$ and change $p_1,p_2$. The final result is that, without loss of generality, we have two cases to consider: either
\begin{equation}
VX_1V^T=X\otimes g,\quad VZ_1V^T=Z\otimes h
\end{equation}
or
\begin{equation}
VX_1V^T=(iY)\otimes g,\quad VZ_1V^T=Z\otimes h
\end{equation}
where in both cases $g,h\in\mathcal{C}_n(1)$ and $[g,h]=0$.

We claim the circuit in Fig.~\ref{fig:compiling_circ} implements these two cases and does so using the allowed gates $\{H,Z,CX\}$. The controlled $n$-qubit orthogonal Paulis (e.g.~controlled-$g$ from qubit $c$ to qubits $t_1,t_2,\dots,t_{n}$) can be implemented by a depth at most $2n$ circuit of $CX$ and $CZ=(I\otimes H)CX(I\otimes H)$ gates. If $g=g_1\otimes g_2\otimes\dots g_n$ where each $g_j\in\{X,iY,Z\}$, then performing controlled-$g_j$ from qubit $c$ to $t_j$ for all $j$ (in any order; they commute) implements controlled-$g$. Controlled-$X$ and controlled-$Z$ operators are simply $CX$ and $CZ$, while controlled-$iY$ is $CX$ followed by $CZ$. Finally, the circuit in Fig.~\ref{fig:compiling_circ} guarantees the correct behavior of $V$ on the first qubit and thus $\tilde V$ is an $n$ qubit orthogonal Clifford, which can be compiled using the same process. The recursion continues until the base case of 1-qubit, discussed earlier.

\begin{figure}
\begin{equation*}
\Qcircuit @C=1em @R=1em {
& \qw      & \gate H             &  \ctrl 1   & \gate H & \ctrl 1  & \qw    \\ 
& / \qw & \gate {\tilde V} & \gate h  & \qw & \gate g   & \qw &  \\
} 
\end{equation*}

\caption{\label{fig:compiling_circ} Compiling an orthogonal Clifford circuit on $n$ qubits. The top line represents $1$ qubit while the bottom represents $n-1$ qubits.}
\end{figure}
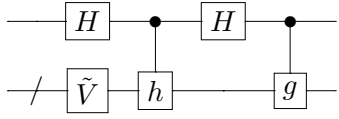

Given the described compilation, CHP simulation \cite{aaronson2004improved} of the orthogonal Clifford circuit for $V$ in the rebit encoding suffices to simulate the $\mathbb{R}$-Clifford $U$, and does so in time $O(n^2)$.
\end{proof}

\section{Discussion and open questions}\label{application}

Our bottom-up simulation paradigm provides a unified framework for realizing any arbitrary antiunitary or partial antiunitary transformations  which are otherwise non-physical and cannot be simulated directly. The rebit simulation can be applied to     measurements of a large variety of quantum mechanical properties as well as the detection and simulation of exotic phases of matters. These applications all require  either antiunitary or partial antiunitary transformations that are unphysical for  quantum mechanical systems.    

For example,   evaluating  the  entanglement of an arbitrary partition of a generic quantum system  usually takes $O(2^N)$  measurements for an $N$-qubit system~\cite{horodecki2001separability}. To avoid such a resource overhead, entanglement monotones such as concurrence~\cite{Wootters1998} and 3-tangle~\cite{Cirac2000} are proposed to provide  convex and monotonic measures that do not increase under local operations and classical communication. However, both the concurrence and 3-tangle are defined by the expectation values of an antiunitary operator which cannot be directly measured. To directly measure these entanglement  monotones,  an extra qubit is needed to simulate the complex conjugation on the original system~\cite{Solano2013}. Our result on partial antiunitary simulation thus further generalizes such approaches to larger systems where the concurrence of   only a subsystem can be measured.

As another example, time-reversal symmetry and particle-hole symmetry are two important ingredients for defining either bosonic or fermionic symmetry protected topological phases~\cite{Wen2013,Wen2014}. The system symmetry  is defined by  the invariance of the system Hamiltonian $\hat{H}$ under the conjugation of the corresponding antiunitary tranformation  $U_c$  for either the time-reversal or particle-hole symmetry as: $U_c\hat{H}U_c^\dagger=\hat{H}$. Partial-time-reversal and partial-particle-hole symmetries corresponding to the invariance under partial antiunitary transformations are also  used for constructing nonlocal order parameters in detecting fermionic symmetry protected topological phases in (1+1) dimension~\cite{Shinsei2017}.  Detection of these symmetry protected topological phases is exceedingly hard since these symmetry operators are non-physical and cannot be directly measured. Being able to   simulate both antiunitary and  partial antiunitary transformation with our rebit encoding can potentially simplify the detection procedure for topological phases proposed in~\cite{Turner2012}.

Our results add new tools to  the existing dictionary of quantum simulation gadgets using qubits. The   Majorana equation~\cite{Solano2011}, for example, is one candidate for describing the dynamics of neutrino or other particles outside the standard model. Simulating the Majorana equation in quantum systems necessitates the application of the    complex conjugation of the wave function, which is readily available in our bottom-up rebit simulation.


We conclude this section by listing a few directions that an extension of this project might take.
\newline\newline\noindent
{\bf Mixed states.} Our treatment in this paper has been restricted to just pure states, unitary transformations and projective measurements. This suffices since we could always ``go to the Church of the Larger Hilbert Space\footnote{Coined by John Smolin \cite{church}.}" by considering mixed states as being part of a larger system described by a pure state. Nevertheless, describing quantum systems using the smaller Hilbert Space has also proven to be fruitful, as it allows for the study of noisy quantum systems without any reference to a fictitious external system. Our bottom-up approach to rebits might benefit from such an approach. What is the rebit generalization of completely positive and trace preserving maps? Can they be described in terms of some generalized Kraus operators?
\newline\newline\noindent
{\bf Quaternions.} In this paper, we studied the relationship between computing using real and complex amplitudes. But the real numbers and complex numbers are just the two base levels of the Cayley-Dickson construction \cite{schafer1966introduction}. The next level of the construction are the quaternions, which was studied in the context of computation by \cite{fernandez2003quaternionic}. It would be interesting to apply  the techniques from our paper to 
study computing based on quaternions (or even other levels of the Cayley-Dickson construction) from a bottom-up perspective. 
\newline\newline\noindent
{\bf Compiling.} We showed (Theorem~\ref{thm:logicalActionOfOrthogonals}) that an arbitrary $\bbR$-unitary $\Gamma$ can always be written as products of partial antiunitaries and further noted that these products always take the form of an alternating sequence $U_NK_{L_N}U_{N-1}K_{L_{N-1}}\dots U_1K_{L_1}U_0$ of unitaries $U_j$ and partial complex conjugations $K_{L_j}$ over languages $L_j$. But given $\Gamma$ and desired accuracy $\epsilon$, how efficient is it to determine the length of the sequence required to approximate $\Gamma$ to within operator norm $\epsilon$ and also the specific unitaries and languages? In principle, applying Solovay-Kitaev \cite{dawson2005solovay} in the simulator space $P$ provides an algorithm and upper bounds, but it is well-known even in unitary compilation that Solovay-Kitaev is not optimal. The exact question $\epsilon=0$ is also interesting and leads to the definition of a minimum $N$ for which exact compilation of $\Gamma$ (call it e.g.~the ``conjugation depth" of $\Gamma$) is possible. For instance, the conjugation depth of any unitary is zero, the partial antiunitaries have conjugation depth one by definition, and Theorem~\ref{thm:partial_antiunitaries_not_group} shows that some $\bbR$-unitaries have conjugation depth at least two. Ideally, the conjugation depth of $\Gamma$ might be determined from some simple property of $\Gamma$.
\newline\newline\noindent
{\bf Clifford hierarchy.} In our discussion of the Clifford hierarchy, we have focused on the first two levels of the hierarchy. For example, in Proposition \ref{prop:inclusion_properties_Clifford}, we showed that, up to a global phase, the $\bbR$-Clifford hierarchy (for the first two levels) is bigger than the standard Clifford hierarchy. Does an analogous result hold  for higher levels of the hierarchy? Next, we see from the definitions in \eq{eq:def_Paulis} and \eq{eq:def_RPaulis} that to get from $\cC_n(1)$ to $\bbR \cC_n(1)$, we need to append the $K$ gate to the list of generators of the Pauli group.
Also, we see from Corollary \ref{cor:GKforRcircuits} that to get from $\cC_n(2)$ to $\bbR \cC_n(2)$, we need to append both the $K$ and $CK$ gate to the list of generators of the Clifford group. Can the $k$th level (for $k>2$) of the $\bbR$-Clifford hierarchy be obtained by appending gates to the corresponding level of the standard Clifford hierarchy?

\section*{Acknowledgments}
 DEK is supported by the National Science Scholarship from the Agency for Science, Technology and Research (A*STAR). MYN acknowledges support from ONR grant number N00014-13-1-0774 and AFOSR grant number FA9550-14-1-0052. TJY is supported by the Department of Defense (DoD) through a National Defense Science and Engineering Graduate (NDSEG) fellowship, and an IBM PhD fellowship award.

\bibliographystyle{ieeetr}
\bibliography{Bib}
\newpage
\appendix

\section{A simple motivating example}
\label{sec:simpleExample}

In this appendix, we present a simple motivating example to illustrate how nonunitary transformations can be simulated using the rebit encoding. Consider a general one-qubit state with complex amplitudes: $$\ket \psi = (a+ib) \ket 0 + (c+id) \ket 1,$$ where $a,b,c,d \in \mathbb R$. The single-ancilla rebit encoding is performed by introducing an additional register and encoding the state $\ket \psi$ as $$\ket {\psi' } = a \ket{00} + b \ket{01} + c \ket{10} + d \ket{11}.$$ 
To illustrate some nonlinear transformations that we can simulate using this rebit encoding, suppose we perfom $Z$ on the second qubit, getting the state $$ \ket{\chi'}=a\ket{00}-b\ket{01}+c\ket{10}-d\ket{11}, $$ which evidently is the rebit encoding of the complex conjugation of $\ket{\psi}$, $$\ket{\bar\psi}=(a-ib)\ket{0}+(c-id)\ket{1}.$$ Thus, via the rebit encoding, we have simulated the \emph{antiunitary} complex conjugation operation $K:\ket{\psi}\rightarrow\Re\ket{\psi}-i\Im\ket{\psi}$. Now, consider a more complicated example in which we perform a controlled-$Z$ operation on $\ket {\psi' }$ to get the state $$\ket{\phi'} = a \ket{00} + b \ket{01} + c \ket{10} - d \ket{11}.$$ We observe that $\ket{\phi'}$ is the rebit encoding of the state $$\ket \phi = (a+ib) \ket 0 + (c-id) \ket 1.$$ Hence, we have shown how to simulate the nonlinear transformation $$ (a+ib) \ket 0 + (c+id) \ket 1 \mapsto  (a+ib) \ket 0 + (c-id) \ket 1$$ using the rebit encoding. This transformation is an example of what we call a \textit{partial antiunitary} operator.
\newline

\section{Complex conjugation as a Gottesman-Knill simulation}\label{app:gottesman_knill}
In the introduction, we discussed how viewing Gottesman-Knill as a bottom-up simulation implies that a more ``advanced" Clifford quantum computer equipped with Pauli measurements that report entire probability distributions instead of just samples from them can be efficiently classically simulated as well. Here we expand on this bottom-up viewpoint, by showing that Gottesman-Knill also gives an efficient classical simulation of circuits consisting of Clifford gates and complex conjugation.

Gottesman-Knill is a bottom-up simulation $(L,P,\mathcal{P},O_P)$ where $L$ is the set of rank one density matrices
\begin{equation}\label{eq:GK_L}
L=\big\{\rho=\ket{\psi}\bra{\psi}:\exists p_j\in\mathcal{C}_n(1)\text{ s.t.~}\rho=\prod_{j=1}^n\frac12(I+p_j)\big\},
\end{equation}
with $\mathcal{C}_n(1)$ the Pauli group on $n$ qubits (see Eq.~\eqref{eq:def_Paulis}). The physical space $P$ is
\begin{equation}\label{eq:GK_P}
P=\{S:S\subseteq\mathcal{C}_n(1);|S|=n;e^{i\alpha}I\in\langle S\rangle\Rightarrow e^{i\alpha}=1\}
\end{equation}
where $\langle S\rangle$ for a set of Paulis $S$ is the group generated by $S$. A Pauli $p=\pm p_1\otimes p_2\otimes\dots\otimes p_n$ on $n$ qubits can be specified by $2n+1$ classical bits \cite{aaronson2004improved}. All Paulis $p\in S\in P$ must have this form (i.e.~with global phase such that $p^2=I$) by the final condition on $S$ in Eq.~\eqref{eq:GK_P}.

We next describe the map $\mathcal{P}$. Starting with a stabilizer state $\rho=\ket{\psi}\bra{\psi}$, identify $n$ linearly independent Paulis $p_i$ such that $p_i\ket{\psi}=\ket{\psi}$. These exist because $\rho$ can be written as in Eq.~\eqref{eq:GK_L}. Set $\mathcal{P}(\rho)=\{p_i:i=1,2,\dots,n\}$.



The bottom-up view stresses that we should define the set of operators $O_P$ that the simulator is capable of implementing. In this case, the most general operator the simulator can perform is any function $f:P\rightarrow P$. If we wanted to, we could also restrict to those functions $f$ that are efficiently computable (say, in polynomial time in $n$). Regardless, we do not have a characterization of $\mathcal{L}(O_P)$ for either of these choices of $O_P$. Our immediate goal is to show that complex conjugation (of density matrices) $K_{dm}:\ket{\psi}\bra{\psi}\rightarrow\Re\ket{\psi}\bra{\psi}-i\Im\ket{\psi}\bra{\psi}$ is an element of $\mathcal{L}(O_P)$ (for either choice of $O_P$) and thus that Gottesman-Knill simulation is intriguingly more powerful than typically imagined.

\begin{prop}
$K_{dm}$ is efficiently classically simulable by Gottesman-Knill.
\end{prop}
\begin{proof}
Suppose $S$ represents a density matrix $\ket{\psi}\bra{\psi}$. Let $\#Y(q)=|\{p_i=Y:i=1,2,\dots,n\}|$ for a Pauli $q=e^{i\alpha}q_1\otimes q_2\otimes\dots\otimes q_n$ and
\begin{equation}
S'=\{(-1)^{\#Y(p)}p:p\in S\}.
\end{equation}
We claim $S'=\mathcal{P}\left(K_{dm}\left(\ket{\psi}\bra{\psi}\right)\right)$. To show this, write $\ket{\psi}\bra{\psi}=\prod_{j=1}^n\frac12(I+p_j)$ where $p_j\in S$ for all $j$. Then $K_{dm}(\ket{\psi}\bra{\psi})=\prod_{j=1}^n\frac12(I+(-1)^{\#Y(p_j)}p_j)$. Thus, $(-1)^{\#Y(p_j)}p_j\in\mathcal{P}(K_{dm}(\ket{\psi}\bra{\psi}))$ for all $j$, which exactly matches the composition of $S'$.
\end{proof}

\section{$\mathbb R$-linear operators} \label{app:rlinear}


In this appendix, we consider operators of the form $A+BK$, where $A$ and $B$ are complex linear operators. Our first result is 
a proof of Theorem \ref{thm:Rlinearcharacterization} (restated here as Theorem \ref{thm:RlinearcharacterizationRestated}):
\begin{theorem} 
\label{thm:RlinearcharacterizationRestated}
(\cite{huhtanen2011real})
Let $V$ and $V'$ be complex vector spaces, and $f:V\rightarrow V'$ be a function on $V$. Then,
there exist linear maps $A$ and $B$ such that $f = A+BK$ if and only if \begin{equation}
f(ax+ by) = a f(x) + bf(y) \label{eq:Rlineareq}
\end{equation} 
for all $a,b \in \mathbb R$ and $x,y\in V$. 
\end{theorem}
\begin{proof}
The forward direction holds since $(A+BK)(ax+by) = a(A+BK)x+b(A+BK)y$ for all $x,y\in V$ and $a,b\in \mathbb R$. To prove the backward direction, assume that $f$ satisfies \eq{eq:Rlineareq}. Let $z \in V$. Let the standard basis of $V$ be $\{e_i\}_i$. Hence, we can write $z = \sum_j z_j e_j$ for some $z_j \in \mathbb C$.
Then,
\begin{eqnarray}
f(z) &=& f\left(\sum_j z_j e_j \right) \nonumber\\
&=&  f\left(\sum_i(\Re z_j + i \Im z_j) e_j\right) \nonumber\\
&=& \sum_j \Re z_j f(e_j) + \Im z_j f(i e_j), \qquad \mbox{by $\mathbb R$-linearity} \nonumber\\
&=& \sum_j \frac{z_j+ \bar z_j}2 f(e_j) + \frac{z_j - \bar z_j}{2i} f(i e_j) \nonumber \\
&=&  \sum_j \frac 12(f(e_j)- i f(i e_j)) z_j +  \sum_j \frac 12(f(e_j) + i f(i e_j)) \bar z_j \nonumber\\
&=& \sum_j A_j z_j + \sum_j B_j Kz_j \nonumber\\
&=& (A+BK) z
\end{eqnarray}
where $A$ is the (complex-valued) matrix whose $j$th column is $A_j = \frac 12(f(e_j)- i f(i e_j))$, and $B$ is the matrix whose $j$th column is $B_j = \frac 12(f(e_j) + i f(i e_j))$. Hence, $f = A+BK$, where $A$ and $B$ are (complex) linear maps on $V$.
\end{proof}

As pointed out in Section \ref{sec:characterizationOf}, in linear algebra, the term $\bbR$-linear is used to describe a map satisfying \eq{eq:Rlineareq}. Our terminology in this paper was chosen so that the two definitions of $\bbR$-linearity coincide. 

We conclude this section with a few remarks about $\bbR$-linear operators satisfying \eq{eq:Rlineareq}. First, note that linear operators and antilinear operators are both special cases of $\mathbb R$-linear operators. (A linear operator $g:V\rightarrow V$ is one that satisfies $g(ax+ by) = a g(x) + bg(y)$ for all $a,b \in \mathbb C$ and $x,y\in V$, and an antilinear operator $h$ is one that satisfies $h(ax+by) = \bar a h(x) + \bar b h(y)$ for all  $a,b \in \mathbb C$ and $x,y\in V$.)

Second, note that when $A$ and $B$ are complex linear operators, the operator $A$ is linear while the operator $BK$ is antilinear. Hence, Theorem \ref{thm:RlinearcharacterizationRestated} implies that any $\mathbb R$-linear operator can be written as a sum of a linear operator and an antilinear operator.

\section{The ring of $\mathbb R$-linear operators: algebraic properties} \label{app:algebraicProperties}
In Section \ref{app:rlinear}, we showed that every $\mathbb R$-linear operator on $\cH_n(\bbC)$ can be written as $A+BK$, where $A$ and $B$ are linear operators on $\cH_n(\bbC)$. In this appendix, we show that the set of $\mathbb R$-linear operators $\bbR L_n$ forms a ring with identity\footnote{For an introduction to ring theory, see \cite{dummit2004abstract}, for example.}, with addition $+$ given by
\begin{equation}
(A+BK) + (C+DK) = (A+C) + (B+D)K
\end{equation}
and multiplication $\star$ given by
\begin{equation} \label{eq:multiplstar}
(A+BK)\star (C+DK) = (AC + B\bar D) + (AD+B\bar C)K.
\end{equation}

\begin{prop} \label{prop:LringWithId}
Let $n \in \bbZ^+$. Then $(\bbR L_n, +, \star)$ is a ring with identity. The multiplicative identity is $I+0K$.
\end{prop}
\begin{proof}
It is straightforward to check that $\bbR L_n$ satisfies the properties of a ring with identity (see Chapter 7 of \cite{dummit2004abstract}).
\end{proof}

It is easy to check that $(\bbR L_n, +, \star)$ is neither a division ring nor a commutative ring, and hence is not a field. Note that the multiplication in \eq{eq:multiplstar} was defined so that for any vector $v \in \cH_n(\bbC)$, 
\begin{equation}
((A+BK)\star (C+DK)) v = (A+BK)((C+DK)v) .
\end{equation}

More generally, the set of vectors in $\cH_n(\bbC)$ forms a module over $\bbR L_n$, as the following proposition states.
\begin{prop} \label{prop:LLeftModule}
The set $\cH_n(\bbC)$ is a left module over the ring $(\bbR L_n, +, \star)$, with the addition on 
$\cH_n(\bbR)$ being the usual addition of functions, and the module action $\circ$ of $\bbR L_n$ on $\cH_n(\bbC)$  given by
$$(A+BK)\circ v := (A+BK) v = Av + B\bar v.$$
\end{prop}
\begin{proof}
It is straightforward to check that the above satisfies the properties of a left module (see Chapter 10 of \cite{dummit2004abstract}).
\end{proof}

It is useful to equip the ring $(\bbR L_n, +, \star)$ with the operator $\dag$ defined as follows: \begin{equation}
\label{eq:dagstar}
(A+BK)^\dag = A^\dag + B^T K.
\end{equation}

We now show that the $\dagger$ (super)operator is the image of the transpose map $(\cdot)^T$ under the rebit decoding. More precisely, let $\cE$ be an operator on $L_{n+1}^\bbR$. Define $\cL(\cE)$ to be the unique operator $\tilde \cE$ such that $$\tilde \cE(A+BK) = \cL(\cE(\cP(A+BK))).$$
It then follows that
\begin{prop}
\label{prop:LofTranspose}
$\cL\big((\cdot)^T\big) = (\cdot)^\dag$.
\end{prop}
\begin{proof} 
\begin{eqnarray}
A + BK &=& (\Re A + i \Im A)+(\Re B + i \Im B)K \nn
&\xrightarrow{\cP} &
\Re A \otimes I + \Im A \otimes XZ + \Re B \otimes Z + \Im B \otimes X \nn
&\xrightarrow{(\cdot)^T} &
\Re A^T \otimes I + \Im A^T \otimes ZX + \Re B^T \otimes Z + \Im B^T \otimes X \nn
&=& \Re A^T \otimes I + \Im (-A^T) \otimes XZ + \Re B^T \otimes Z + \Im B^T \otimes X \nn
&\xrightarrow{\cL} &
(\Re A^T - i \Im A^T) + (\Re B^T + i \Im B^T) K \nn
&=& A^\dag + B^T K = (A+BK)^\dag.
\end{eqnarray}
Hence, $\cL\big((\cdot)^T\big) = (\cdot)^\dag$.
\end{proof}

It is straightforward to check that the operator $\dag$ is an involutive antiautomorphism\footnote{Let $R$ be a ring, and let $*:R\rightarrow R$. $R$ together with $*$ is a $*$-ring if for all $x,y \in R$, (i) $(x^*)^* = x$, (ii) $(x+y)^* = x^* + y^*$, (iii) $(xy)^* = y^* x^*$. The map $*$ is called an \textit{involutive antiautomorphism}.} Hence, we obtain the following proposition.

\begin{prop}
The ring $(\bbR L_n, +, \star)$ together with the operator $\dag$ defined in \eq{eq:dagstar} is a $\dag$-ring.
\end{prop}

Note that the involutive antiautomorphism $\dag$ generalizes the definition of the adjoint of linear operators. Indeed, when $B=0$ in \eq{eq:dagstar}, we recover $A^\dag = (A+0K)^\dag = A^\dag$.
We may now generalize the notion of unitarity to $\dag$-rings with identity. Let $R$ be a $\dag$-ring with identity $1$. We say that $U \in R$ is a \textit{unitary element} with respect to $\dag$ if 
\begin{equation}
U^\dag U = 1.
\end{equation}
Applying the above definition to the $\dag$-ring $(\bbR L_n, +, \star)$, we get that an element $A+BK \in \bbR L_n$ is a unitary element if and only if 
\begin{equation}
(A+BK)^\dag \star (A+BK) = I .
\end{equation}
The group of unitary elements are called $\bbR U_n$ in the main text, as a result of Theorem~\ref{thm:OrthogonalABK2} showing that they are the simulated operators of a real unitary rebit simulator.

We now give an equivalent condition for the unitarity of $\bbR$-linear operators.
\begin{prop} \label{prop:unitaryElement}
An element $A+BK \in \bbR L_n$ is a unitary element with respect to $\dag$ if and only if
\begin{eqnarray}
A^\dag A + B^T \bar B = I ,\nn
A^\dag B + B^T \bar A = 0.
\label{eq:propUnitaryElement}
\end{eqnarray}
\end{prop}
\begin{proof}
An element $A+BK \in \bbR L_n$ is a unitary element if and only if $I = (A+BK)^\dag \star (A+BK) = (A^\dag + B^T K) \star (A+BK) = (A^\dag A + B^T \bar B) + ( A^\dag B + B^T \bar A )K$ if and only if $A^\dag A + B^T \bar B = I, A^\dag B + B^T \bar A = 0$. 
\end{proof}

\section{Equivalent expressions for the rebit encoding of a linear operator} \label{app:rebitEncodingLinear}

From \eq{eq:rebitEncodingLinear}, we find that the rebit encoding of a linear operator is given by
$\mathcal P(A) = \Re A \otimes I + \Im A \otimes XZ$. In this appendix, we derive alternative expressions for \eq{eq:rebitEncodingLinear}.
 
\begin{prop}\label{prop:rebit_encoding_linear_operator} The rebit encoding of a linear operator $A$ is equal to
\begin{equation}
\mathcal P(A) = \bar A \otimes \ketbra{\otimes}{\otimes}+  A \otimes \ketbra{\odot}{\odot}   = ( \bar A \otimes I)(
I \otimes \ketbra{\otimes}{\otimes} + A^T \otimes \ketbra{\odot}{\odot}) ,
\end{equation}
where $\ketbra\otimes\otimes = \frac 12(I+Y)$ and $\ketbra\odot\odot = \frac 12(I-Y)$ are the orthogonal projectors onto the $+1$ and $-1$ eigenspaces of the Pauli matrix $Y$, respectively.
\end{prop}
\begin{proof}
Substituting into \eq{eq:rebitEncodingLinear} the identities $\Re A = 1/2(A+\bar A)$, $\Im U = 1/(2i)(A-\bar A)$ and $XZ = -i Y$, we obtain
\begin{eqnarray*}
\mathcal P(A) &=& \frac 12\left[(A+\bar A)\otimes I - (A-\bar A)\otimes Y\right] \\
&=& \frac 12\left[A\otimes (I-Y) + \bar A\otimes (I+Y)\right] ,
\end{eqnarray*}
which is equal to $\bar A \otimes \ketbra{\otimes}{\otimes}+  A \otimes \ketbra{\odot}{\odot}$.
\end{proof}
In particular, if $A=U$ is unitary,
\begin{equation}\label{eq:controlledYencoding}
\mathcal{P}(U)=(\bar U\otimes I)\left(I\otimes \ketbra{\otimes}{\otimes}+  U^TU \otimes \ketbra{\odot}{\odot}\right).
\end{equation}

To compare with the rebit simulation of linear operators in \cite{aharonov2003simple}, we calculate the action of $\mathcal{P}$ on states written in the computational basis as follows:
\begin{prop}
\begin{equation} \label{eq:PUCoor}
\mathcal P(A): \sum_{ij} \psi_{ij} \ket{ij} \mapsto \sum_i [(\psi_{i0} \Re A - \psi_{i1} \Im A)\ket i] \ket 0 + \sum_i [(\psi_{i0} \Im A + \psi_{i1} \Re A)\ket i] \ket 1 .
\end{equation}
\end{prop}
\begin{proof}
\begin{equation} \label{eq:PUCoor1}
\mathcal P(A): \sum_{ij} \psi_{ij} \ket{ij} = \sum_{ij} \psi_{ij}[(\Re A\ket i)\ket j + (\Im A\ket i)\otimes XZ\ket j ].
\end{equation}
But $XZ\ket j = (-1)^j\ket{1- j}$. Plugging this into \eq{eq:PUCoor1} and expanding out the $j$ index, we obtain \eq{eq:PUCoor}.
\end{proof}

\noindent By setting $\psi_{i'j'}= \delta_{i'i}\delta_{j'j}$ in \eq{eq:PUCoor}, we obtain
\begin{eqnarray}
\mathcal P(A) \ket i \ket 0 &=& (\Re A \ket i)\ket 0 + (\Im A \ket i)\ket 1, \\
\mathcal P(A) \ket i \ket 1 &=& -(\Im A \ket i)\ket 0 + (\Re A \ket i)\ket 1 ,
\end{eqnarray}
which is equivalent to Definition 1 of \cite{aharonov2003simple}.

\section{Equivalence of norm definitions}\label{app:equivalenceNorm}

In this appendix, we provide a proof that on $L_n$ the operator norm for $\bbR$-linear operators coincides with that for linear operators.
\begin{prop}\label{prop:norms_coincide}
Let $A\in L_n$ be a linear operator on $n$-qubits. Then $\|A\|$ as defined by Definitions~\ref{defn:linear_norm} and \ref{defn:rlinear_norm} are the same.
\end{prop}
\begin{proof}
For clarity, denote $\|A\|_l$ as the operator norm for linear operators from Definition~\ref{defn:linear_norm} and $\|A\|_r$ as the operator norm for $\bbR$-linear operators from Definition~\ref{defn:rlinear_norm}. Because $\|A\|_l$ is the largest singular value of $A$ and $\|A\|_r$ is the largest singular value of $\mathcal{P}(A)$, we need only show these coincide. Say $A=UDV$ for unitaries $U,V\in U_n$ and diagonal matrix $D$ so that $\text{sing}(A)=\{D_{ii},\forall i\}$ are the singular values of $A$ and $\|A\|_l=\max\{|\lambda|:\lambda\in\text{sing}(A)\}$. Now, $\mathcal{P}(A)=\mathcal{P}(U)\mathcal{P}(D)\mathcal{P}(V)$, and the singular value decomposition of $\mathcal{P}(D)$ is easily calculated using Proposition~\ref{prop:rebit_encoding_linear_operator},
\begin{equation}
\mathcal{P}(D)=\bar D\otimes\ket{\otimes}\bra{\otimes}+D\otimes\ket{\odot}\bra{\odot}=(I\otimes SH)(\bar D\otimes\ket{0}\bra{0}+D\otimes\ket{1}\bra{1})(I\otimes HS^\dag).
\end{equation}
Thus, $\text{sing}(\mathcal{P}(A))=\text{sing}(\mathcal{P}(A))=\text{sing}(A)\cup\text{sing}(\bar A)$. Finally, $\|A\|_r=\|\mathcal{P}(A)\|_l=\max\{|\lambda|:\lambda\in\text{sing}(\mathcal{P}(A))\}=\|A\|_l$.
\end{proof}

\section{Alternative formulation of Theorem \ref{thm:OrthogonalABK}} \label{app:altFormulation}

In Theorem \ref{thm:OrthogonalABK}, we showed that for an $\bbR$-linear operator $A+BK$, the operator $\mathcal P(A+BK)$ is orthogonal if and only if $A^\dag A + B^T \bar B = I$ and $A^\dag B + B^T \bar A = 0$.

We now find an equivalent condition for orthogonality.
\begin{theorem} \label{thm:OrthogonalABKalt}
Let $A+BK$ be an $\mathbb R$-linear operator. Then $\mathcal P(A+BK)$ is orthogonal if and only if
\begin{eqnarray} \label{eq:OrthogonalABKalt}
A A^\dag + B B^\dag = I , \nn
A B^T + B A^T = 0 .
\end{eqnarray}
\end{theorem}

In the proof of Theorem \ref{thm:OrthogonalABK}, we used the property that a matrix $W$ is orthogonal if and only if $W^T W = I$. But this is equivalent to the condition that $W W^T = I$. Repeating the proof of Theorem \ref{thm:OrthogonalABK} using this condition yields \eq{eq:OrthogonalABKalt}. An alternative approach, which we use here, is to directly show that \eqref{eq:OrthogonalABK} is equivalent to \eqref{eq:OrthogonalABKalt}. 

\begin{proof}
\begin{eqnarray}
&& A^\dag A + B^T \bar B = I,\  A^\dag B + B^T \bar A = 0 \nn
&\iff & I = (A^\dag A + B^T \bar B) + (A^\dag B + B^T \bar A) K \nn
&& \quad\! = (A^\dag + B^T K) \star (A+BK) \nn
&\iff & I = (A+BK) \star  (A^\dag + B^T K) \nn
&& \quad\! = (A A^\dag + B B^\dag) + (BA^T + AB^T) K \nn
&\iff & A A^\dag + B B^\dag = I,\ A B^T + B A^T = 0 ,
\end{eqnarray}
where we used the star product $\star$ defined in \eq{eq:multiplstar}, and the fact that left inverses are equal to right inverses in a ring (See Appendix \ref{app:algebraicProperties}).
\end{proof}

\section{On orthogonal projections}
\label{app:OrthogonalProjections}

In this appendix, we recall some definitions about orthogonal projections. Let $\mathbb{H}\subseteq \mathcal H$ be a subspace of a vector space $\mathcal H$. The orthogonal complement of $\mathbb{H}$ is the set $\mathbb{H}^\perp = \set{u \in \mathcal H| \langle v,u\rangle = 0 \ \forall v\in \mathbb{H}}$. For finite-dimensional vector spaces, $(\mathbb{H}^\perp)^\perp = \mathbb{H}$, and
$\mathbb{H}$ and $\mathbb{H}^\perp$ are complementary subspaces, i.e.\ $\mathbb{H} \cap \mathbb{H}^\perp = \set 0$ and $\mathbb{H}\oplus\mathbb{H}^\perp = \mathcal H$, where $\oplus$ denotes direct sum. Moreover, for any $v \in \mathcal H$, there exists a unique $a \in \mathbb{H}$ and a unique $b \in \mathbb{H}^\perp$ such that $v=a+b$. The map $v\mapsto a$ is called the orthogonal projection onto $\mathbb{H}$, and we shall denote it by $\proj_\mathbb{H}(\cdot)$. It then follows that the map $v\mapsto b$ is equal to $\proj_{\mathbb{H}^\perp}(\cdot)$. Two immediate consequences are that $\proj_\mathbb{H}+\proj_{\mathbb{H}^\perp} = I$ and that $\proj_\mathbb{H} \circ\proj_{\mathbb{H}^\perp}=\proj_{\mathbb{H}^\perp} \circ\proj_\mathbb{H} =0$. It is also easy to verify that orthogonal projections are linear operators that are idempotent and hermitian, i.e. $\proj_\mathbb{H} = \proj_\mathbb{H}^2 = \proj_\mathbb{H}^\dag$.

An alternate characterization of orthogonal projections is as follows:
\begin{prop} \label{prop:orthogonalprojection}
Let $\mathbb{H}\subseteq \mathcal H$ be a subspace of a vector space $\mathcal H$. Let $\set{\ket{a_i}}_{i=1}^s$ be an orthonormal basis for $S$. Then $P$ is an orthogonal projection onto $\mathbb{H}$ if and only if 
\begin{equation} \label{eq:orthogonalprojection}
P=\sum_{i=1}^s \ketbra{a_i}{a_i}.
\end{equation}
\end{prop}
Note that Proposition \ref{prop:orthogonalprojection} implies that the formula in \eq{eq:orthogonalprojection} is in fact independent of the basis chosen for $\mathbb{H}$, i.e. if $\set{\ket{a_i}}_{i=1}^s$ and $\set{\ket{b_i}}_{i=1}^s$ are two bases for $\mathbb{H}$, then $\sum_{i=1}^s \ketbra{a_i}{a_i} = \sum_{i=1}^s \ketbra{b_i}{b_i}$.

\section{Matrix representation of $\bbR$-linear operators}\label{app:matrixrep}

In this appendix, we develop a matrix notation for the $\bbR$-linear operators. We define the \textit{matrix representation} of an $\bbR$-linear operator $A+BK$ to be the matrix $[A+BK]$ whose $(\mu,\nu)$th element is the column vector $\mat{A_{\mu\nu} \\B_{\mu\nu} } \in \bbC^2$. Each of these elements $\mat{A_{\mu\nu} \\B_{\mu\nu} }$  can be seen as belonging to the ring $R = \left\{\mat{a\\b}: a,b \in \mathbb C\right\}$, where addition is defined by $$\mat{a \\b }+\mat{c \\d } = \mat{a+c\\b+d}$$ and multiplication is defined by $$\mat{a \\b }\mat{c \\d } = \mat{ac + b\bar d \\ ad+ b\bar c}.$$
We shall sometimes also denote the column vector $\mat{a \\b}$ belonging to $R$ by $a+bk$, where $k$ is treated as a \textit{formal} symbol (for example, see its use in the proof of Theorem \ref{thm:partial_antiunitaries_not_group}). With this definition, the set of matrices $[A+BK]$ forms a matrix ring under the \textit{usual}\footnote{
Unlike usual matrix multiplication, the elements of the matrix belong to a ring, while the elements in the column vector representation of a vector belong to a field that is different from the ring. For a more rigorous treatment, we should treat the vector space $\mathbb C^n$ as a unital left $R$-module \cite{dummit2004abstract}.} rules of matrix addition and matrix multiplication.

Note that the matrix representation of an $\bbR$-linear operator is unique, as the following theorem, which expresses the matrix representation of an $\bbR$-linear operator $\Gamma$ in terms of $\Gamma$, shows 
\begin{theorem} \label{thm:matrixRep}
Let $\Gamma$ be an $\bbR$-linear operator. Then its matrix elements are given by\footnote{The notation $(A|B)$ refers to the augmented matrix formed from the matrices $A$ and $B$.}
\begin{equation} \label{eq:matrixPartialAntilinear}
\Gamma_{\mu\nu} = \frac 12 \left( \Gamma(e_\nu) - i \Gamma(i e_\nu) \right. \left| \Gamma(e_\nu) + i \Gamma(i e_\nu) \right)^T e_\mu,
\end{equation}
where $e_\mu$ is the $\mu$th basis vector defined by $e_\mu(\nu) = \delta_{\mu\nu}$.
\end{theorem}
\begin{proof} Since $\Gamma$ is an $\bbR$-linear operator, we can write $\Gamma = A+BK$, where $A$ and $B$ are both linear operators.
\begin{eqnarray*}
RHS &=& \frac 12 \left( \Gamma(e_\nu) - i \Gamma(i e_\nu) \right. \left| \Gamma(e_\nu) + i \Gamma(i e_\nu) \right)^T e_\mu  \\
&=& \frac 12 \left( \langle e_\mu, \Gamma(e_\nu) - i \Gamma(i e_\nu)\rangle , \langle e_\mu , \Gamma(e_\nu) + i \Gamma(i e_\nu) \right)^T \\
&=& \frac 12 ( \underbrace{\langle e_\mu , \Gamma(e_\nu) \rangle}_{(0)} - i \underbrace{\langle e_\mu , \Gamma(i e_\nu) \rangle}_{(1)} , \underbrace{\langle e_\mu , \Gamma(e_\nu) \rangle}_{(0)} + i \underbrace{\langle e_\mu , \Gamma(i e_\nu) \rangle}_{(1)})^T,
\end{eqnarray*}
where for $\alpha = 0, 1$,
\begin{eqnarray*}
(\alpha) &=& \langle e_\mu, \Gamma(i^\alpha e_\nu) \rangle \\
&=& \langle e_\mu, (A+BK)(i^\alpha e_\nu) \rangle \\
&=& \langle e_\mu, i^\alpha A(e_\nu) + (-1)^\alpha i^\alpha B(e_\nu) \rangle \\
&=& i^\alpha (A_{\mu\nu} + (-1)^\alpha  B_{\mu\nu} ).
\end{eqnarray*}
Hence, $(0) - i (1) = 2A_{\mu\nu}$ and $(0) + i (1) = 2B_{\mu\nu}$.

So,
\begin{eqnarray*}
RHS &=& \frac 12 (2 A_{\mu\nu}, 2 B_{\mu\nu})^T \\
&=& \mat{ A_{\mu\nu} \\ B_{\mu\nu} } \\
&=& (A+BK)_{\mu,\nu} = \Gamma_{\mu\nu}. 
\end{eqnarray*}
\end{proof}

Note that when $\Gamma$ is linear, the expression in \eq{eq:matrixPartialAntilinear} reduces to the following familiar expression.
\begin{eqnarray*}
\Gamma_{\mu\nu} &=& \frac 12 \left( \Gamma(e_\nu) + \Gamma(e_\nu) \right. \left| \Gamma(e_\nu) - \Gamma(e_\nu) \right)^T e_\mu \\
&=& (\Gamma(e_\nu)|0)^T e_\mu \\
&=& (\langle e_\mu , \Gamma(e_\nu) \rangle ,0)^T \\
&=& \langle e_\mu , \Gamma(e_\nu) \rangle + 0 K \\
&=& \langle e_\mu , \Gamma(e_\nu) \rangle.
\end{eqnarray*}

We will now illustrate the use of Theorem \ref{thm:matrixRep} in an example.
\begin{example}
Let $\Gamma: \bbC^2 \rightarrow \bbC^2$ be an $\bbR$-linear operator. Then 
$$ \Gamma: \mat{\alpha \\ \beta} \mapsto \mat{2 \Re\alpha + 2 \bar\beta \\ 3 \beta}$$
if and only if the matrix representation of $\Gamma$ is 
$$\mat{ 1+k & 2k \\ 0 & 3}.$$
\end{example}
\begin{proof}
The backward direction follows from matrix multiplication
$$ \mat{ 1+k & 2k \\ 0 & 3} \mat{\alpha \\ \beta} = \mat{ 1 & 0 \\ 0 & 3} \mat{\alpha \\ \beta} + \mat{ 1 & 2 \\ 0 & 0}\mat{\bar\alpha \\ \bar\beta} = \mat{\alpha + \bar\alpha + 2 \bar\beta \\ 3 \beta} = \mat{2 \Re\alpha + 2 \bar\beta \\ 3 \beta}$$
To obtain the forward direction, we make use of Theorem \ref{thm:matrixRep}:
\begin{eqnarray}
\Gamma_{\mu,0} &=& \frac 12 \left[ \mat{2 \\ 0} - i \mat{0 \\ 0} \right. \left| \mat{2 \\ 0} + i \mat{0 \\ 0} \right]^T e_\mu \\
 &=& \mat{1 & 0 \\ 1 & 0} e_\mu
\end{eqnarray}
and 
\begin{eqnarray}
\Gamma_{\mu,1} &=& \frac 12 \left[ \mat{2 \\3} - i \mat{-2 i \\ 3i} \right. \left|  \mat{2 \\3} + i \mat{-2 i \\ 3i} \right]^T e_\mu \\
 &=& \mat{0  & 3 \\ 2 & 0} e_\mu.
\end{eqnarray}
Hence, $\Gamma_{00} = \mat{1 \\ 1} = 1+K$, $\Gamma_{10} = \mat{0 \\ 0} = 0$,  $\Gamma_{01} = \mat{0 \\ 2} = 2K$ and $\Gamma_{11} = \mat{3 \\ 0} = 3$.

\end{proof}

\section{Proof of Lemma~\ref{lem:orthog_compiling}}\label{app:orthog_compiling}
Here we sketch the proof of Lemma~\ref{lem:orthog_compiling} following Chapter 4.5 of \cite{nielsen2010quantum}. The first thing to note is that an orthogonal gate $W$ on an $d$-dimensional system (i.e.~represented as a $d\times d$ matrix) can be broken down into a product of at most $d(d-1)/2$ ``two-level" orthogonal gates. A two-level orthogonal gate $V$ is one with non-zero off-diagonal entries in at most two rows and the corresponding two columns.

To show this, we write $W=(w_{ij})$ with matrix elements $w_{ij}$ for $i,j\in\{1,2,\dots,d\}$. We note that for any $i,j$ with $i>j$, we can find a two-level orthogonal gate $V$ such that $(VW)_{ij}=0$. This is done by choosing $V$ such that $V_{ab}=0$ except for (letting $N_{ij}=\sqrt{w_{jj}^2+w_{ji}^2}$)
\begin{eqnarray}
V_{jj}=w_{jj}/N_{ij},&\quad V_{ji}=w_{ij}/N_{ij}\\
V_{ij}=w_{ij}/N_{ij},&\quad V_{ii}=-w_{jj}/N_{ij}
\end{eqnarray}
and $V_{aa}=1$ whenever $a\not\in\{i,j\}$. One can check $V$ is orthogonal and it is clearly two-level. Calculating the $(i,j)$ element of $VW$ we get
\begin{equation}
(VW)_{ij}=\sum_h V_{ih}W_{hj}=V_{ii}W_{ij}+V_{ij}W_{jj}=(-w_{jj}w_{ij}+w_{ij}w_{jj})/N_{ij}=0.
\end{equation}
Incidentally, the $(j,j)$ element of $VW$ is unity
\begin{equation}
(VW)_{jj}=\sum_h V_{jh}W_{hj}=V_{ji}W_{ij}+V_{jj}W_{jj}=(w_{ji}^2+w_{jj}^2)/N_{ij}=1.
\end{equation}
Meanwhile, other elements in column $j$ are unchanged, $(VW)_{aj}=W_{aj}$ for $a\not\in\{i,j\}$. Also, if $W_{ib}=W_{jb}=0$ for $b<j<i$, then $(VW)_{ab}=W_{ab}$ for any $a$ as well.

Repeating this reduction with two-level orthogonal gates $V_1,V_2,\dots,V_{d-1}$, we can zero all off-diagonal elements in the first column of $V_{d-1}V_{d-2}\dots V_1W$. Since this product is still orthogonal, all off-diagonal elements of the first row are also zeroed. Now the same process can be repeated on the remaining block matrix in rows and columns 2 through $d$, and so on until only a $2\times 2$ block in the lower right remains. The inverse of this remaining matrix is two-level. In summary, we obtain two-level orthogonals $V_j$ such that
\begin{equation}
V_kV_{k-1}\dots V_1W=I
\end{equation}
and thus $W=V_1^TV_2^T\dots V_k^T$ is a product of two-level orthogonals. Also, $k$ is at most $(d-1)+(d-2)+\dots+1=d(d-1)/2$.

In the remainder of the proof we need to write two-level orthogonals as a product of $C^hZ$ gates and arbitrary single-qubit orthogonal gates. We note that $C^hX$ (which equals $C^hZ$ up to Hadamards, which are orthogonal, on the target) can swap two $n$-qubit basis states $\ket{x}$ and $\ket{y}$ when they differ in at most one place. Let $x_j\oplus y_j=1$ but $x_i\oplus y_i=0$ for $i\neq j$. Then, letting $C$ be the $C^{n-1}X$ gate with target $j$ and controls on all the remaining $n-1$ qubits, we find
\begin{eqnarray}
\left(\bigotimes_{i=1}^nX^{1-x_i}\right)C\left(\bigotimes_{i=1}^nX^{1-x_i}\right):&&\ket{x}\mapsto\ket{y}\\
&&\ket{y}\mapsto\ket{x}\\
&&\ket{z}\mapsto\ket{z},\quad z\neq x,y.
\end{eqnarray}
Say we have a two-level orthogonal $V$ acting non-trivially on $\ket{x_1}$ and $\ket{x_2}$. On those two states $V$ applies an orthogonal operator, call it $\tilde V$. First, we can swap $\ket{x_1}$ with other basis states in sequence $y_1,y_2,\dots,y_m$, where subsequent states differ in at most one bit, so that $y_m$ differs in one bit, bit $j$, from $x_2$. Now, applying $C^{n-1}\tilde S$ with target $j$ and controls the remaining $n-1$ qubits and reversing the swaps from the first step, we implement the two-level orthogonal gate $V$.

The last step is breaking down $C^{n-1}\tilde V$ into $C^{h}X$ gates and single-qubit rotations. Since $\tilde V$ is orthogonal and single-qubit, we can write without loss of generality
\begin{equation}\label{eq:generic_1qubit_orthog}
\tilde V=\cos\theta-i\sin\theta Y.
\end{equation}
Define $R=\cos(\theta/2)-i\sin(\theta/2) Y$, such that $RR^\dag=I$ and $RXR^\dag X=R^2=\tilde V$. The circuit in Fig.~\ref{fig:CS} implements $C^{n-1}\tilde V$. 

\begin{figure}
\begin{equation*}
\Qcircuit @C=1em @R=1em {
& \qw & \ctrl 2 & \qw & \ctrl 2 & \qw \\
& \qw &   \ctrl 1   & \qw & \ctrl 1  & \qw  \\
& \gate R & \targ & \gate {R^\dag}  & \targ & \qw \\
}
\end{equation*}

\caption{\label{fig:CS} The circuit, built from only orthogonal gates, implementing a $C^{n-1}V$ gate with two controls (the top two qubits), corresponding to $n=3$. Further controls are added in the natural way, extending the $CZ$ and $CCX$ gates to include the new qubits as controls too.}
\end{figure}
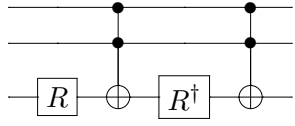

Finally, we note that this entire process breaks an arbitrary $n$-qubit orthogonal gate into at most $(n+1)2^{n-1}(2^n-1)$ $C^{h}Z$ gates, giving an upper bound on the number of $K_L$ gates that may need to be part of the $\bbR$-linear compilation in Eq.~\eqref{eq:compile_rlinear}.

\section{Pauli sets with different allowed phases}\label{app:p_sets_phases}

Our definition of the standard Pauli group $\cC_n(1)$ in \eq{eq:def_Paulis} differs from the Pauli group $G_n(1)$ in \eq{eq:def_Paulis_ikemike} in terms of the \textit{allowed phases} of the Pauli operators: while the Pauli operators in $\cC_n(1)$ are allowed to have arbitrary phases, those in $G_n(1)$ have phases which are constrained to be powers of $i$. In this appendix, we show that choosing different allowed phases has little effect on the definition of higher levels of the Clifford hierarchy.

We start by introducing some definitions. Let
\begin{equation}
P_n = \big\{I,X,Y,Z\}^{\otimes n} = \{p_1\otimes p_2 \otimes \ldots \otimes p_n: p_j \in \{I,X,Y,Z\} \mbox{ for all } j \in 1,2,\ldots, n \big\}
\end{equation}
denote the set of Pauli operators without any phases. 

\begin{definition}
\label{def:pauliSet}
Let $\Omega \subseteq \bbR$. The \textit{first level of the} $\Omega$-\textit{Clifford hierarchy} (called the $\Omega$-\textit{Pauli set}) is
\begin{equation} \label{eq:omegaPauliset}
\cC_n(1;\Omega) = \{e^{i\alpha} p: p \in P_n, \alpha \in \Omega \}.
\end{equation}
For $k\geq 2$, the $k$\textit{th level of the} $\Omega$\textit{-Clifford hierarchy} is
\begin{equation}
\cC_n(k;\Omega) = \{U \in U_n : U \cC_n(1;\Omega) U^\dag \subseteq \cC_n(k-1;\Omega) \}.
\end{equation}
\end{definition}

Note that \eq{eq:omegaPauliset} generalizes the Pauli groups discussed above: $\cC_n(1) = \cC_n(1;\bbR)$ and $G_n(1) = \cC_n(1;\tfrac \pi 2 \bbZ)$. Also, $P_n = \cC_n(1;2\pi \bbZ)$.
Note that $\cC_n(1;\Omega)$ is in general not a group, and hence we refer to it as a \textit{Pauli set} and not a \textit{Pauli group}. 

An immediate consequence of the definition of the $\Omega$-Clifford hierarchy is that it is closed under multiplication by global phases:
\begin{theorem}
\label{thm:closureGlobalPhasesClifford}
Let $k\geq 2$. If $\theta \in \bbR$ and $V \in \cC_n(k;\Omega)$, then $e^{i\theta} V \in \cC_n(k;\Omega)$.
\end{theorem}

For $k \geq 2$, how does the set $\cC_n(k;\Omega)$ depend on $\Omega$? In the rest of this appendix, we will show that as long as $\Omega$ is nonempty and $\pi$-\textit{periodic} (to be defined next),  
the set $\cC_n(k;\Omega)$ is in fact independent of $\Omega$ and is equal to the standard Clifford hierarchy $\cC_n(k)$.

\begin{definition}
A subset $\Omega \subseteq \bbR$ is $\pi$-\textit{periodic} if for all $\alpha \in \Omega$, 
\begin{equation}
\alpha \in \Omega \iff \alpha + \pi \in \Omega.
\end{equation}
\end{definition}

Note that while the subsets $\bbR$ and $\tfrac \pi 2$ are $\pi$-periodic, the subset $2\pi \bbZ$ is not.

\begin{theorem}
\label{thm:cliffordhierarchyindependence}
Let $\Omega_1, \Omega_2 \subseteq \bbR$ be nonempty $\pi$-periodic sets. Then for all $k\geq 2$,
\begin{equation} \label{eq:CnOmegaEquality}
\cC_n(k;\Omega_1) = \cC_n(k;\Omega_2).
\end{equation}
\end{theorem}
\begin{proof}
To prove \eq{eq:CnOmegaEquality}, it suffices to prove that for all $k\geq 2$, \begin{equation}
\label{eq:sufficesC12}
\cC_n(k;\Omega_1) \subseteq \cC_n(k;\Omega_2),
\end{equation}
since, by symmetry, interchanging the roles of $\Omega_1$ and $\Omega_2$ will give the opposite inclusion. 
We shall proceed by induction on $k$. 

We start with the base case $k=2$. Let $V \in \cC_n(2;\Omega_1)$. Then
\begin{equation}
\label{eq:VCCVdag}
V \cC_n(1;\Omega_1) V^\dag \subseteq \cC_n(1;\Omega_1).
\end{equation}

Our goal is to show that
\begin{equation}
\label{eq:VCCVdag2}
V \cC_n(1;\Omega_2) V^\dag \subseteq \cC_n(1;\Omega_2).
\end{equation}

To this end, pick an arbitrary element $e^{i\alpha} p \in \cC_n(1;\Omega_2)$, where $\alpha \in \Omega_2$ and $p \in P_n$. 
Since $\Omega_1$ is nonempty, there exists $\beta \in \Omega_1$ such that $e^{i\beta}p \in \cC_n(1;\Omega_1)$. By \eq{eq:VCCVdag}, there exists $\gamma \in \Omega_1$ and $q \in P_n$ such that 
\begin{equation}
\label{eq:Veibeta}
V(e^{i\beta} p) V^\dag = e^{i\gamma} q ,
\end{equation}
which implies that 
\begin{equation}
e^{i(\beta-\gamma)} VpV^\dag = q.
\end{equation}
Taking squares on both sides, and using the property that $p^2 = q^2 = 1$, we get
\begin{equation}
e^{2i(\beta-\gamma)} = 1,
\end{equation}
which implies that 
\begin{equation}
e^{i \gamma} = \pm e^{i \beta}.
\end{equation}
Substituting this into \eq{eq:Veibeta}, we get
\begin{equation}
VpV^\dag = \pm q.
\end{equation}
Therefore,
\begin{eqnarray}
V(e^{i\alpha} p) V^\dag &=& \pm e^{i\alpha} q \nn
&=& e^{i\alpha} q \mbox{ or } e^{i(\alpha+\pi)} q \nn
& \in & \cC_n(1;\Omega_2) ,
\end{eqnarray}
where in the last step, we used the $\pi$-periodicity of $\Omega_2$ to obtain the inclusion $\alpha+\pi \in \Omega_2$.

Since $e^{i\alpha}p \in \cC_n(1;\Omega_2)$ was chosen arbitrarily, we obtain \eq{eq:VCCVdag2}, which shows that $V \in \cC_n(2;\Omega_2)$. Hence, $\cC_n(2;\Omega_1) \subseteq \cC_n(2;\Omega_2)$.

Next, we prove the inductive step. Assume that \eq{eq:sufficesC12} holds for $k$. Let $V \in \cC_n(k+1;\Omega_1)$. Pick an arbitrary element $e^{i\alpha} p \in \cC_n(1;\Omega_2)$, where $\alpha \in \Omega_2$ and $p \in P_n$. As above, there exists $\beta \in \Omega_1$ such that $e^{i\beta}p \in \cC_n(1;\Omega_1)$.
Then,
\begin{eqnarray}
V(e^{i\alpha}p)V^\dag &=&
e^{i(\alpha-\beta)}V(e^{i\beta}p)V^\dag \nn
&\in & e^{i(\alpha-\beta)} \cC_n(k;\Omega_1) \nn
&=& \cC_n(k;\Omega_1), \quad\mbox{by Theorem \ref{thm:closureGlobalPhasesClifford}} \nn
&\subseteq & \cC_n(k;\Omega_2), \quad\mbox{by induction hypothesis.}
\end{eqnarray}
Hence, $V \in \cC_n(k+1;\Omega_2)$. This implies that $\cC_n(k+1;\Omega_1) \subseteq \cC_n(k+1;\Omega_2)$, which completes the proof.
\end{proof}

\begin{corollary}
\label{cor:cliffordhierarchyindt}
Let $\Omega \subseteq \bbR$ be a nonempty $\pi$-periodic set. Then for all $k\geq 2$, \begin{equation}
\cC_n(k;\Omega) = \cC_n(k) = G_n(k),
\end{equation}
i.e. $\cC_n(k;\Omega)$ is independent of $\Omega$.
\end{corollary}
\begin{proof}
Since $\cC_n(1) = \cC_n(1;\bbR)$ and $G_n(1) = \cC_n(1;\tfrac \pi 2 \bbZ)$, and $\bbR$ and $\tfrac \pi 2 \bbZ$ are both nonempty and $\pi$-periodic, Corollary \ref{cor:cliffordhierarchyindt} follows directly from Theorem \ref{thm:cliffordhierarchyindependence}.
\end{proof}

\end{document}